\documentclass[aps,pra,twocolumn,amssymb,amsmath,amsfonts,showkeys,nofootinbib,floatfix,superscriptaddress,longbibliography]{revtex4-2}

\usepackage[english]{babel}
\usepackage{amsthm}
\usepackage[toc,page]{appendix}
\usepackage[colorlinks=true,citecolor=blue,linkcolor=magenta]{hyperref}
\usepackage{bbm}
\usepackage{graphicx}
\usepackage{stackengine,xcolor}
\usepackage{tcolorbox}
\usepackage{nicematrix}
\usepackage{cleveref}
\usepackage{mathtools}
\usepackage{quantikz}
\usepackage{adjustbox}
\usepackage{mathrsfs}  
\usepackage{ytableau}
\newcommand{\Tr}{{\rm Tr}}
\newcommand{\Span}{{\rm span}}

\newcommand{\End}{\mathrm{End}}

\newcommand{\Alg}{\mathrm{Alg}}
\newcommand{\GT}{\mathrm{GT}}
\newcommand{\Cat}{\mathrm{Cat}}

\newcommand{\g}{\mathfrak{g}}
\newcommand{\su}{\mathfrak{su}}
\newcommand{\so}{\mathfrak{so}}

\newcommand{\SO}{\mathbb{SO}}
\newcommand{\SU}{\mathbb{SU}}
\newcommand{\U}{\mathbb{U}}
\newcommand{\Spin}{\mathbb{SPIN}}

\newcommand{\s}{\vec{\sigma}}

\renewcommand{\vec}[1]{\boldsymbol{#1}}

\newcommand{\id}{\openone}
\newcommand{\ad}{^{\dagger}}
\newcommand{\ketbra}[2]{|#1\rangle\!\langle #2|}

\newcommand{\CC}{\mathcal{C}}

\newcommand{\FC}{\mathcal{F}}

\newcommand{\HC}{\mathcal{H}}

\newcommand{\OC}{\mathcal{O}}
\newcommand{\PC}{\mathcal{P}}

\newcommand{\CBB}{\mathbb{C}}
\newcommand{\EBB}{\mathbb{E}}

\newcommand{\UBB}{\mathbb{U}}

\newcommand{\SPBB}{\mathbb{SP}}

\newcommand{\sydiagram}[1]{%
  {\ytableausetup{boxsize=0.35em}%
   \ydiagram{#1}}%
}

\newtheorem{theorem}{Theorem}
\newtheorem{remark}{Remark}

\newtheorem{lemma}{Lemma}

\newtheorem{supproposition}{Supplemental Proposition}
\newtheorem{supdefinition}{Supplemental Definition}

\begin{document}

\title{The commutant of fermionic Gaussian unitaries}

\author{Paolo Braccia}
\affiliation{Theoretical Division, Los Alamos National Laboratory, Los Alamos, New Mexico 87545, USA}

\author{N. L. Diaz}
\affiliation{Information Sciences, Los Alamos National Laboratory, Los Alamos, New Mexico 87545, USA}
\affiliation{Center for Non-Linear Studies, Los Alamos National Laboratory, Los Alamos, NM 87545, USA}

\author{Martin Larocca}
\affiliation{Theoretical Division, Los Alamos National Laboratory, Los Alamos, New Mexico 87545, USA}

\author{M. Cerezo}
\thanks{cerezo@lanl.gov}
\affiliation{Information Sciences, Los Alamos National Laboratory, Los Alamos, New Mexico 87545, USA}

\author{Diego Garc\'ia-Mart\'in}
\affiliation{Department for Quantum Information and Computation at Kepler (QUICK),\\ Johannes Kepler University, Linz, Austria }

\begin{abstract}

In this work, we characterize the $t$-th order commutants of fermionic Gaussian unitaries and of their particle-preserving subgroup acting on $n$ fermionic modes. These commutants govern Haar averages over the corresponding groups and therefore play a central role in fermionic randomized protocols, invariant theory, and resource quantification. Using Howe dualities, we show that the particle-preserving commutant is generated by generalized copy-hopping operators, while that for general Gaussian commutant is generated by generalized quadratic Majorana bilinears together with parity. We then derive closed formulas for the dimensions of both commutants as functions of $t$ and $n$, and develop constructive Gelfand--Tsetlin procedures to obtain explicit orthonormal bases, with detailed low-$t$ examples. Our framework also clarifies the structure of replicated fermionic states and connects naturally to measures of fermionic correlations, generalized Pl\"ucker-type constraints, and the stabilizer entropy of fermionic Gaussian states. These results provide a unified algebraic description of higher-order invariants for fermionic Gaussian dynamics.
\end{abstract}

\maketitle

\section{Introduction}

The theory of quantum information and computation is intimately related to the mathematical disciplines of group and representation theory~\cite{hayashi2017group, nielsen2000quantum}. This follows from the simple observation that the evolution of closed quantum systems whose dynamics are governed by the Schr\"odinger equation is described by unitary matrices acting on a Hilbert space $\HC$. Since the set of all such unitaries (equipped with the matrix product) is a group, the aforementioned connection is readily established, allowing one to import powerful tools to study quantum systems. Moreover,  whenever there exist symmetries that constrain the dynamics, the admissible evolutions respecting said symmetries also form a group. Specifically, a subgroup of the full unitary group. 

 This realization becomes particularly relevant whenever one is interested in computing average properties of a quantum system over a group of unitaries. In such cases, one needs to compute an integral over the uniform Haar measure associated to the group, which effectively is a projection onto the group's $t$-th order commutant. That is, the algebra of matrices $M$ such that $\left[M,U^{\otimes t}\right]=0$ for all unitaries $U$ in the group. Therefore, it is clear that characterizing a group's commutant is paramount to extracting average properties over the group, which explains the ubiquitousness of the commutant in quantum information sciences~\cite{mele2023introduction}.

 For the case of the unitary group, its $t$-th order commutant is well understood, and it is spanned by the permutation matrices that permute the different copies of $\HC$ in $\HC^{\otimes t}$, as implied by the famous Schur-Weyl duality~\cite{fulton1991representation}. This result has found a wealth of applications in quantum information and quantum computation, ranging from classical shadows~\cite{huang2020predicting}, randomized benchmarking~\cite{elben2022randomized}, the construction of ensembles of unitaries that reproduce the first $t$ moments of the Haar measure (i.e., $t$-designs)~\cite{schuster2024random}, quantum error correction~\cite{bennett1996mixed,calderbank1997quantum}, the resource theory of non-stabilizerness (or magic)~\cite{veitch2014resource,howard2017application}, the theory of barren plateaus~\cite{mcclean2018barren,fontana2023theadjoint,ragone2023unified,diaz2023showcasing,larocca2024review} in quantum machine learning~\cite{cerezo2020variationalreview,chang2025primer}, etc. In essence, any setting involving randomness or averages over unitaries is subject to the pervasive nature of the commutant. 

 In addition, the $t$-th order commutants of some subgroups of the unitary group are also known, for instance those of the classical Lie groups, i.e., the orthogonal~\cite{hashagen2018real,garcia2023deep,west2024real} and unitary symplectic groups~\cite{garcia2024architectures,west2024random}, whose commutants are representations of Brauer algebras~\cite{collins2006integration}. Moreover, the commutant of the Clifford group has also been studied, and a complete characterization of it has recently been found~\cite{gross2021schur,bittel2025complete}. This group is so important in quantum computation that the characterization of its commutant has led to discoveries in quantum property testing, the resource theory of magic, or quantum state $t$-designs, to name just a few.

In this work, we focus on the group of fermionic Gaussian unitaries, including its particle-preserving (PP) subgroup. In particular, we study their $t$-th order commutants, that we respectively denote $\CC_{t,n}$ and $\CC^{\rm PP}_{t,n}$. These unitaries naturally find applications in the simulation of fermionic systems, and also in the verification and benchmarking of quantum computers. Nevertheless, despite its widespread use, a complete characterization of their commutants in terms of operators acting on $\HC^{\otimes t}$ is still missing. Explicit bases for the cases $t=1,2,3$ are known~\cite{diaz2023showcasing,wan2022matchgate} for the non-particle-preserving case, but to the best of our knowledge higher values of $t$ remain unexplored. The PP subgroup has received even less attention, and its higher-order commutants remain largely uncharacterized.
Nonetheless, mathematicians have found answers~\cite{hasegawa1989spin,wenzl2020dualities,aboumrad2022skew} that we here translate to the language and concrete representation provided by fermionic operators. Specifically, we first find generating sets for $\CC^{\rm PP}_{t,n}$ and $\CC_{t,n}$.

The manuscript is organized as follows. In Sec.~\ref{sec:background} we review the necessary background on fermionic Gaussian unitaries, their particle-preserving subgroup, $t$-th order commutants, Howe dualities, and the Gelfand-Tsetlin method. In Sec.~\ref{sec:PP} we characterize the commutant of particle-preserving fermionic Gaussian unitaries, presenting its generators, a closed formula for its dimension, and explicit basis constructions via Gelfand--Tsetlin techniques, with detailed examples for low $t$. In Sec.~\ref{sec:general} we turn to general fermionic Gaussian unitaries and derive the analogous results for their commutant, including generators, dimension formulas, and explicit constructions for the cases $t=2,3$. In Sec.~\ref{sec:applications} we discuss applications, connecting our results to   quantifiers of fermionic correlations, generalized Pl\"ucker-type conditions, and the stabilizer entropy of fermionic Gaussian states. We provide an outlook in  Sec.~\ref{sec:outlook}, while technical proofs and additional constructions are included in the Appendices.

\section{ Background}\label{sec:background}

\subsection{Fermionic Gaussian unitaries}

In the second quantization picture, the Fock space associated to a quantum system of $n$ fermionic modes is given by
\begin{equation}
    \FC_n \cong \Lambda(\mathbb{C}^n)= \bigoplus_{r=0}^n \Lambda^r\left( \mathbb{C}^n \right)\,,
\end{equation}
where $\Lambda^r$ denotes the exterior or wedge product that provides the necessary anti-symmetrization to reproduce fermionic statistics. Under the Jordan--Wigner transformation, one can map $n$ fermionic modes to $n$ qubits, which allows us to identify the Fock space with the  usual $n$-qubit Hilbert space $\HC = \left(\CBB^2\right)^{\otimes n}$, i.e.,
\begin{equation}
    \FC_n \cong \HC\,.
\end{equation}
In particular, occupation-basis or Fock states are identified with computational-basis states.

In this context, fermionic Gaussian unitaries are defined as those preserving the vector space spanned by Majorana operators under conjugation~\cite{jozsa2008matchgates,wan2022matchgate,de2013power,oszmaniec2022fermion,cherrat2023quantum, matos2022characterization}.  We recall that the Majorana operators can be considered the fermionic analogues of position and momenta in bosonic systems, as they satisfy the canonical anticommutation relations. Under the Jordan--Wigner transformation, Majoranas are mapped to Pauli strings as
\begin{equation}
\begin{split}
    c_1&=X\id\dots \id,\; c_3= ZX\id\dots \id, \;\;\dots, \; c_{2n-1}=Z\dots Z X\,, \nonumber\\
        c_2&=Y\id\dots \id,\; c_4= ZY\id\dots \id,\;\;\dots, \; \;\; c_{2n}\;\;\;=Z\dots Z Y\,,
\end{split}
\end{equation}
where $\id,X,Y,Z$ are the usual single-qubit Pauli matrices. One can readily verify that these operators indeed anticommute, i.e., $\left\{c_\mu,c_\nu\right\}=2\delta_{\mu\nu}$, meaning that they have the structure of a Clifford algebra. More specifically, we can characterize fermionic Gaussian unitary as being generated by Hamiltonians that are quadratic in the Majorana operators. That is, any fermionic Gaussian unitaries can be expressed as  $R(U)= e^{-it H}$ for some Hermitian operator $H=\frac{i}{2} \sum_{\mu,\nu=1}^{2n}  h_{\mu\nu} c_\mu c_\nu$, with $h$ a real anti-symmetric matrix. In particular, in the Jordan--Wigner representation, these unitaries can be identified with matchgate circuits~\cite{casas2025matchgate}.

From an abstract point of view, fermionic Gaussian unitaries are a representation $R$ of the Lie group $\mathbb{SPIN}(2n)$, which is the double cover of the special orthogonal group $\SO(2n)$. This implies that there exists a map $2\to 1$ from  $\mathbb{SPIN}(2n)\to\SO(2n)$. More precisely, it holds that the Majorana operators transform under conjugation by matchgate unitaries as
\begin{equation} \label{eq:Majorana-action}
    R(U) c_\mu R^\dagger(U) = \sum_{\nu=1}^{2n} U_{\mu\nu} c_\nu\,,
\end{equation}
 where $U\in\SO(2n)$. 

 Having identified the group underpinning fermionic Gaussian unitaries, we can study how its action decomposes $\HC$ into irreducible representations (irreps) of $\mathbb{SPIN}(2n)$. Here, one finds that the Hilbert space $\HC$ splits into even and odd parity irreps as 
 \begin{equation}
     \HC\cong\HC_+ \oplus \HC_-\,,
 \end{equation}
 where the fermionic parity operator is $\Gamma =(-i)^n c_1\cdots c_{2n}=Z^{\otimes n}$~\cite{zimboras2014dynamic,kazi2022landscape}. In particular, $\HC_\pm$ is spanned by all occupation-number basis states with an even/odd number of occupied modes.

\subsection{PP fermionic Gaussian unitaries}

Particle-preserving (PP) Gaussian unitaries are a subgroup of fermionic unitaries, which in addition to parity preserve the total particle number $N=\frac{1}{2} \sum_{i=1}^n (\id - Z_i)$. They correspond, in Eq.~\eqref{eq:Majorana-action}, to matrices that are both orthogonal and symplectic, i.e., that belong to a representation $R$\footnote{Note that in what follows we will use $R$ to denote the representation of general, and also of PP Gaussian unitaries,  with context indicating which one we are referring to.} of the group $\SO(2n)\cap \SPBB(2n, \mathbb{R})$, which is isomorphic to the unitary group $\UBB(n)$ (see e.g.,~\cite{mele2024efficient,braccia2025optimal}). 
In the PP setting, it is most convenient to introduce the creation and annihilation operators, which can be defined in terms of the Majoranas as
\begin{equation}
    a_p=\frac{1}{2}\left(c_{2p-1}+ic_{2p}\right)\,, \quad 
    a_p^\dagger=\frac{1}{2}\left(c_{2p-1}-ic_{2p}\right)\,.
\end{equation}
Then, any PP Gaussian unitary can be expressed as $R(U)=e^{-iHt}$ for $H=\sum_{p,q=1}^n h_{pq} a^\dagger_p a_q$, with $h$ a Hermitian matrix. 
Under conjugation by a PP fermionic Gaussian unitary, the operators $a_p$ and $a_p^\dagger$ transform as
\begin{align} 
R(U) a^\dagger_p R^\dagger(U) &= \sum_{q=1}^n U_{qp} a^\dagger_q \,,\nonumber \\
R(U) a_p R^\dagger(U) &= \sum_{q=1}^n U^*_{qp} a_q \,,\label{eq:creation-PPaction}
\end{align}
where $U \in \mathbb{U}(n)$ is unitary.

Under the action of PP fermionic Gaussian unitaries, the Hilbert space $\HC$ decomposes into fixed-particle-number irreps as
\begin{equation}
    \HC\cong\bigoplus_{r=0}^n \HC_r\,,
\end{equation}
where $\HC_r$ denotes the subspace with exactly $r$ fermions. In particular, under the
Jordan--Wigner transformation, $\HC_r$ is equivalently spanned by all
computational-basis states of Hamming weight $r$. In particular,
\begin{equation}
   \HC_r \cong \Lambda^r(\mathbb{C}^n)\,,\qquad  \dim(\HC_r)=\binom{n}{r} \,,
\end{equation}
so the restriction of the PP action to $\HC_r$ is the $r$-th exterior-power
representation of $\U(n)$. 

\subsection{$t$-th order commutants}\label{sec:tcommutatbackground}

The central objects of our study are the $t$-th order commutants for general and PP fermionic Gaussian unitaries, respectively denoted as  $\CC_{t,n}$ and $\CC^{\rm PP}_{t,n}$. For convenience, we recall that given $t$ copies of the Hilbert space, $\HC^{\otimes t}$, with associated space of linear operators $\End(\HC^{\otimes t})$, the $t$-th order commutant of a group's representation acting on $\HC$ is defined as the associative algebra of matrices $M$ that commute with the $t$-th fold tensor product of the representation. For instance, $\CC_{t,n}\subseteq\End(\HC^{\otimes t})$ and
\begin{equation}
    \CC_{t,n}\!=\!\left\{M\,|\, \left[M, R(U)^{\otimes t}\right]=0\,,\quad \forall R(U)\in\mathbb{SPIN}(2n)\right\}\,.
\end{equation}
Importantly, $\CC_{t,n}$ has the structure of an associative algebra, as any linear combination or product of matrices from the commutant also belongs to the commutant. 

While a general characterization of $\CC_{t,n}$ and $\CC^{\rm PP}_{t,n}$ has not been reached, for small values of $t$ some results are known. For instance, for $t=1$
one can readily find
\begin{equation}
    \CC_{1,n}={\rm span}_{\mathbb{C}}\{\id,\Gamma\}\,.
\end{equation}
This is the reason why the irreps $\HC_\pm$ are preserved by the representation of $\Spin(2n)$. For $t=2$, the commutant is of dimension $2(2n+1)$, and
\begin{equation}
    \CC_{2,n}={\rm span}_{\mathbb{C}}\left\{Q_k^0,Q_k^1\right\}\quad \text{with }\, k=0,1,\ldots,2n \,.
\end{equation}
Above, the basis of the commutant is given by~\cite{diaz2023showcasing}
\begin{align}\label{eq:Qt2}
    &Q_k^0 = \sum_{\mu_1<\cdots<\mu_k} c_{\mu_1}\cdots c_{\mu_k} \otimes  c_{\mu_1}\cdots c_{\mu_k}\,,\nonumber \\ &Q_k^1 =   \sum_{\mu_1<\cdots<\mu_k} \Gamma c_{\mu_1}\cdots c_{\mu_k}\otimes  \, c_{\mu_1}\cdots c_{\mu_k}\,,
\end{align}
where $Q_0^{\{0,1\}}\equiv \id$~\footnote{We will use the same symbol, $\id$, to denote the identity matrix of any size whenever the size of the matrix is clear from the context.}. 
Furthermore, the third-order commutant $\CC_{3,n}$, can be obtained from the characterization in~\cite{wan2022matchgate}, and is of dimension $2\binom{2n+3}{3}$. The situation is even bleaker for PP fermionic Gaussian unitaries. Here, one can show that  $\CC^{\rm PP}_{1,n}$ is of dimension $n+1$ and spanned by~\cite{oszmaniec2022fermion}
\begin{equation}\label{eq:comm-powers-N}
    \CC^{\rm PP}_{1,n}={\rm span}_{\mathbb{C}}\{\id,N,N^2,\ldots N^n\}\,.
\end{equation}
However,  larger values of $t$ remain largely unexplored to the best of our knowledge.

The main goal of this work is to obtain a rigorous characterization of $\CC_{t,n}$ and $\CC^{\rm PP}_{t,n}$. In both cases, we will present a set of generators for the $t$-th order commutant, a closed formula for its dimension, and provide constructive recipes for obtaining an orthonormal basis based on the Gelfand-Tsetlin method.

\subsection{Howe dualities}

When analyzing the commutant of a representation $R$ of a Lie group $G$ over a vector space $\HC$, one can leverage a standard result from representation theory known as Howe dualities~\cite{hasegawa1989spin,wenzl2020dualities,aboumrad2022skew}, which constitute a form of the famous Schur-Weyl duality, and  whose key ideas we now recall. First, we note that the isotypical decomposition of $\HC$ under the action of $R(G)$ is
\begin{equation}
\label{eq:isotypical_decomp_general}
    \HC=\bigoplus_{\lambda} V_\lambda^G\otimes\CBB^{m_\lambda}\,,
\end{equation}
where the subspaces $V_\lambda^G$ are irreps of $G$, indexed by some label $\lambda$, and $m_\lambda$ is their multiplicity. In a nutshell, Howe dualities  describe the case when the representation space $\HC$ carries the representation of two commuting group actions $G$ and $G'$, under which $\HC$ decomposes as
\begin{equation}
    \mathcal H =  \bigoplus_{\lambda} V_\lambda^{G}\otimes W_\lambda^{G'}\,,
    \label{eq:howe-general-decomposition}
\end{equation}
where $V_\lambda^{G}$ and $W_\lambda^{G'}$ are irreducible representations of $G$ and $G'$, respectively, paired in a multiplicity-free way~\cite{hasegawa1989spin}.
In other words, the existence of a dual commuting group action $G'$, in addition to the original $G$, makes it so that the latter's multiplicity spaces organize as irreps of the former, and vice versa.

This decomposition has an immediate consequence for the commutants. Indeed, once $\mathcal H$ is written as in Eq.~\eqref{eq:howe-general-decomposition}, the $t$-th order commutant $\CC_{t}$ is
\begin{equation}
    \CC_{t}  = \bigoplus_{\lambda} \End\left(W_\lambda^{G'}\right)\,.
    \label{eq:howe-general-commutant}
\end{equation}
Hence, by specifying the dual group action $G'$ associated to the action over $\HC$ of the group $G$, a Howe duality allows one to completely characterize the commutant of $G$ by the irreps of $G'$.

Since statements of this kind appear in several equivalent forms in the literature, it is useful to keep track of the level at which the duality is being used. One may then present alternative formulations at the group level, as a pair of commuting group actions $(G,G')$; at the Lie-algebra level, as the commuting pair $(\mathfrak g,\mathfrak g')$, for $G=e^{\mathfrak g}$ and $G'=e^{\mathfrak g'}$ which specify the dual infinitesimal actions; or at the associative-algebra level, where the commutant of one action is realized as the image of the universal enveloping algebra of the other~\cite{wenzl2020dualities,aboumrad2022skew}.

The group-level formulation is the most convenient one for the representation-theoretic decomposition of $\HC$ into irreps. On the other hand, when one wishes to generate the commutant by explicit operators, it is more natural to pass to the infinitesimal description and work with the dual Lie algebra. In most cases indeed, the resulting associative algebra is exactly the commutant, i.e.,
\begin{equation}
    \CC_{t}
    =
    \mathrm{Im}\,\mathscr U(\mathfrak g')\,,
    \label{eq:howe-general-envelope}
\end{equation}
where $\mathscr{U}$ is the universal enveloping algebra of $\g'$, spanned by products of the elements $\g'$ subject to the relations imposed by the Lie bracket  (see Appendix~\ref{ap:PP-generators} for a precise definition).
In other words, Eq.~\eqref{eq:howe-general-envelope} is simply the operator-algebraic counterpart of the block decomposition in Eq.~\eqref{eq:howe-general-commutant}. In the next sections, will choose between these viewpoints according to convenience.

\subsection{Gelfand-Tsetlin method}\label{sec:GT}

Knowing how a commutant decomposes as a direct sum of blocks carrying irreps of the group pair $(G,G')$ is not enough to fully characterize it. Indeed, for practical purposes, one would like to have access to an explicit basis of the commutant, which enables, among other things, the use of Weingarten calculus~\cite{mele2023introduction} to compute group averages. 

The Gelfand-Tsetlin (GT) method~\cite{collins2006integration,zhelobenko1973compact,molev2006gelfand} is a standard representation-theoretic procedure for constructing canonical bases of irreps. In its classical form, it applies to nested families of groups, or Lie algebras, for which successive restrictions are multiplicity-free~\cite{goodman2009symmetry}. The basic idea behind the GT construction is to fully resolve a symmetry-adapted basis inside an irrep by restricting along a canonical chain of smaller subgroups~\cite{molev2006gelfand}. Starting from a unique state which is completely specified by the irrep itself, known as the highest-weight, each reduction onto a subgroup refines it until we end up with a particular vector, and one finds that the set of vectors corresponding to all the possible ways of carrying out the restrictions form an orthonormal basis of the irrep. 
Here we present the key steps of this procedure.

Let us consider the case where the group of interest is $G'$, the reason will be evident in brief. Furthermore, let us write it as $G'_d$, for $d$ its size, so that we can think of $G'_d$ as a group family over $d$.
Then take an irrep $W_\lambda^{G'_d}$ for some fixed irrep label $\lambda$. For the GT method one considers a multiplicity-free~\cite{molev2006gelfand} nested subgroup chain
\begin{equation}
    G'_1\subset G'_2\subset \cdots \subset G'_d\,,
\end{equation}
or, equivalently, the restriction chain
\begin{equation}
    G'_d \downarrow G'_{d-1} \downarrow \cdots \downarrow G'_1\,.
\end{equation}
Under this assumption, when one restricts the irrep $W_\lambda^{G'_d}$ into the irreps of $G'_{d-1}$, it decomposes as a direct sum of pairwise non-isomorphic irreps,
\begin{equation}
    W_\lambda^{G'_d}\!\downarrow_{G'_{d-1}}
    \cong
    \bigoplus_{\lambda^{(d-1)}} W^{G'_{d-1}}_{\lambda^{(d-1)}}\,.
\end{equation}
The same procedure is then iterated through each successive step of the chain.

After iterating all the way and reaching $G'_1$, we get a sequence of intermediate irrep labels
\begin{equation}
    \lambda^{(d)}=\lambda,\,
    \lambda^{(d-1)},\,
    \ldots,\,
    \lambda^{(1)}\,,
\end{equation}
which records which irrep is selected at each restriction step. This full sequence is what is called a GT pattern, and we denote it by $T$. In turn, each admissible pattern $T$ labels one vector in a canonical basis of $W_\lambda$~\cite{molev2006gelfand},
\begin{equation}
    \bigl\{\ket{\lambda,T}\,:\, T\in \GT(\lambda)\bigr\},
\end{equation}
where $\GT(\lambda)$ denotes the set of GT patterns with initial irrep label $\lambda$.

Equivalently, using $\ket{\lambda,{\rm hw}}\in W_\lambda$, a highest-weight vector of the starting irrep, as the initial reference state, then each pattern $T$ determines a canonical ordered product $Z_T$ of lowering operators, such that
\begin{equation}
    \ket{\lambda,T}\propto Z_T\ket{\lambda,{\rm hw}}\,.
\end{equation}
After normalization, these vectors form an orthonormal basis of $W_\lambda$. In this way, the GT method refines the abstract highest-weight label $\lambda$, which only specifies the representation, into a complete set of basis vectors indexed by the finer GT path data $T$.

We can now readily see how this machinery is relevant to the commutant problem. Indeed, assume a Howe duality for the representation over some vector space $\HC$ of the group of interest $G$. Then, by the ensuing decomposition of Eq.~\eqref{eq:howe-general-decomposition} and the corresponding decomposition of the commutant of $G$ in Eq.~\eqref{eq:howe-general-commutant}, we immediately see that if $G'$ is such that the GT method can be applied to its irreps we can use it to construct a basis of $\End(W_\lambda^{G'})$ and hence of the commutant.
Indeed, once an orthonormal basis of each irreducible block $W^{G'}_\lambda$ has been constructed, an orthonormal basis of the full matrix algebra $\End(W^{G'}_\lambda)$ is obtained immediately from the projectors
\begin{equation}
    X^{(\lambda)}_{T,T'}
    =
    \ketbra{\lambda,T}{\lambda,T'}\,.
\end{equation}
Equivalently, if $P_\lambda^{\mathrm{hw}}$ denotes the projector onto the highest-weight of the irrep $\lambda$, one can write
\begin{equation}
\label{eq:X-ZPZ}
    X^{(\lambda)}_{T,T'}
    =
    Z_T\,P_\lambda^{\mathrm{hw}}\,Z_{T'}^\dagger\,,
\end{equation}
where the normalization is assumed.

While here we outlined a high-level description of the GT method, in Appendices~\ref{ap:PP_GT_method} and~\ref{ap:gauss_GT_method} we present more details about how the method is applied to the group of fermionic Gaussian unitaries and its particle-preserving subgroup.

\section{Commutant of PP fermionic Gaussian unitaries}\label{sec:PP}

\subsection{Generators of the commutant}

As our first contribution, we present a generating set for $\CC_{t,n}^{\rm PP}$. To this end, we begin by denoting the parity operator acting on the $j$-th copy of the Hilbert space within $\HC^{\otimes t}$ as 
\begin{equation}
    \Gamma_j = (-i)^n c_1^{(j)}\cdots c_{2n}^{(j)} = \left(Z^{\otimes n}\right)^{(j)}\,.
\end{equation}
Next, we define the ``dressed'' creation and annihilation operators
\begin{align}
\widetilde a^{(j)}_p= \Gamma_1 \otimes \Gamma_2 \otimes \cdots \otimes \Gamma_{j-1}
\otimes a^{(j)}_p \otimes \mathbbm{1} \otimes \mathbbm{1} \otimes \cdots\, ,\nonumber \\
\widetilde a^{\dagger (j)}_p = \Gamma_1 \otimes \Gamma_2 \otimes \cdots \otimes \Gamma_{j-1} \otimes a^{\dagger (j)}_p \otimes \mathbbm{1} \otimes \mathbbm{1} \otimes \cdots\,. \nonumber
\end{align}
These satisfy the canonical anticommutation relations
\begin{align} 
&\left\{\widetilde a^{(j)}_p,\, \widetilde a^{\dagger (k)}_q\right\} = \delta_{pq}\delta_{jk}\,, \nonumber \\  &\left\{\widetilde a^{(j)}_p,\,\widetilde a^{(k)}_q\right\} = \left\{\widetilde a^{\dagger (j)}_p,\,\widetilde a^{\dagger (k)}_q\right\} = 0\,.
\end{align}
That is, they are genuine fermionic operators both within each copy of $\HC$ and also across copies.

From the previous, we introduce the operators
\begin{equation} \label{eq:PP-gen}
\widetilde{\Omega}_{j,k} = \sum_{p=1}^n  \widetilde{a}^\dagger_p\,\!^{(j)}\, \widetilde{a}^{\,(k)}_p\,.
\end{equation}
Here we note that when $j=k$, then $\widetilde{\Omega}_{j,j}=\sum_{p=1}^n\widetilde{a}_p^{\dagger{(j)}} \widetilde{a}^{\,(j)}_p$ is simply the number operator associated to the $j$-th copy. Then, for $j\neq k$ we can use the fact that each term $\widetilde{a}^\dagger_p\,\!^{(j)} \widetilde{a}_p^{\,(k)}$ annihilates a fermion on mode $p$ at position $k$ and creates another on the same mode at position $j$, to see that $\widetilde{\Omega}_{j,k}$ simply applies this action over all modes. Thus, $\widetilde{\Omega}_{j,k}$ is a hopping operator for fermions on every mode from copy $j$ to copy $k$. Importantly, we can prove that the following two lemmas hold (see Appendix~\ref{ap:PP-generators}):
\begin{lemma}\label{lem:Omega-in-C}
    Let $\widetilde{\Omega}_{j,k}$ be an operator as defined in Eq.~\eqref{eq:PP-gen}. Then, 
    \begin{equation}
        \widetilde{\Omega}_{j,k}\in\CC^{\rm PP}_{t,n}\,,\quad \forall j,k\,.
    \end{equation}
\end{lemma}

\begin{lemma}\label{lem:Omega-rep-gl}
The operators $\widetilde\Omega_{j,k}$ satisfy the Lie algebra commutation relations
\begin{equation} \label{eq:Omega_algebra}
\left[\widetilde{\Omega}_{j,k}, \widetilde{\Omega}_{j',k'}\right] = \delta_{kj'}\widetilde\Omega_{j,k'} - \delta_{jk'}\widetilde\Omega_{j',k}\,,
\end{equation}
and therefore generate a representation of the $\mathfrak{u}(t)$ algebra.
\end{lemma}

 Hence, having realized that  the operators $\widetilde{\Omega}_{j,k}$ in $\CC^{\rm PP}_{t,n}$ realize a representation of  $\mathfrak{u}(t)$, we can leverage the skew Howe duality~\cite{aboumrad2022skew} to further characterize $\CC^{\rm PP}_{t,n}$. Specifically, under the identification 
\begin{equation}\label{eq:hilbert-fock-PP}
     \HC^{\otimes t}\cong \FC_n^{\otimes t}\cong\Lambda(\mathbb{C}^n\otimes \mathbb{C}^t)\,,
\end{equation}
the duality states that over $\FC_n^{\otimes t}$, the pair $(\mathbb{U}(n),\mathbb{U}(t))$ admits representations which are mutually commuting. Specifically, $\mathbb{U}(n)$ acts non-trivially on $\mathbb{C}^n$ through the PP fermionic Gaussian representation, whereas $\mathbb{U}(t)$ acts on $\mathbb{C}^t$  with the representation induced by the infinitesimal generators $\widetilde\Omega_{j,k}$. This result is central in proving the following theorem. 

\begin{theorem}\label{th:PP-generators}
    The commutant $\CC^{\rm PP}_{t,n}$, is generated as
    \begin{equation}\nonumber
\CC^{\rm PP}_{t,n} = {\rm span}_{\mathbb C} \Big\langle \widetilde\Omega_{j,j+1}, \widetilde\Omega_{j+1,j}, \widetilde\Omega_{1,1} \Big\rangle \,,\;\text{with } 1\leq j\leq t-1\,.
\end{equation}
\end{theorem}
The derivation of this theorem can be found in Appendix~\ref{ap:PP-generators}. Theorem~\ref{th:PP-generators} reveals the structure of the commutant as an algebra spanned by polynomials on the $\widetilde{\Omega}_{j,k}$, and generated by $\OC(t)$  operators. Of course, Theorem~\ref{th:PP-generators} does not provide an explicit orthonormal basis for $\CC_{t,n}^{\rm PP}$, but we will address this problem below.

\begin{figure}[t]
    \centering
    \includegraphics[width=\linewidth]{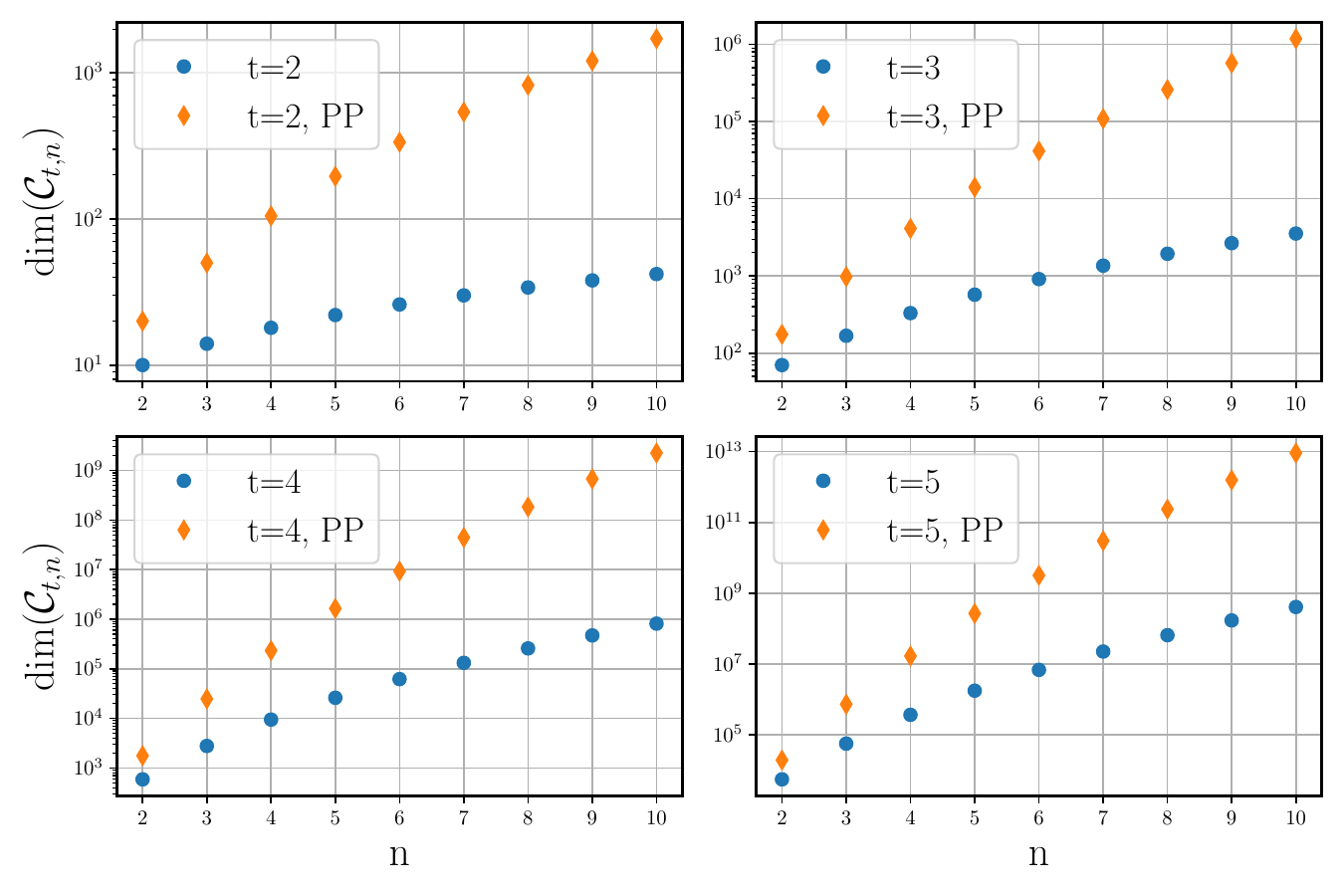}
    \caption{{\bf Dimensions of the commutants}. We plot in a logarithmic scale the dimensions of $\CC_{t,n}$ and $\CC_{t,n}^{\rm PP}$ as a function of the number of fermionic modes $n$, for values of $t=2,3,4,5$. }
    \label{fig:dimensions}
\end{figure}

\subsection{Dimensions of the commutant}

Next, we present an explicit formula for the dimension of $\CC^{\rm PP}_{t,n}$. We do so by resorting to representation-theoretic tools such as Weyl integration and Selberg-type integrals. The result is stated in the next theorem, and proven in Appendix~\ref{ap:PP-dim}.

\begin{theorem}
\label{th:PP-dimension}
    The dimension of $\CC^{\rm PP}_{t,n}$ is 
    \begin{equation}
        \dim\left(\CC^{\rm PP}_{t,n}\right)=\prod_{j=0}^{n-1}\frac{j!\,(j+2t)!}{(j+t)!^2}\,.
    \end{equation}
\end{theorem}

To better grasp the meaning of Theorem~\ref{th:PP-dimension}, it is instructive to study the behavior of $\dim\left(\CC^{\rm PP}_{t,n}\right)$ in different regimes. First, we can fix the number of copies $t$ and let the number of fermionic modes $n$ scale. In this regime, the dimension of the PP commutant behaves as $\dim\left(\CC^{\rm PP}_{t,n}\right)\in \Theta\left(n^{t^2}\right)$, as shown in Appendix~\ref{ap:scaling-dim-PP}. That is, polynomial in $n$ but with a polynomial degree that grows rapidly as $t^2$. This growth demonstrates that there exists a rich structure of invariant polynomials in the hopping operators $\widetilde{\Omega}_{j,k}$, even at moderate number of copies $t$. In Fig.~\ref{fig:dimensions}, we plot the dimension of $\CC_{t,n}$ for fixed $t=2-5$ and scaling $n$. There, we can observe that, even for small number of $n$ (i.e., $n\leq 10$) and $t$, the dimension of the PP commutant grows very quickly. Moreover, for fixed $n$ and scaling $t$, we find that $\dim\left(\CC^{\rm PP}_{t,n}\right)\in \Theta\left(4^{nt}\,t^{-n^2/2}\right)$, an exponential scaling with $t$.

\subsection{Constructing a basis for the commutant via GT}\label{sec:GTappliedtoPP}

As previously noted, the skew Howe duality states that the pair $(\mathbb{U}(n),\mathbb{U}(t))$ admits representations which are mutually commuting on $\FC_n^{\otimes t}=\Lambda(\mathbb{C}^n\otimes \mathbb{C}^t)$. Instantiating Eq.~\eqref{eq:howe-general-decomposition}  leads to  the decomposition 
\begin{equation}
    \FC_n^{\otimes t}\cong\bigoplus_{\lambda\subseteq n^t} V_{\lambda^T}^{\mathbb{U}(n)}\otimes V_{\lambda}^{\mathbb{U}(t)}\,.\label{eq:Schur-Weyl-passive-main}
\end{equation}
Above, the sum runs over partitions, or Young diagrams, $\lambda$ that fit inside an $n\times t$ rectangle, and we defined $\lambda^T$ as the partition obtained by transposing the Young diagram. We refer the reader to Appendix~\ref{ap:PP_GT_method} for additional details. The latter contains all the ingredients needed to implement the GT procedure (see the background Section~\ref{sec:GT}) for different values of $t$. We here present explicit results for $t=1,2$ and in the appendix present some considerations for $t=3$.

\subsubsection{$t=1$ case}

For $t=1$ the copy group is $\mathbb{U}(1)$ and we can only obtain partitions with at most one row, and row length at most $n$. That is,  $\lambda=(0),(1),\cdots ,(n)$ (with transpose given by $\lambda^T=(r)^T=(1^r)$). Since the group is abelian, all such irreps are one dimensional, meaning that
\begin{equation}
    \dim(\CC^{\rm PP}_{1,n})=n+1\,.
\end{equation}
Here, the GT restriction chain for $\mathbb{U}(1)$ is trivial, implying from Eq.~\eqref{eq:X-ZPZ} that the projectors into each one-dimensional irrep are a basis for the commutant. 

Indeed, defining the number operator
\begin{equation}
    N=\widetilde\Omega_{1,1}\,,
\end{equation}
we can construct the projectors onto the irreps as (see Appendix~\ref{ap:construction-GT-PP})
\begin{equation}
    P_r=\prod_{\substack{s= 0\\s\neq r}}^n\frac{N-s}{r-s}\,,
\end{equation}
whose action projects onto the $r$-particle subsector, i.e.,  $P_r:\Lambda(\mathbb{C}^n)\rightarrow\Lambda^r(\mathbb{C}^n)$, and which satisfy $P_rP_r'=\delta_{r,r'}P_r$ and $\sum_{r=0}^nP_r=\id$. As such, we have
\begin{equation}
    \CC_{1,n}^{\rm PP}={\rm span}_{\mathbb{C}}\{P_0,P_1,\cdots,P_n\}\,.
\end{equation}
Here, one may wonder how this basis relates to the one previously known in Eq.~\eqref{eq:comm-powers-N}. To relate these two alternative commutant bases, we begin by orthogonalizing the basis of Eq.~\eqref{eq:comm-powers-N}, which leads to  
\begin{equation}
    \CC_{1,n}^{\rm PP}={\rm span}_{\mathbb{C}}\{E_0,
E_1,\ldots,E_n\}\,,\end{equation} 
with $E_0 = \mathbb I$ and
\begin{equation}
E_a=\sum_{i_1<i_2<...<i_a} Z_{i_1}Z_{i_2}...Z_{i_a}\,.
\end{equation}
This basis is related to $N$ through Krawtchouk polynomials $E_a=K_a(n-N;n)$,
which in turn can be found to imply 
\begin{equation}\label{eq:passiveEPrelation}
   E_a=\sum_r K_a(n-r;n)P_r\,.
\end{equation}
See Appendix \ref{ap:construction-GT-PP} for details.

\subsubsection{$t=2$ case}\label{sec:PPt2}

Here, the partitions allowed for the copy side group $\mathbb{U}(2)$ are 
\begin{equation}
    \lambda=(\lambda_1,\lambda_2)\,,\quad n\geq\lambda_1\geq\lambda_2\geq 0\,,
\end{equation}
such that
\begin{equation}
\lambda^T=(2^{\lambda_2},1^{\lambda_1-\lambda_2})\,,
\end{equation}
meaning that the Fock space decomposes as
\begin{equation}
    \FC_n^{\otimes 2}\cong\bigoplus_{n\geq\lambda_1\geq\lambda_2\geq 0} V_{(2^{\lambda_2},1^{\lambda_1-\lambda_2})}^{\mathbb{U}(n)}\otimes V_{(\lambda_1,\lambda_2)}^{\mathbb{U}(2)}\,.
\end{equation}
Next, since
\begin{equation}\label{eq:dim-t-2-passive}
    \dim\left(V_{(\lambda_1,\lambda_2)}^{\mathbb{U}(2)}\right)=\lambda_1-\lambda_2+1\,,
\end{equation}
we find
\begin{align}
    \dim(\CC_{2,n}^{\rm PP})&=\sum_{n\geq\lambda_1\geq\lambda_2\geq 0}(\lambda_1-\lambda_2+1)^2\nonumber\\
    &=\frac{(n+1)(n+2)^2(n+3)}{12}\,.
\end{align}

Next, we note that the subgroup chain is
\begin{equation}
    \mathbb{U}(1)\subset\mathbb{U}(2)\,,
\end{equation}
and the GT patterns will be a triangular array 
\begin{equation}\label{eq:GT-pattern-passive-t2}
    \begin{array}{ccc}
       \lambda_1  & & \lambda_2  \\
         & m & 
    \end{array}
\end{equation}
with $\lambda_1\geq m\geq \lambda_2$. Notably, we show in Appendix~\ref{ap:construction-GT-PP} that the rows of the GT pattern indicates the particle sector to which the irrep belongs to, as well as the number of particles in each copy of the Hilbert space. Namely, the irrep belongs to the sector with $|\lambda|=\lambda_1+\lambda_2$ particles, where $m$ of them are in the first copy, and $\lambda_1+\lambda_2-m$ in the second one. In the same appendix, we prove that the following operators belong to $\CC_{2,n}^{\rm PP}$. From the zero particle sector we have one commutant element
\begin{equation}
    X^{(0,0)}_{0,0}=\ket{0}\bra{0}\,.
\end{equation}
From the one-particle sector there are four commutant elements
\begin{align}
    X^{(1,0)}_{T_1,T_1}\!=\!\sum_{p=1}^n{\widetilde{a}_p^{\dagger(1)}}\ket{0}\!\bra{0}\widetilde{a}_p^{(1)}\,,\,    X^{(1,0)}_{T_2,T_2}\!=\!\sum_{p=1}^n{\widetilde{a}_p^{\dagger(2)}}\ket{0}\!\bra{0}\widetilde{a}_p^{(2)},\nonumber\\ X^{(1,0)}_{T_1,T_2}\!=\!\sum_{p=1}^n{\widetilde{a}_p^{\dagger(1)}}\ket{0}\!\bra{0}\widetilde{a}_p^{(2)}\,,\,    X^{(1,0)}_{T_2,T_1}\!=\!\sum_{p=1}^n{\widetilde{a}_p^{\dagger(2)}}\ket{0}\!\bra{0}\widetilde{a}_p^{(1)}.\nonumber
\end{align}
Then, from the two-particle subsector there are 10 basis elements coming from two different irreps. First, the single element
\begin{equation}
  X^{(1,1)}_{T,T} =\sum_{1\leq p\leq q\leq n}\ket{p,q;+}\bra{p,q;+}\,,
\end{equation}
where for $1\leq p<q\leq n $ we defined
\begin{equation}
    \ket{p,q;+}=\frac{1}{\sqrt{2}}({\widetilde{a}_p^{\dagger(1)}}{\widetilde{a}_q^{\dagger(2)}}+{\widetilde{a}_q^{\dagger(1)}}{\widetilde{a}_p^{\dagger(2)}})\ket{0}\,,
\end{equation}
and for  $p=q$
\begin{equation}
    \ket{p,p;+}={\widetilde{a}_p^{\dagger(1)}}{\widetilde{a}_p^{\dagger(2)}}\ket{0}\,.
\end{equation}
And second, nine elements given by
\begin{align}
   X^{(2,0)}_{T_i,T_j}=\sum_{1\leq p<q\leq n}\ket{p,q;i}\bra{p,q;j}
\end{align}
for
\begin{equation}
\ket{p,q;i}=\begin{cases}
    {\widetilde{a}_p^{\dagger(1)}}{\widetilde{a}_q^{\dagger(1)}}\ket{0}\,,\quad p< q\,, \text{ if }i=2\,,\\ \\
    \frac{1}{\sqrt{2}}({\widetilde{a}_p^{\dagger(1)}}{\widetilde{a}_q^{\dagger(2)}}-{\widetilde{a}_q^{\dagger(1)}}{\widetilde{a}_p^{\dagger(2)}})\ket{0}\,, \text{ if }i=1\,,\\ \\
    {\widetilde{a}_p^{\dagger(2)}}{\widetilde{a}_q^{\dagger(2)}}\ket{0}\,,\quad p< q\,,\text{ if }i=0\,.
\end{cases}
\end{equation}

The examples above already make the general structure of the PP $t=2$ commutant basis apparent. For a fixed block  $(\lambda_1,\lambda_2) $, there will be an operator $X^{(\lambda)}_{T,T}=P_\lambda^{\mathrm{hw}}$, where the highest weight line is built from  $\lambda_1-\lambda_2 $ creation operators in copy  $1 $, together with  $\lambda_2 $ copy-singlet pairs, projected onto the physical  $\mathbb U(n) $-symmetry type  $(2^{\lambda_2},1^{\lambda_1-\lambda_2}) $. Then, the rest of the basis elements are obtained from an effective lowering operator that moves particles from the first copy onto the second (while still respecting the physical copy symmetry).

In Appendix~\ref{ap:construction-GT-PP} we fully flesh out this intuition, as well as provide a nice interpretation of the commutant elements in terms of spin representations (e.g., zero particles come from a spin-$0$ representation, one particle from a spin-$\frac{1}{2}$, and two particle from spin-$1$).
The basic idea is to make use of the operators 
$N=N_1+N_2$, $J_z=\frac{N_1-N_2}{2}$
and  identify $
J_+=\widetilde{\Omega}_{12}$, 
$J_-=\widetilde{\Omega}_{21}$ such that we can relate the GT index $m$ to a spin index $m_{\rm spin}=-j,-j+1\ldots,j$ by subtracting the central charge $N$, i.e., $m_{\rm spin}=m-\frac{N}{2}$.

\section{Commutant of fermionic Gaussian unitaries}\label{sec:general}

Next, we turn our attention to general fermionic Gaussian unitaries, and we study the corresponding  $t$-th order commutant $\CC_{t,n}$.

\subsection{Generators of the commutant}
As in the PP case, we begin by finding a generating set of the commutant $\CC_{t,n}$ for general fermionic Gaussian unitaries, by translating the abstract mathematical results in Ref.~\cite{hasegawa1989spin,wenzl2020dualities} to the full spinor Majorana representation on $\HC^{\otimes t}$. To this end, we define the operators
\begin{equation}
    \widetilde{c}_\mu^{\,(j)} = \Gamma_1 \otimes \Gamma_2 \otimes \cdots \otimes c_\mu^{(j)} \otimes \id \otimes \id \otimes \cdots \,,
\end{equation}
which we refer to as ``dressed'' Majorana operators. These operators satisfy the anticommutation relations
\begin{equation}\label{eq:Majoranaanticomm}
   \left \{ \widetilde{c}_\mu^{\,(j)}, \widetilde{c}_\nu^{\,(k)} \right\} = 2\delta_{\mu\nu}\delta_{jk}\,.
\end{equation}
That is, they behave as genuine fermionic Majorana operators both within and  across copies of $\HC^{\otimes t}$.
In particular, when $j\neq k$ we have
\begin{equation}
    \widetilde{c}_\mu^{\,(j)} \widetilde{c}_\mu^{\,(k)} =   \id \otimes \cdots \otimes \Gamma_j c_\mu \otimes \Gamma_{j+1} \otimes  \Gamma_{j+2} \otimes \cdots \otimes c_\mu \otimes  \id \otimes\cdots \,, \nonumber
\end{equation}
and when $j=k$, we find $\left(\widetilde{c}_\mu^{\,(j)}\right)^2=\id \otimes \id \otimes \cdots$.

Next, let us introduce the important family of operators
\begin{equation}\label{eq:tildeQdef}
    \widetilde{Q}_{j,k} =\frac{1}{2}\sum_{\mu=1}^{2n} \widetilde{c}_\mu^{\,(j)} \widetilde{c}_\mu^{\,(k)}\,.
\end{equation}
In general we  assume $j\neq k$ since $(\widetilde{Q}_{j,j})^2\propto \mathbbm{1}$ as a consequence of Eq.\ \eqref{eq:Majoranaanticomm}. These operators are crucial for the two following important Lemmas, whose proof can be found in Appendix~\ref{ap:generators-active}:
\begin{lemma}\label{lemma:QinC}
    Let $\widetilde{Q}_{j,k}$ be an operator defined as in Eq.~\eqref{eq:tildeQdef}. Then, $\widetilde{Q}_{j,k}\in \CC_{t,n}$.
\end{lemma}

\begin{lemma}\label{lem:Q-rep-gl}
The operators $\widetilde{Q}_{j,k}$ satisfy the Lie algebra commutation relations
\small
\begin{equation} 
\left[\widetilde{Q}_{j,k},\widetilde{Q}_{j',k'}\right]= \delta_{kj'} \widetilde{Q}_{j,k'} + \delta_{jk'} \widetilde{Q}_{k,j'} -\delta_{kk'} \widetilde{Q}_{j,j'} - \delta_{jj'} \widetilde{Q}_{k,k'} \nonumber\,,
\end{equation}
\normalsize
and therefore generate a representation of the $\mathfrak{so}(t)$ algebra.
\end{lemma}

Lemma \ref{lem:Q-rep-gl} is crucial to construct the commutant $\CC_{t,n}$ of fermionic Gaussian unitaries, as one can now invoke the Howe duality~\cite{hasegawa1989spin,wenzl2020dualities}. 
In our language, the Howe duality asserts that the commutant algebra of $\mathbb{SPIN}(2n)$ acting on $\HC_+^{\otimes t}$ is the image of the universal enveloping algebra of $\so(t)$. Therefore, it follows from this duality that linear combinations of monomials in the $\widetilde{Q}_{jk}$ and the identity span the commutant in the $\HC_+^{\otimes t}$ sector. To extend the commutant to the entire space $\HC^{\otimes t}$, we use the fact that this extension can be realized as a $\mathbb{Z}_2$ extension by including a generator $F$ that satisfies $F^2 = \id$ (i.e., $F$ is an involution), and also~\cite{wenzl2020dualities}
 \begin{equation}\label{eq:involution}
        F \widetilde{Q}_{j,k}  F =\begin{cases}
            - \widetilde{Q}_{j,k} \quad {\rm if\;} j=1\; {\rm or}\; k=1 \\
            \phantom{-} \widetilde{Q}_{j,k} \quad {\rm otherwise}
        \end{cases}\,.
\end{equation}
The parity operator $\Gamma_1$ can be readily seen to fulfill the requirements of the additional generator $F$. Hence, employing the notation $\left< \widetilde{Q}_{j,j+1}, \Gamma_1 \right>$ to denote the products that can be obtained from the generators, the previous discussion leads to the following theorem.
\begin{theorem}[Commutant generators]\label{th:generators}
    The $t$-th order commutant of $n$-mode fermionic Gaussian unitaries, $\CC_{t,n}$, is generated as
    \begin{equation}
        \CC_{t,n} = \Span_{\CBB} \left< \widetilde{Q}_{j,j+1}, \Gamma_1 \right> \,,\quad {\rm with}\; j\in[t-1] \,.
    \end{equation}
\end{theorem}
We refer the reader to Appendix~\ref{ap:generators-active} for a detailed proof. Notice here that we are not choosing all the $\widetilde{Q}_{j,k}$ as generators, since the $\OC(t)$ neighboring $\widetilde{Q}_{j,j+1}$ suffice to generate the entire $\so(t)$ algebra, as directly implied by the commutation relations in Lemma~\ref{lem:Q-rep-gl}. Theorem~\ref{th:generators} identifies $\CC_{t,n}$ as the algebra spanned by polynomials in the generators $\widetilde{Q}_{j,j+1}$ and $\Gamma_1$. As in the PP case, the question of how to find a basis for $\CC_{t,n}$ is not addressed by Theorem~\ref{th:generators} by itself. Instead, we address this problem below where we also identify connections with the approaches for PP unitaries of Section \ref{sec:GTappliedtoPP}.

\subsection{Dimensions of the commutant}

We next compute an explicit formula for the dimension of the fermionic Gaussian unitaries commutant. We do so by again resorting to representation-theoretic tools, namely, Weyl integration and Selberg-type integrals. The result is stated in the next theorem, and proven in Appendix~\ref{ap:PP-dim}.

\begin{theorem}[Commutant dimension]\label{th:dimension}
    The dimension of the $t$-th order commutant of the $n$-qubit matchgate group is given by
    \begin{equation}
        \dim(\CC_{t,n}) = \frac{1}{2^{n-1}} \prod_{j=0}^{n-1} \frac{(2j)!\,(2t+2j)!}{(t+j)!\,(t+n+j-1)!}\ \,.
    \end{equation}
\end{theorem}

In Fig.~\ref{fig:dimensions} we plot $\dim(\CC_{t,n})$ alongside $\dim\left(\CC^{\rm PP}_{t,n}\right)$. There, we see that while $\dim(\CC_{t,n})$ still grows extremely rapidly with $n$ and $t$, it is much smaller than the commutant for PP fermionic Gaussian unitaries.

\subsection{Constructing a basis for the commutant via GT}

Given that the pair $(\mathbb{SPIN}(2n),\mathscr{U}\left(\so(t)\right))$  admit mutually commuting representations over $\mathcal H^{\otimes t}$, we can use the extended skew Howe duality~\cite{wenzl2020dualities} to obtain an irrep decomposition
\begin{equation}
    \mathcal H =  \bigoplus_{\lambda} V_\lambda^{\mathbb{SPIN}(2n)}\otimes W_\lambda^{\so(t)}\,.
\end{equation}
Here,  the irreducible $\mathfrak{so}(t)$-modules are indexed by dominant the so-called highest-weights $P_+$
\begin{equation}
    \lambda=(\lambda_1,\dots,\lambda_r)\in P_+\,,
\end{equation}
together with the standard dominance conditions
\begin{equation}
    \lambda_1\ge \lambda_2\ge \cdots \ge \lambda_r\ge 0
    \qquad\text{for $t=2r+1$ (type $B_r$)},\nonumber
\end{equation}
and
\begin{equation}
    \lambda_1\ge \lambda_2\ge \cdots \ge \lambda_{r-1}\ge |\lambda_r|
    \qquad\text{for $t=2r$ (type $D_r$)}.\nonumber
\end{equation}

From here, we can implement the GT procedure for different values of $t$. We here present explicit results for $t=2,3$ and in the appendix we present some considerations for $t=4$. Importantly, we note that, as discussed in Appendix~\ref{ap:gauss_GT_method}, the fact that parity is preserved allows us to construct the commutant in two steps. First, we find a partial basis of the commutant $A_{t,n}$ obtained by  focusing on the even parity sector, and then generalize to the full commutant as
\begin{equation}\label{eq:active:Ant}
    \CC_{t,n} = {\rm span}_{\mathbb{C}}\left< A_{t,n}, \Gamma_1\right>\,,
\end{equation} 
which is precisely why the form obtain in Theorem~\ref{th:generators}.

In addition, we will see that our analysis implies that for the first few $t$ the commutants of general gaussian unitaries very closely resemble those for PP ones, a fact that will prove to be particularly useful. 

\subsubsection{$t=1$ case}

Here we have that since $\mathfrak{so}(1)=0$ is trivial, and there are no nontrivial orthogonal highest-weights to choose from. Moreover, there are no $\widetilde{Q}_{j,k}$ operators to realize as one needs $j$ strictly smaller than $k$. Therefore, the associative algebra is just the scalar algebra $A_{1,n}=\mathbb{C}\id$, meaning that  we can complete the commutant as
\begin{equation}
    \CC_{1,n}={\rm span}_{\mathbb{C}}\{\id,\Gamma\}\,.
\end{equation}

\subsubsection{$t=2$ case}\label{sec:t2active}
While this corresponds to the first non trivial, we can readily solve it as the algebra $\mathfrak{so}(2)$ is abelian. As such, the GT chain is trivial since it terminates at the first term. However, unlike the $t=1$ scenario, we do have a operators $\widetilde{Q}_{1,2}$, from which we can define the Hermitian generator $M=-i\widetilde{Q}_{1,2}$ of the infinitesimal action of $\mathfrak{so}(2)$. One can easily check that the operator $M$ has eigenvalues $\{-n,-n+1,\cdots,n-1,n\}$, meaning that we can construct the projectors onto the irreps as
\begin{equation}\label{eq:proj-M-t-2}
    P_m=\prod_{\substack{s= -n\\s\neq m}}^n\frac{M-s}{m-s}\,.
\end{equation}
Then, a basis of the full commutant is
\begin{equation}
    \CC_{2,n}={\rm span}_{\mathbb{C}}\left< P_m,\Gamma P_m\,|\, m=-n,-n+1,\ldots,n\right>\,.
\end{equation}
The dimension of the commutant is thus
\begin{equation}
    \dim(\CC_{2,n})=2(2n+1)\,.
\end{equation}
This agrees with the dimension of the basis provided in Section~\ref{sec:tcommutatbackground}. 

Here we can again wonder how this basis relates to that in Eq.~\eqref{eq:Qt2}. Notably, we can use the fact that $t=2$ case for general fermionic Gaussian unitaries very  closely resembles the $t=1$ for PP ones. This connection goes well beyond the simple fact that both cases are abelian, and we refer the reader to Appendix~\ref{ap:generalgaussianconstr} for additional details. Importantly, we can leverage this deep relation along with Eq.~\eqref{eq:passiveEPrelation} to prove the following Lemma.
\begin{lemma}\label{lemm:QfromP} Consider the basis of $\CC_{2,n}$ in Eq.~\eqref{eq:Qt2} and take the operators
\begin{equation}
E_{2k}=Q_{2k}^0,
\qquad
E_{2k+1}=-i\,Q_{2k+1}^1\,,
\end{equation}
as well as $\Gamma_1 E_k$. Then, consider the operators $P_m$ of Eq.~\eqref{eq:proj-M-t-2}. We obtain
    \begin{equation}\label{eq:QfromP}
    E_k=\sum_m K_k(n-m;2n)P_m\,,
\end{equation}
where $K_a(x;n)=\sum_{j=0}^a (-1)^j \binom{x}{j}\binom{n-x}{a-j}$ are the Krawtchouk polynomials.
\end{lemma}
The proof of this relation is provided in the Appendix \ref{app:active2}. Notice that using the orthogonality relation of the Krawtchouk polynomials we can  invert Eq.\ \eqref{eq:QfromP} to write 
\begin{equation}
P_m=\frac{1}{d^2}
\binom{2n}{n-m}
\sum_{k=0}^{2n}
\binom{2n}{k}^{-1}K_k(n-m;2n)\,E_k.
\end{equation}

\subsubsection{$t=3$ case}

The $t=3$ case corresponds to the first non-abelian algebra $\mathfrak{so}(3)$. However, here we can leverage the isomorphism $ \mathfrak{so}(3)\cong\mathfrak{su}(2)$ to greatly simplify the calculations. 
Since $\mathfrak{so}(3)$ has rank $1$ the available irreps are labeled by the single spin weight $\lambda=0,1,\ldots,n$, and we know that for each $\dim(W_\lambda^{\mathfrak{so}(2)})=2\lambda+1$. 
Then, the GT subgroup chain is 
\begin{equation}
    \mathbb{SO}(3)\downarrow\mathbb{SO}(2)\,,
\end{equation}
and the patterns for a given $\lambda$ are obtained from the $\mathbb{SO}(2)$ weight, i.e.,
\begin{equation}
    \GT(\lambda)=-\lambda,\lambda+1,\ldots,\lambda\,,
\end{equation}
such that $|\GT(\lambda)|=2\lambda+1$ (this makes the connection to spin irreps even more suggestive). 

To obtain the commutant basis we define the Hermitian spin operators
\begin{equation}
    J_z=-i\widetilde{Q}_{1,2}\,,\;J_+=\widetilde{Q}_{1,3}+i\widetilde{Q}_{2,3}\,,\; J_-=-(\widetilde{Q}_{1,3}-i\widetilde{Q}_{2,3})\,,\nonumber
\end{equation}
which can be readily found to satisfy the $\mathfrak{su}(2)$ commutation relations. From here, one introduces the ``total angular momentum'' operator
\begin{equation}
    J^2=J_z^2+\frac{1}{2}(J_+J_-+J_-J_+)\,,
\end{equation}
with eigenvalues $\lambda(\lambda+1)$, and obtains the irrep projectors
\begin{equation}
    P_\lambda=\prod_{\substack{\lambda'= 0\\\lambda'\neq \lambda}}^n\frac{J^2-\lambda'(\lambda'+1)}{\lambda(\lambda+1)-\lambda'(\lambda'+1)}\,,
\end{equation}
and the one dimensional ones
\begin{equation}
    P_{\lambda,m}=P_\lambda\prod_{\substack{m'= -n\\m'\neq m}}^n\frac{J_z-m'}{m-m'}\,.
\end{equation}
From here, the commutant basis simply follows the standard GT recipe.

Let us now remark that contrary to our $t=2$ construction for PP Gaussian unitaries, we have not yet provided a closed form for $P_{\lambda,m}$. Even if the algebras at play are isomorphic, in the PP case one obtains that the projector onto the vacuum state is a natural element of the commutant (in any number of copies). Instead, for general Gaussian operators there is no straightforward choice of such a state. Notably, it is still possible to find a construction that closely resembles the PP case.
The trick is to introduce creation and annihilation operators for a new ``emergent'' set of fermions $f$ lying in, say, the first two copies of the Hilbert space. We describe this construction in detail in the Appendix \ref{ap:activet3} where we show for example that $P^{\mathrm{hw}}_{j=n}
    =
    |0_f\rangle\langle 0_f|\otimes \id_{(3)}$, which projects onto the vacuum of the fermions $f$,  is an element of the commutant.
Just as the projectors in section \ref{sec:t2active}
provide an alternative description of the commutant with respect to the $Q^i_k$, the current description differs from the one provided in \cite{wan2022matchgate}, and we leave for future work the determination of the change of basis that connects them.

\section{Applications}\label{sec:applications}

In this section we present different applications for our main results. First, we will connect our commutant elements with known quantifiers of genuine fermionic correlations, as well as propose new techniques for studying non-fermionic Gaussianity in quantum states. Next, we analytically compute the  stabilizer entropy of general and particle-preserving fermionic Gaussian states.

\subsection{Commutants and genuine Fermionic correlations}

\subsubsection{General fermionic Gaussian unitaries}
Recent works~\cite{diaz2025unified,bermejo2025characterizing,deneris2025analyzing,coffman2026group} have shown that the commutant elements provides a natural starting point for studying quantum resource theories~\cite{chitambar2008tripartite}. Indeed, by characterizing all the possible invariants under the action of free operations (here taken to be as general fermionic Gaussian unitaries), one can provide a fine-grained description of the resources in a quantum state.

For example, for $t=2$, if one computes quantities such as
\begin{equation}
    \Tr[\rho^{\otimes 2}M]\,,\quad\text{for}\quad M\in\CC_{t,n}\,,
\end{equation}
then one obtains scalars which remain invariant under the transformation $\rho\rightarrow R(U)\rho R\ad(U)$ for all $U\in \mathbb{SPIN}(2n)$, and which quantify genuine fermionic correlations~\cite{diaz2023showcasing,bermejo2025characterizing}. As we discuss now,  our results allow us to reinterpret results in the literature in terms of our commutant basis elements. 

In particular, recall the basis for $\CC_{2,n}$ given by the operators $Q_k^0,
Q_k^1$ in Eq.~\eqref{eq:Qt2}.  The quantities $p_{k,i}=\Tr\Big[\rho^{\otimes 2} \frac{(-1)^{\lfloor \frac{k}{2}\rfloor}}{d}Q^i_k\Big]$, which are a probability distribution for pure states, were explored in \cite{diaz2023showcasing} and shown to provide a fingerprint of the state's resourcefulness. Additionally,  for $k=2$ one also recovers known measures of entanglement based on the generalized one-body reduced density matrix \cite{barnum2004subsystem,gigena2015entanglement}. In light of Section~\ref{sec:t2active}, we know that we can describe such commutant basis in terms of the spectral projectors $P_m$ of Eq.~\eqref{eq:proj-M-t-2} as per Lemma \ref{lemm:QfromP}.
The $P_m$ description of the commutant makes manifest that the only non-trivial part occupied by pure Gaussian states $|\psi\rangle^{\otimes 2}$ is the identity. This follows from the fact that~\cite{bravyi2004lagrangian}
\begin{equation}
    \widetilde{Q}_{12}|\psi\rangle^{\otimes 2}\equiv Q_{12}|\psi\rangle^{\otimes 2}=0\,.
\end{equation}
Moreover, it is also instructive to recover the purities $p_k=\Tr\Big[\rho^{\otimes 2} \frac{(-1)^{k}}{d}Q^0_{2\kappa}\Big]$ of fermionic Gaussian states from $P_m$ and the change of basis. Using Eq.\ \eqref{eq:QfromP} along with $P_m|\psi\rangle^{\otimes 2}=\delta_{m0}|\psi\rangle^{\otimes 2}$ leads to
\begin{equation}
    {\rm Tr}[Q^0_{2k} (\ket{\psi}\bra{\psi})^{\otimes 2}]=K_{2k}(0)=(-1)^k \binom{n}{k}\,,
\end{equation}
with the other overlaps vanishing ($K_{2k+1}(0)=0$)
leading to $p_k=\frac{1}{d}\binom{n}{k}$. This recovers exactly the result from \cite{diaz2023showcasing,bermejo2025characterizing}.

The previous discussion shows that our $P_m$ basis make manifest a point that is obscured when using the operators $Q^i_k$ which were obtained by applying Schur's lemma to the irrep decomposition of the operator space~\cite{diaz2023showcasing}. Namely, that the information in the invariants obtained from an operator irrep-based construction (such as $Q^i_k$) is not independent--in agreement with the extreme case of Gaussian states and Wick's theorem. Thus  the current characterization of the commutant, even for $t=2$,  provides a different and complementary angle for studying quantum resources from an algebraic setting.

\subsubsection{PP Gaussian unitaries: purities of reduced density matrices}

We now consider the quantum resource theory where  PP fermionic Gaussian unitaries are free operations.  We will show how Theorem \ref{th:PP-generators} can lead to new resource quantifiers already in the $t=2$ case.

Let us first make a few basic remarks. When working with dressed fermionic operators,  their two-copy vacuum $|\tilde{0}\rangle$ is $|\tilde{0}\rangle=|0\rangle \otimes |0\rangle$ since the latter is annihilated by all $\widetilde{a}^{(l)}_i$. Moreover, given a general state $\rho$ with well defined parity,  we have $\Gamma_1 \rho^{\otimes 2} \Gamma_1=\Gamma_1 \rho \Gamma_1 \otimes \rho = \rho^{\otimes 2}$. This means that when working on two copies we can always interchange $\rho^{\otimes 2}$ with its corresponding dressed version and viceversa. Equivalently we can replace $a_i^{(l)}\leftrightarrow \widetilde{a}_i^{(l)}$. Let us then notice that for a state $|\psi\rangle$ with fixed particle number $N$, $J_z |\psi\rangle^{\otimes 2}=0$ as it follows from $J_z=(N_1-N_2)/2$ (see last part of section \ref{sec:PPt2}). This means that in studying invariant quantities such as $\Tr[\rho^{\otimes 2}X^{(N,j)}_{mm'}]$ only the sector $m=m'=0$ with $N=2r$  is not trivial (here we use the convention $m\equiv m_{\text{spin}}$ introduced at the end of Section~\ref{sec:PPt2}).\\

With the previous clarifications, we can now discuss the relation between the commutant elements and the reduced density matrices (RDMs). Let $|\psi\rangle$ be a pure  state of $r$ fermions and let
\begin{equation}
\rho^{(K)}_{\alpha\alpha'}
=
\langle \psi|C^{(K)\dagger}_{\alpha'}C^{(K)}_{\alpha}|\psi\rangle\,,
\end{equation}
denote its $K$-body reduced density matrix, where $C^{(K)\dagger}_{\alpha}=a^\dag_{i_1}\dots a^\dag_{i_K}$ creates the Slater determinant associated with the ordered $K$-tuple $\alpha=(i_1,\dots,i_K)$. As shown in Ref.~\cite{gigena2021many}, the matrix $\rho^{(K)}$ determines expectation values of arbitrary $K$-body observables and its spectrum encodes the $(K,r-K)$ many-body entanglement structure.
For our purposes, the key observation is that the purity of $\rho^{(K)}$ can be written as an overlap of $\rho^{\otimes 2}$ with a natural operator in the PP second commutant. Indeed, let us define
\begin{equation}
\Omega_K
=
\sum_{\alpha,\alpha'}
C_{\alpha'}^{(K,1)\dagger}C_{\alpha}^{(K,1)}
\,C_{\alpha}^{(K,2)\dagger}C_{\alpha'}^{(K,2)},
\label{eq:OmegaM}
\end{equation}
where superscripts $(1),(2)$ denote the two copies. After dressing and reordering the ladder operators, $\Omega_K$
 becomes a polynomial in $\widetilde{\Omega}_{ij}$ so that 
Theorem \ref{th:PP-generators} implies $\Omega_K\in  \CC_{2,n}^{\rm PP}$ (see the example of $\Omega_1$ below). Then
\begin{align}
\Tr[\rho^{\otimes 2}\Omega_K]
&=
\sum_{\alpha,\alpha'}
\langle \psi|C_{\alpha'}^{(K)\dagger}C_{\alpha}^{(K)}|\psi\rangle
\langle \psi|C_{\alpha}^{(K)\dagger}C_{\alpha'}^{(K)}|\psi\rangle
\nonumber\\
&=
\sum_{\alpha,\alpha'}
\rho^{(K)}_{\alpha\alpha'}\rho^{(K)}_{\alpha'\alpha}
=
\Tr[(\rho^{(K)})^2].
\label{eq:purityOmega}
\end{align}
Hence the quadratic purity of the $K$-body reduced state is itself the overlap with an element in $\CC_{2,n}^{\rm PP}$.
If we normalize the  $K$-body reduced density matrix,
then its quadratic entropy may be written as
$
S_2^{(K)}
=
1-\binom{r}{K}^{-2}\Tr[(\rho_n^{(K)})^2]
=
1-\binom{r}{K}^{-2}\Tr[\rho^{\otimes 2}\Omega_K].
$
In particular, for $K=1$ we have $C^{(1)\dagger}_{\alpha}=a_p^\dagger$, and Eq.~\eqref{eq:OmegaM} becomes
$
\Omega_1
=
\sum_{p,q}
a_{q}^{\dagger(1)}a_{p}^{(1)}
a_{p}^{\dagger(2)}a_{q}^{(2)}$. 
In the notation of Section~\ref{sec:PPt2}
one finds
\begin{equation}
\Omega_1=N_1-J_+J_-\,,
\end{equation}
such that the one-body entropy is given by 
\begin{equation}
\Tr[(\rho^{(1)})^2]
=
\Tr\left[\rho^{\otimes 2}\left(\frac{N}{2}-J_+J_-\right)\right].
\label{eq:onebodypurityJ}
\end{equation}
This expression allows us to derive an important conclusion: In the $m=0$ subspace
$
J_+J_-|j,0\rangle=j(j+1)|j,0\rangle
$
so that Eq.~\eqref{eq:onebodypurityJ} may be rewritten as
\begin{equation}
\Tr[(\rho^{(1)})^2]
=
\sum_j [r-j(j+1)]\,p_j\,,
\label{eq:puritypj}
\end{equation}
where we defined the probabilities 
\begin{equation}
    p_j=\Tr\left[\rho^{\otimes 2}X^{(2r,j)}_{00}\right]\,.
\end{equation}
This shows that the overlaps with the projectors $X^{(2r,j)}_{00}$ provide a refinement of the well known one-body purity. Indeed, the quadratic entropy depends  on the coarse grained sum in Eq.~\eqref{eq:puritypj}, while the individual probabilities $p_j$ contain fine grained information. As such, in the spirit of \cite{bermejo2025characterizing} we conjecture that the complete collection of $p_j$ will provide a finer-grained characterization of the state's resourcefulness. 

As an example of such a proposal, we first consider the case when the quantum state is a pure Slater determinants of fixed number of particles $r$, i.e., a state with no genuine fermionic correlations. For example let $|r\rangle=a_1^\dagger a_2^\dagger\cdots a_r^\dagger|0\rangle$
be a reference Slater determinant. Its two-fold tensor product can be written, up to an overall sign, as
\begin{equation}
|r\rangle^{\otimes 2}
\sim
\prod_{p=1}^r
\Big(a_p^{\dagger(1)}a_p^{\dagger(2)}\Big)|0\rangle.
\end{equation}
Each individual factor $d_p^\dagger=a_p^{\dagger(1)}a_p^{\dagger(2)}$
satisfies
\begin{equation}
J_z\,d_p^\dagger|0\rangle=0,
\qquad
J_\pm\,d_p^\dagger|0\rangle=0\,.
\end{equation}
 Hence the copy Slater determinant is itself a singlet 
\begin{equation}
J_z|r\rangle^{\otimes 2}=
J_\pm|r\rangle^{\otimes 2}=
J^2|r\rangle^{\otimes 2}=0.
\end{equation}
Since every Slater determinant with $r$ particles is obtained from $|r\rangle$ by a passive FLO transformations, and the projectors $X^{(2r,j)}_{00}$ belong to the passive commutant, the same statement holds for any Slater determinant $|\psi_{\rm SD}\rangle$. We thus have 
\begin{equation}\label{eq:Sloverlap}
p_j=\Tr\left[\bigl(|\psi_{\rm SD}\rangle\langle\psi_{\rm SD}|\bigr)^{\otimes 2}
X^{(2r,j)}_{00}\right]
=\delta_{j0}\,.
\end{equation}
This equation implies that Slater determinants only occupy the $j=0$ sector of the commutant. This captures the intuition of Ref.~\cite{bermejo2025characterizing} where it was stated that free states contains the most compressible amount of information. In the current case, this intuition is extreme as a Slater determinant only lies in the smallest possible subspace compatible with the corresponding number of particles and the representation. One finds that the corresponding dimension is $d_{r}=\dim V_{(2^r)}^{\mathbb U(n)}=\frac{1}{r+1}\binom{n}{r}\binom{n+1}{r}$ (one makes use of $\lambda_1=\lambda_2=r$ so that $\lambda^T=(2^r)$ and the Weyl dimensional formula in Appendix \ref{ap:PP-dim}). Notice that the formula is invariant under the particle-hole transformation mapping $r\to n-r$ and that the maximum is attained at the center ($r=n/2$ for even $n$ and for both $r=\left \lfloor{\frac{n}{2}}\right \rfloor, \left \lfloor{\frac{n}{2}}\right \rfloor+1$ for odd $n$). For large $n$ one finds the asymptotic Gaussian behavior 
\begin{equation}\label{eq:gaussian}
    d_{r}\simeq d_{\lfloor{\frac{n}{2}}\rfloor}\exp\left[-\frac{4}{n}\Big(r-\left \lfloor{\frac{n}{2}}\right \rfloor\Big)^2\right]\,.
\end{equation}

To finish, we recall that 
\begin{equation}
\Tr[\rho^{\otimes 2}J_+J_-]=r-\Tr(\rho^{(1)2}),
\end{equation}
is equal to zero if and only if the state is a Slater determinant, so that we recover a witness
of free state from the condition $p_j=\delta_{j0}\Leftrightarrow \Tr[\rho^{\otimes 2}J_+J_-]=0$. But this is just one particular witness, based on Eq.\ \eqref{eq:puritypj}. Instead, one can diagnose deviations from a free state in all the different $j$ directions individually via $p_j$ or suitable convex combinations of $p_j$.

\subsubsection{PP Gaussian unitaries: generalized Pl\"ucker conditions from spin truncation in ${\mathfrak{su}(2)}$}

In this section we show how our results can be used to recover the resourcefulness criteria of Ref.~\cite{semenyakin2025classifying}. Therein, the authors classified fermionic states from the quantities $\omega_k={\rm tr}[\rho^{\otimes 2}\widetilde{\Omega}^k_{12}]$ which they showed can be written as a ``twisted'' version of the purity of the one-body RDM. From Theorem \ref{th:PP-generators} it is clear that $\widetilde{\Omega}^k_{12}\in \CC^{\rm PP}_{2,n}$ since it is just a monomial of the generator of the algebra. In this sense, the twisted-purity witnesses of \cite{semenyakin2025classifying} are naturally embedded in the broader scheme provided by the full knowledge of $\CC^{\rm PP}_{2,n}$.

Moreover, in \cite{semenyakin2025classifying} the authors introduce the ``Generalized Pl\"ucker relations'':
\begin{equation}\label{eq:plucker}
    \widetilde{\Omega}_{12}^k|\psi\rangle^{\otimes 2}=0
\end{equation}
where $k$ is a positive integer. These conditions define a hierarchy of complexity of states with $k=1$ corresponding to Slater determinants.
 Notably, these conditions can be easily understood in our framework: \begin{equation}
\widetilde{\Omega}_{12}^k\,|\psi\rangle^{\otimes 2}=0
\;\Leftrightarrow \;\;
p_j=0\; \text{for}\; j\geq k.
\end{equation}
In other words, the Pl\"ucker condition $\widetilde{\Omega}_{12}^k|\psi\rangle^{\otimes 2}=0$ is equivalent to the statement that the copy state has support only on spin sectors $j<k$. This is easy to understand  by noting that a fermionic state with fixed number of particles leads to $|\psi\rangle^{\otimes 2}
=\sum_{j,\alpha} c_{j,\alpha}\,|j,m=0,\alpha\rangle$ for $\alpha$ a degeneracy index. Namely, we only occupy the $m=0$ sector. As a consequence,  $\widetilde{\Omega}_{12}^k|\psi\rangle^{\otimes 2}=\sum_{j,\alpha} c_{j,\alpha}\,J_+^k|j,0,\alpha\rangle$ which is zero iff $c_{j,\alpha}=0$ for $j\geq k$. This result clearly captures once more the intuition presented in~\cite{bermejo2025characterizing}: A state such that $|\psi\rangle^{\otimes 2}$  occupies larger (and more) $j$  sectors requires, in principle, more information to be represented than states living in smaller subspaces of the commutant, and is thus more resourceful.

By following Weyl dimension formula explained in Appendix \ref{ap:PP-dim} we can actually find the dimension of each subspace of fixed $j$ and $r$. Fixing both
forces $\lambda=(r+j,r-j)$
which in  the transpose partition is
$
\lambda^T=(2^{\,r-j},1^{\,2j})$. From this we get 
\begin{equation}
    d_{r,j}=\dim V^{\mathbb U(n)}_{(2^{\,r-j},1^{\,2j})}=
\frac{2j+1}{r+j+1}\binom{n+1}{r-j}\binom{n}{r+j}\,,\nonumber
\end{equation}
holding for $j=0,1,\dots,\min(r,n-r)$. 
Thus the allowed space for states characterized by the condition \eqref{eq:plucker} for a given $k$ has dimension  $\sum_{j=0}^{k-1} d_{r,j}$. One can show that all terms peak at ``half''-filling  and that the particle-hole symmetry ($r\to n-r$) is preserved, just as for $j=0$. The same asymptotic behavior of Eq.\ \eqref{eq:gaussian} is also recovered for large $n$ (with a different overall constant).

\subsubsection{Properties of Gaussian states and Slater determinants for general $t$}

To finish this section, we show that the basic properties of the dressed operator allow us to readily prove the following results: 
 Given a Slater determinant $|\psi\rangle$ and $ \widetilde{\Omega}_{ij}$ for $i\neq j$ defined in $t$ copies we have
    \begin{equation}\label{eq:pluckert}
        \widetilde{\Omega}_{ij}|\psi\rangle^{\otimes t}=0\quad \forall t\,.
    \end{equation}
Similarly, given a Gaussian state $|\psi\rangle$ and $ \widetilde{Q}_{ij}$ defined in $t$ copies we have
    \begin{equation}\label{eq:bravit}
        \widetilde{Q}_{ij}|\psi\rangle^{\otimes t}=0\quad \forall t\,.
    \end{equation}
Notice that $t=2$ the first equation generalizes the Pl\"ucker conditions~\cite{semenyakin2025classifying} while the second recovers the criteria of Ref.~\cite{bravyi2004lagrangian}. Hence, our results constitute a non-trivial generalization of the aforementioned results.

To see why these results hold, note that for a computational basis state, each term $\widetilde a_p^{\dagger(i)}\widetilde a_p^{(j)}$ necessarily annihilates $|\psi\rangle^{\otimes t}$ since if the mode $p$ is empty $\widetilde a_p^{(j)}$ will destroy it and if it is full $\widetilde a_p^{\dagger(i)}$ will. Then it is clear that the complete sum will annihilate the copied reference state. The same holds for any Slater determinant as $\widetilde{\Omega}_{ij}$ commutes with PP Gaussian unitaries.
Similarly, for the general case we can just take the vacuum and expand
\begin{equation}
    \widetilde Q_{ij}
    =
    \frac12\sum_{\mu=1}^{2n}\widetilde c_\mu^{(i)}\widetilde c_\mu^{(j)}=\sum_{p=1}^n
    \Big(
    \widetilde a_p^{\dagger(i)}\widetilde a_p^{(j)}
    -
    \widetilde a_p^{\dagger(j)}\widetilde a_p^{(i)}
    \Big)\,.
\end{equation}
It is then clear that $
    \widetilde Q_{ij}|0\rangle^{\otimes t}=0$ since the $\widetilde a_p^{(j)}$ annihilate the vacuum.

Hence, as a corollary of Eqs.~\eqref{eq:pluckert}--\eqref{eq:bravit} and of Theorems \ref{th:PP-generators}, \ref{th:generators} we thus know that the free fermionic states of the general Gaussian theory, when copied to $\mathcal{H}^{\otimes t}$ only occupy a one dimensional sector. For PP instead, the free states lie in the sector generated by $\widetilde{\Omega}_{ii}$, in agreement with the discussion about the $j=0$ sector in the example of $t=2$ of the previous section.

\subsection{Stabilizer entropy of fermionic Gaussian states}

The stabilizer entropy introduced in~\cite{leone2022stabilizer} has become the standard measure of the magic possessed by quantum states. In brief, this quantity can be used to assess the number of non-Clifford operations needed to prepare a quantum state, a characterization that is of pivotal importance in the near-term quantum devices landscape.
The representation-theoretic framework that allows us to completely characterize the commutant of fermionic Gaussian unitaries, also enables us to exactly compute the fermionic Gaussian states average of the linearized version of the stabilizer entropy, which we here recall to be given by $M_{\rm lin}\left(\ket{\psi}\right) = 1 - 2^n S_4(\ket{\psi})$, with
\begin{equation}
     S_4(\ket{\psi}) =\frac{1}{4^n}\sum_{P\in\PC_n}\Tr[\ketbra{\psi}{\psi}P]^4\,.
\end{equation}
Above, $\PC_n=\{I,X,Y,Z\}^{\otimes n}$ is the set of all Pauli strings of $n$ qubits.

\subsubsection{Fixed-particle fermionic Gaussian states}
\begin{figure}[t]
    \centering
    \includegraphics[width=\linewidth]{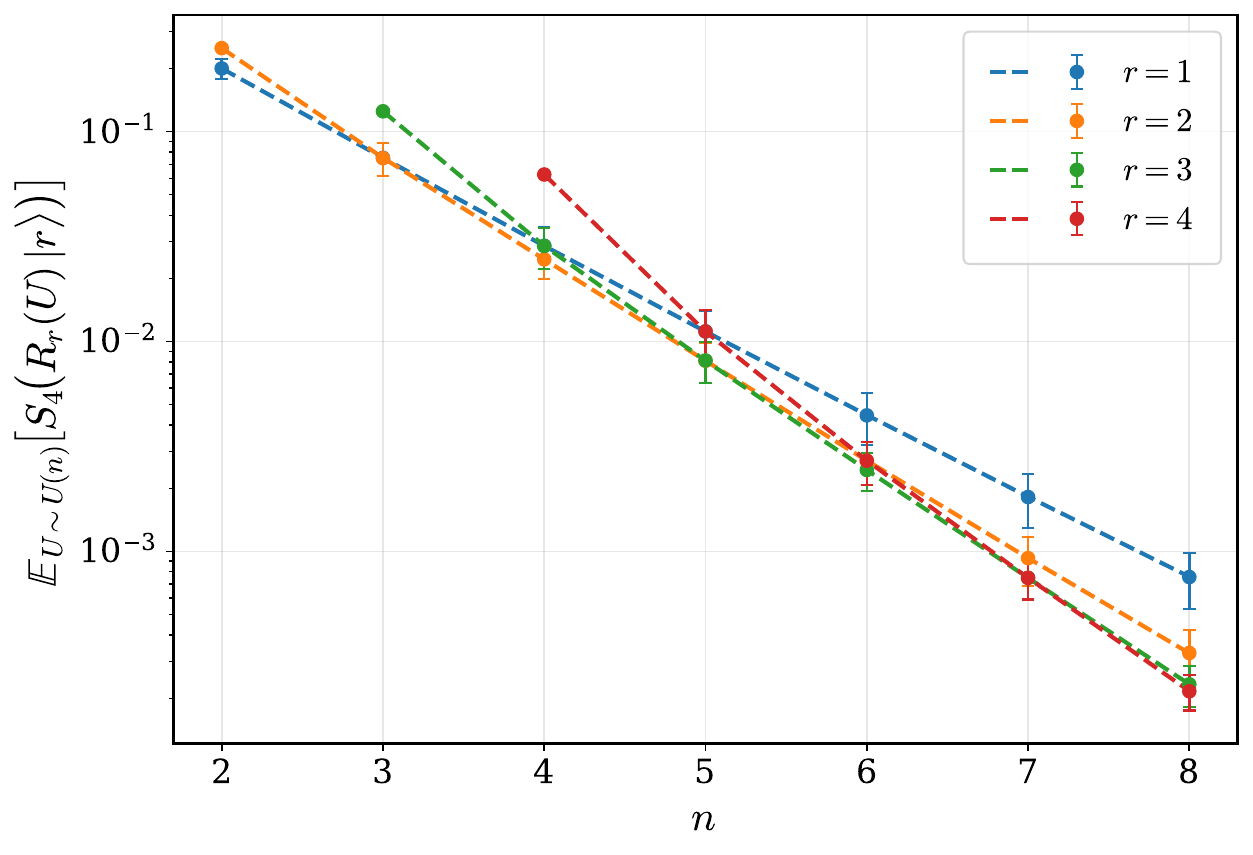}
    \caption{\textbf{Comparison of numerical estimates and analytical predictions for $\mathbb{E}_{U\sim\U(n)}S_4(R_r(U)\ket{r})$.} The solid dots correspond to the average over $10^4$ independent samples of particle-preserving fermionic gaussian unitaries, and the associated error bars show one standard deviation from the mean. The dashed lines instead are the closed formula shown in Eq.~\eqref{eq:avg_pp_stab_sector_formula_1}. The system sizes considered are $n=2,\dots,8$ and for each $n$ we show the results from the sectors with $r=1,\dots,\min(n,4)$ particles, corresponding to colors blue, orange, green and red, respectively.}
    \label{fig:pp_avg_stab}
\end{figure}
As we show in Appendix~\ref{ap:avg_pp_stab_sector}, the average linear stabilizer entropy over the set of fermionic Gaussian states with fixed particle number $r$ is
\begin{align}
\label{eq:avg_pp_stab_sector_formula_1}
    &\mathbb{E}_{U\sim\U(n)}
    S_4\!\left(R_r(U)\ket{r}\right)\\
    &=
    \frac{1}{2^n\,d_{n,r}}
    \sum_{k=0}^{\min(r,n-r)}
    \frac{n!}{(n-r-k)!\,(2k)!\,(r-k)!}\,c_k\,.
\end{align}
Above, $R_r(U)$ is the representation of $\U(n)$ acting on the $r$-particle sector $\Lambda^r(\CBB^n)$, and we defined
\begin{equation}
    \ket{r}=a_1^\dagger a_2^\dagger \cdots a_r^\dagger \ket{0}
    \in \Lambda^r(\CBB^n)\,,
\end{equation}
as the reference $r$-particle Slater determinant. Then, we denoted by
\begin{equation}
    d_{n,r}=\dim\bigl(V^{\U(n)}_{(4^r)}\bigr)=\prod_{i=1}^r\prod_{j=1}^4
    \frac{n+j-i}{r+5-i-j}\,,
\end{equation}
the dimension of the $\U(n)$ irrep associated with the parition $(4^r)$, the rectangular tableau with $r$ rows and $4$ columns.
Furthermore, given a partition $\lambda=(\lambda_1,\lambda_2,\dots)$ of $r$, we define its double as
\begin{equation}
    2\lambda=(2\lambda_1,2\lambda_2,\dots)\,,
\end{equation}
and we denote by $H(2\lambda)$ the hook-product of the Young diagram $2\lambda$.
Lastly, given the generalized Pochhammer symbol
\begin{equation}
    (a)_\lambda
    =
    \prod_{j\geq 1}
    \left(a-\frac{j-1}{2}\right)_{\lambda_j}\!,
    \,
    (x)_k=x(x+1)\cdots(x+k-1)\,,\nonumber
\end{equation}
we defined
\begin{equation}
\label{ap-eq:ck_def}
    c_k
    =
    4^k
    \sum_{\lambda\vdash k,\ \ell(\lambda)\leq 3}
    \frac{(2k)!}{H(2\lambda)}
    \frac{\left(\frac{3}{2}\right)_\lambda}{(3)_\lambda}\,,
\end{equation}
for $\ell(\lambda)$ the number of rows of the Young diagram associated to the partition $\lambda$.

In Fig.~\ref{fig:pp_avg_stab} we plot the derived analytical expression in Eq.~\eqref{eq:avg_pp_stab_sector_formula_1} against its numerical estimates obtained by sampling particle-preserving fermionic Gaussian unitaries with the algorithm outlined in~\cite{braccia2025optimal}. Each estimate is obtained as the average over $10^4$ independent samples. The analytical prediction is in perfect agreement with the numerical data, bringing additional validation to our derivation.

\subsubsection{General fermionic Gaussian states}

\begin{figure}[t]
    \centering
    \includegraphics[width=\linewidth]{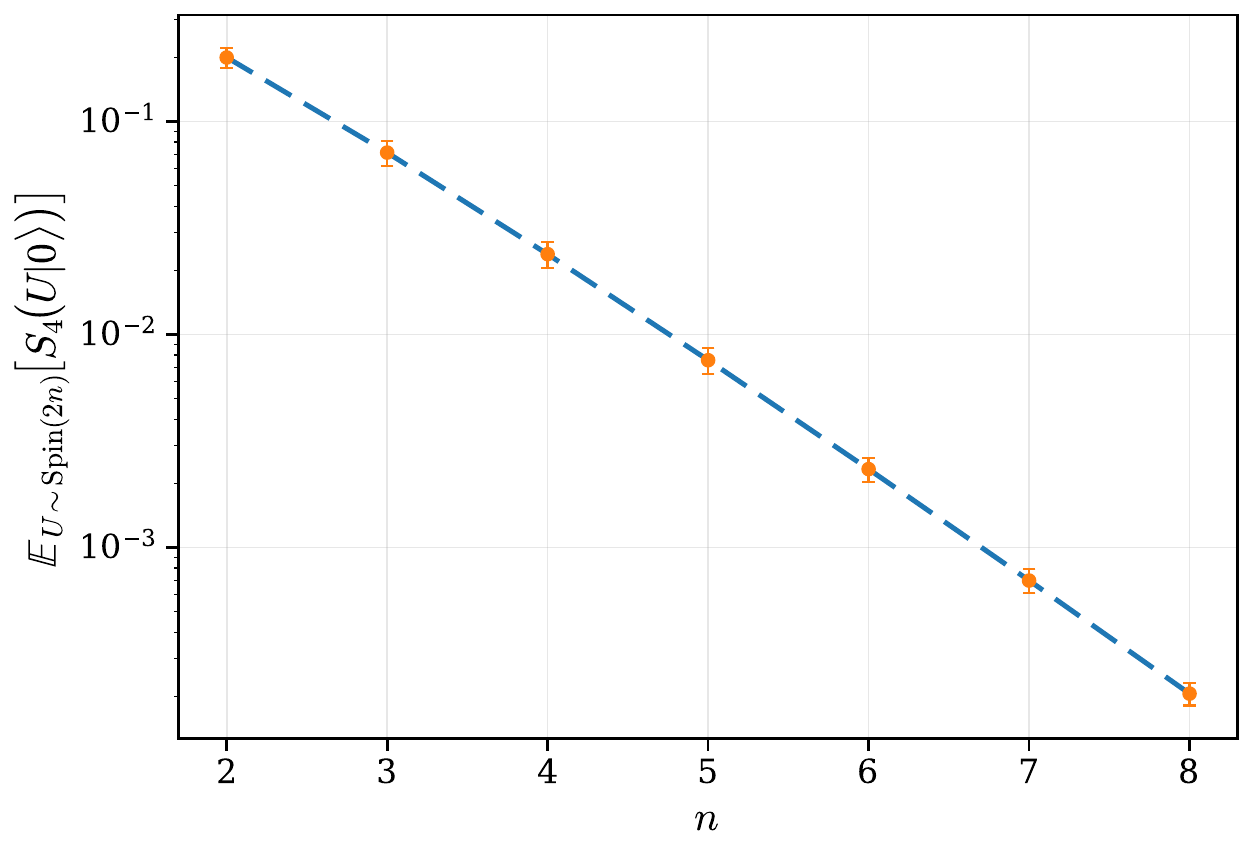}
    \caption{\textbf{Comparison of numerical estimates and analytical predictions for $\mathbb{E}_{U\sim\Spin(2n)}S_4(U\ket{0})$.} For increasing system sizes $n=2,\dots,8$, the dots show the average of $S_4(U\ket{0})$ over $10^4$ independent samples of $U\sim\Spin(2n)$ and the associated error bars show one standard deviation from the mean. The dashed line instead shows the analytical expression for $\mathbb{E}_{U\sim\Spin(2n)}S_4(U\ket{0})$ given in Eq.~\eqref{eq:avg_gauss_stab}.}
    \label{fig:avg_gaussian_stab}
\end{figure}

As we prove in Appendix~\ref{ap:avg_gaussian_stab}, over the set of even-parity fermionic Gaussian pure states $\ket{\psi}=R(U)\ket{0}$, for $U\in\Spin(2n)$ one finds
\begin{equation}
\label{eq:avg_gauss_stab}
    \mathbb{E}_{U\sim\Spin(2n)}S_4(R(U)\ket{0})= \frac{1}{\Cat_{n+1}}\,,
\end{equation}
and hence
\begin{equation}
    \mathbb{E}_{U\sim\Spin(2n)}M_{\rm lin}(R(U)\ket{0})= 1 - \frac{2^n}{\Cat_{n+1}}\,,
\end{equation}
for $\Cat_m=\frac{1}{m+1}\binom{2m}{m}$ the $m$-th Catalan number. 
We also notice that the very same result holds for the odd-parity fermionic Gaussian states $\ket{\phi}\in\HC_-$. Indeed, notice that one can switch between parity sectors $\HC_+,\HC_-$ via the action of the operator $X_1=c_1$. This operator either commutes or anticommutes with the Majorana binomials $c_\mu c_\nu$ that generate any fermionic Gaussian unitary via $R(U)=e^{\frac{1}{2} \sum_{\mu,\nu=1}^{2n}  h_{\mu\nu} c_\mu c_\nu}$ for some real antisymmetric matrix $h$. As a consequence, it stabilizes the fermionic Gaussian unitaries group $\Spin(2n)$, implying that integrating over $R(U)$ is the same as integrating over $X_1R(U)X_1$. 
Now, any odd-parity fermionic Gaussian state can be expressed as $\ket{\phi}=R(U)X_1\ket{0}=X_1(X_1 R(U) X_1)\ket{0}$. Thus, after the map $X_1 R(U) X_1 \to R(U)$, the only difference with the even-parity sector calculation is that the Paulis $P$ appearing in the summation that defines $S_4$ are now conjugated by $X_1$. However this leaves $S_4$ invariant, as $X_1$ is also a Clifford operator and hence normalizes the Pauli group, and any eventual phase $\{\pm1,\pm i\}$ disappears when we take the fourth power.

In Fig~\ref{fig:avg_gaussian_stab} we plot the derived analytical expression for the average $\mathbb{E}_{U\sim\Spin(2n)}S_4(U\ket{0})$ against the numerical estimates obtained by sampling $10^4$ random instances of fermionic Gaussian unitaries obtained via the method presented in~\cite{braccia2025optimal}. The data show perfect agreement, providing additional validation to our result.

Lastly, we compare Eq.~\eqref{eq:avg_gauss_stab} to the averages of $S_4$ over completely Haar random states over $\HC$ and product states where each qubit's state is Haar random in $\CBB^2$. One has~\cite{deneris2025analyzing}
\begin{align}
    \mathbb{E}_{\ket{\psi}\sim\HC}S_4(\ket{\psi})&= \frac{1}{2^{n-2}(2^n+3)}\\
    \mathbb{E}_{\ket{\phi_i}\sim\CBB^2}S_4(\otimes_i\ket{\phi_i}) &= \left(\frac{2}{5}\right)^n\,.
\end{align}
Particularly, we consider the large $n$ limit, where
\begin{align}
    \mathbb{E}_{U\sim\Spin(2n)}S_4(U\ket{0})&\sim \frac{n^{3/2}}{4^{n}}\frac{\sqrt{\pi}}{4}\\
    \mathbb{E}_{\ket{\psi}\sim\HC}S_4(\ket{\psi})&\sim \frac{4}{4^n}\\
    \mathbb{E}_{\ket{\phi_i}\sim\CBB^2}S_4(\otimes_i\ket{\phi_i}) &= \left(\frac{2}{5}\right)^n\,.
\end{align}
We hence see that taking the full Haar random average as benchmark, product states have exponentially less magic, while the magic of random fermionic Gaussian states is suppressed by only a factor $n^{3/2}$.

\section{Outlook}\label{sec:outlook}

In this work we characterized the $t$-th order commutants of both PP and general fermionic Gaussian unitaries. On the PP side, we showed that the commutant is generated by the dressed hopping operators between copies, while for general fermionic Gaussian unitaries it is generated by dressed quadratic Majorana bilinears together with parity. Beyond identifying generators, we provide general dimension formulas for the commutants and showcase how the GT construction can be used to obtain the basis elements. 

Beyond the purely mathematical constructions, our work also showcases the importance of providing some ``physical'' interpretation to the commutants. For instance, in the PP case, our $t=2$ analysis revealed a rich underlying structure which allowed us to view how different GT patterns map to basis elements via particles hopping from one copy onto the next. Moreover, the subtle connection uncovered between PP and general fermionic Gaussian unitary commutants enabled us to borrow results from one case, and apply it onto another.  

Then, our results also clarify the information-theoretic role of fermionic invariants arising in $\rho^{\otimes t}$. In particular, we make a connection with recent advancements in the algebraic characterization of the resourcefulness of quantum states~\cite{diaz2025unified,bermejo2025characterizing,deneris2025analyzing,coffman2026group}. Here, our basis for the commutant provides an alternative descriptor for resource and complexity analyses. From this perspective, several known quantities—such as purities of reduced density matrices, witnesses based on generalized Pl\"ucker-type conditions, and low-order Gaussianity criteria—appear as coarse-grained projections of the much finer algebraic structure enabled by our results. Indeed, the full commutant leads to a hierarchy of fermionic correlations that is invisible to standard low-order diagnostics, and suggests a route toward finer resource theories in which replicated-sector occupations quantify deviations from Gaussianity or from Slater-determinant structure.

To finish, we note that our work leaves several directions open. On the technical side, it would be valuable to push the GT construction to arbitrary $t$ in a fully explicit and computationally efficient way, and to derive closed change-of-basis formulas connecting our GT-adapted operators with previously known low-order commutant bases. On the structural side, it would be interesting to better understand the asymptotic representation theory of these commutants, as this could enable, e.g., understanding whether  PP fermionic Gaussian unitaries form Gaussian processes~\cite{garcia2023deep,garcia2024architectures}. Finally, on the application side, an important challenge is to determine how much of the commutant invariant data is operationally accessible in experiments, and which subsets of these invariants provide the sharpest signatures of non-Gaussianity, many-body complexity, or computational advantage. We expect that progress along these directions will further strengthen the role of commutant methods in fermionic quantum information theory.

\section*{Acknowledgments}
P.B. and M.C. were supported Laboratory Directed Research and Development (LDRD) program of Los Alamos National Laboratory (LANL) under project number 20260043DR. N.L.D. acknowledges supported by the Center for Nonlinear Studies at LANL. M.L. and M.C. were supported  by LANL's ASC Beyond Moore’s Law project. D.G.-M. acknowledges financial support from the European Research Council (ERC) via the Starting grant q-shadows (101117138) and from the Austrian Science Fund (FWF) via the SFB BeyondC (10.55776/FG7).

\medskip

\noindent {\bf Note added:} A few days before our work was submitted
to the arXiv, the manuscript~\cite{sierant2026theory} which independently studies the commutant of general matchgate circuits, or Bogoliubov transformations, which constitute a representation of $\mathbb{O}(2n)$. We remark that while the methods  and results in our work and in Ref.~\cite{sierant2026theory} are similar (e.g., the use of GT techniques to construct orthonormal bases of the commutant, or the derivation for the average stabilizer entropy for fermionic Gaussian states), there are some important key differences to highlight. First, we note that both works are complementary. Second, we stress that we here present a characterization of the commutants for the matchgate (or fermionic Gaussian) representations of the groups $\mathbb{U}(n)\subset\mathbb{SO}(2n)\subset\mathbb{O}(2n)$. Indeed, one can show that the restriction $\mathbb{SO}(2n)\subset\mathbb{O}(2n)$ doubles the dimension of the commutant, meaning that the results in Ref.~\cite{sierant2026theory} constitute the basis for $A_{t,n}$ in Eq.~\eqref{eq:active:Ant}, but our bases are twice as large. Still, even when focusing on the terms in $A_{t,n}$, our bases are slightly different than those in Ref.~\cite{sierant2026theory}, and we also showcase how the GT derivations are related to other known constructions. Then, and most importantly, our work considers the PP fermionic Gaussian representation of $\mathbb{U}(n)$, which is not studied in Ref.~\cite{sierant2026theory}. Finally, we remark that our application discussions significantly differ from those of Ref.~\cite{sierant2026theory}, especially when it comes to using the commutant elements to fine-grain the characterization of the resource of a state.

\bibliography{quantum}
\clearpage
\newpage

\appendix

\onecolumngrid

\setcounter{theorem}{0}
\setcounter{lemma}{0}

\section*{Appendix}

In the following appendices we provide additional details, as well a proofs for the results in the main text. 

\section{Proofs of Lemmas 1, 2 and Theorem 1}
\label{ap:PP-generators}

 In this appendix, we begin by providing a proof for Lemmas~\ref{lem:Omega-in-C} and~\ref{lem:Omega-rep-gl}, which we then use to prove  Theorem~\ref{th-ap:PP-generators}. 

 \begin{lemma}\label{lem:Omega-in-C-AP}
    Let $\widetilde{\Omega}_{j,k}$ be an operator as defined in Eq.~\eqref{eq:PP-gen}. Then, 
    \begin{equation}
        \widetilde{\Omega}_{j,k}\in\CC^{\rm PP}_{t,n}\,,\quad \forall j,k\,.
    \end{equation}
\end{lemma}

\begin{proof}
Let us first show that for all PP fermionic Gaussian unitaries $R(U)$ and all $j,k$,
\begin{equation} 
\left[R(U)^{\otimes t},\Omega_{j,k}\right] = 0\,,
\end{equation}
with
\begin{equation}
\Omega_{j,k}= \sum_{p=1}^n a^{\dagger (j)}_p a^{(k)}_p\,.
\end{equation}
Here, the creation and annihilation operators can be defined in terms of the Majoranas, as
\begin{equation}
    a_p^{(j)}=\frac{c_{2p-1}^{(j)}+i\,c_{2p}^{(j)}}{2}\,, \qquad 
    a_p^{\dagger (j)}=\frac{c_{2p-1}^{(j)}-i\,c_{2p}^{(j)}}{2}\,.
\end{equation}
 Using that under conjugation by a PP fermionic Gaussian unitary $R(U)$, $a_p,a_p^\dagger$ transform as
\begin{equation}
R(U) a^\dagger_p R^\dagger(U) = \sum_{q=1}^n u_{qp} a^\dagger_q \,, \qquad R(U) a_p R^\dagger(U) = \sum_{q=1}^n u^*_{qp} a_q \,,
\end{equation}
where $u \in \mathbb{U}(n)$ is unitary, we obtain
\begin{align}
R(U)^{\otimes t}\Omega_{j,k}R^\dagger(U)^{\otimes t} &= \sum_{p=1}^n R(U)^{\otimes t} a^{\dagger (j)}_p a^{(k)}_p R^\dagger(U)^{\otimes t} \nonumber\\
&= \sum_{p=1}^n \left(\sum_{q=1}^n U_{qp} a^{\dagger (j)}_q\right) \left(\sum_{r=1}^n U^*_{rp} a^{(k)}_r\right)\nonumber\\
&= \sum_{q,r=1}^n \left(\sum_{p=1}^n U_{qp} U^*_{rp}\right)a^{\dagger (j)}_q a^{(k)}_r \nonumber\\
&=\sum_{q=1}^n a^{\dagger (j)}_q a^{(k)}_q = \Omega_{j,k}\,,
\end{align}
as implied by the fact that $\sum_p U_{qp} U^*_{rp}=\delta_{qr}$ since $U\in\UBB(n)$. Therefore, the operators $\Omega_{j,k}$ belong to the commutant $\CC^{\rm PP}_{t,n}$.

Next, we introduce the dressed fermionic operators
\begin{align}
&\widetilde a^{(j)}_p = \Gamma_1 \otimes \Gamma_2 \otimes \cdots \otimes \Gamma_{j-1} \otimes a^{(j)}_p \otimes \mathbbm{1} \otimes \mathbbm{1} \otimes \cdots \,, \nonumber \\
&\widetilde a^{\dagger (j)}_p = \Gamma_1 \otimes \Gamma_2 \otimes \cdots \otimes \Gamma_{j-1} \otimes a^{\dagger (j)}_p \otimes \mathbbm{1} \otimes \mathbbm{1} \otimes \cdots \, .
\end{align}
They satisfy the canonical anticommutation relations
\begin{equation} \label{eq-ap:Omega-comm}
\left\{\widetilde a^{(j)}_p,\, \widetilde a^{\dagger (k)}_q\right\} = \delta_{pq}\delta_{jk}, \qquad \left\{\widetilde a^{(j)}_p,\,\widetilde a^{(k)}_q\right\} = \left\{\widetilde a^{\dagger (j)}_p,\,\widetilde a^{\dagger (k)}_q\right\} = 0\,.
\end{equation}
Hence, besides the usual anticommutation relations of distinct fermionic operators acting on the same copy, the $\widetilde a^{(j)}_p$ and $\widetilde a^{\dagger (j)}_p$ operators anti-commute across different copies of $\HC^{\otimes t}$. To see this, let us first assume that $j \neq k$, so that without loss of generality we can choose $k > j$. Recall that $\left\{a^{(j)}_p,\Gamma_j\right\}=0$ and $\left\{a^{\dagger (j)}_p,\Gamma_j\right\}=0$, and we arrive at
\begin{equation}
\left\{\widetilde a^{(j)}_p,\,\widetilde a^{(k)}_q\right\}=0, \qquad \left\{\widetilde a^{(j)}_p,\,\widetilde a^{\dagger (k)}_q\right\}=0\,.
\end{equation}
Moreover, when $j=k$ we simply recover the usual Clifford algebra relations on the $j$-th copy,
\begin{equation}
\left\{\widetilde a^{(j)}_p,\widetilde a^{\dagger (j)}_q\right\} = \left\{a^{(j)}_p,a^{\dagger (j)} _q\right\} = \delta_{pq}, \qquad \left\{\widetilde a^{(j)}_p,\widetilde a^{(j)}_q\right\} = \left\{\widetilde a^{\dagger (j)}_p,\widetilde a^{\dagger (j)}_q\right\} = 0\,.
\end{equation}
We can now define
\begin{equation}
\widetilde\Omega_{j,k} = \sum_{p=1}^n \widetilde a^{\dagger (j)}_p \widetilde a^{(k)}_p \,.
\end{equation}
Since $\widetilde\Omega_{j,k}$ can be generated from $\Omega_{j,k}$ by simply multiplying by parity operators, and each parity operator $\Gamma_r$ commutes with PP fermionic Gaussian unitaries, it follows that, for all PP Gaussian unitaries,
\begin{equation}
\left[R(U)^{\otimes t},\widetilde\Omega_{j,k}\right] = 0  \,.
\end{equation}
Hence, the generators $\widetilde\Omega_{j,k}$ belong to the commutant $\CC^{\rm PP}_{t,n}$.
\end{proof}

\medskip

Next, let us derive Lemma~\ref{lem:Omega-rep-gl-AP}, which we restate for convenience.

\begin{lemma}\label{lem:Omega-rep-gl-AP}
The operators $\widetilde\Omega_{j,k}$ satisfy the Lie algebra commutation relations
\begin{equation}  \label{eq-lem:Omega-rep-gl-AP}
\left[\widetilde{\Omega}_{j,k}, \widetilde{\Omega}_{j',k'}\right] = \delta_{kj'}\widetilde\Omega_{j,k'} - \delta_{jk'}\widetilde\Omega_{j',k}\,,
\end{equation}
and therefore generate a representation of the $\mathfrak{u}(t)$ algebra.
\end{lemma}

\begin{proof}
We simply need to compute the commutation relations of the $\widetilde\Omega_{j,k}$, and compare them to the commutation relations of the canonical basis $E_{j,k}$ of $\mathfrak{u}(t)$, namely
\begin{equation}
[E_{j,k},E_{j',k'}] = \delta_{kj'}E_{j,k'} - \delta_{jk'}E_{j',k}\,.
\end{equation}
Using Eq.~\eqref{eq-ap:Omega-comm}, we obtain
\begin{align}
\left[\widetilde a^{\dagger (j)}_p \widetilde a^{(k)}_p,
 \widetilde a^{\dagger (j')}_q \widetilde a^{(k')}_q\right]
&=
\widetilde a^{\dagger (j)}_p
\{\widetilde a^{(k)}_p,\widetilde a^{\dagger (j')}_q\}
\widetilde a^{(k')}_q
-
\widetilde a^{\dagger (j')}_q
\{\widetilde a^{(k')}_q,\widetilde a^{\dagger (j)}_p\}
\widetilde a^{(k)}_p
\nonumber\\
&=
\delta_{pq}\delta_{kj'}
\widetilde a^{\dagger (j)}_p \widetilde a^{(k')}_q
-
\delta_{pq}\delta_{jk'}
\widetilde a^{\dagger (j')}_q \widetilde a^{(k)}_p .
\end{align}
Summing over $p$ and $q$ then yields
\begin{equation}
\left[\widetilde\Omega_{j,k},\widetilde\Omega_{j',k'}\right]
=
\delta_{kj'}\widetilde\Omega_{j,k'}
-
\delta_{jk'}\widetilde\Omega_{j',k},
\end{equation}
which are exactly the commutation relations of $\mathfrak{u}(t)$.
\end{proof}

 \medskip

After proving Lemmas~\ref{lem:Omega-in-C-AP} and~\ref{lem:Omega-rep-gl-AP}, we now derive Theorem~\ref{th-ap:PP-generators}.

\begin{theorem}\label{th-ap:PP-generators}
    The commutant $\CC^{\rm PP}_{t,n}$, is generated as
    \begin{equation}\nonumber
\CC^{\rm PP}_{t,n}
=
{\rm span}_{\mathbb C}
\Big\langle
\widetilde\Omega_{j,j+1},
\widetilde\Omega_{j+1,j},
\widetilde\Omega_{1,1}
\Big\rangle \,,\;\text{with } 1\leq j\leq t-1\,.
\end{equation}
\end{theorem}

\begin{proof}

Since the commutation relations in Lemma~\ref{lem:Omega-rep-gl-AP} are isomorphic to those of the canonical basis of $\mathfrak{u}(t)$,
the operators $\widetilde\Omega_{j,k}$ realize a representation $\phi$ of $\mathfrak{u}(t)$. Once this connection is
established, we can invoke the skew Howe duality~\cite{aboumrad2022skew}, which plays for the pair $(\U(n),\U(t))$ the role that the Schur-Weyl duality plays for the pair $(\U(d),S_t)$. In the present setting, this duality asserts that the commutant of the diagonal action of $\U(n)$ on the fermionic Fock space is the image of the universal enveloping
algebra of $\mathfrak{u}(t)$ under the representation generated by the $\widetilde\Omega_{j,k}$. Let us here recall the definition of a universal enveloping algebra.

\begin{supdefinition}[Universal enveloping algebra]\label{def-ap:universal-algebra}
  Let $\g$ be a matrix Lie algebra. The universal enveloping algebra of $\g$, denoted $\mathscr{U}(\g)$, is the unital associative algebra generated by $\g$ subject to the relations $XY-YX=[X,Y]$ $\,\forall X,Y\in\g $.
\end{supdefinition}

Importantly, to better understand Supplemental Definition~\ref{def-ap:universal-algebra}, we can emphasize that once we choose a Lie algebra representation $\varphi:\g\to\mathrm{End}(V)$, it extends uniquely to an algebra homomorphism $\varphi:\mathscr{U}(\g)\to\mathrm{End}(V)$~\cite{hall2013lie}. In particular, if $\varphi(\g)$ is generated by the operators $\widetilde{\Omega}_{j,k}$, then
\begin{equation} \label{eq-ap:universal-algebra}
    \varphi\left(\mathscr{U}(\g)\right) = \Span_{\CBB} \left\langle \widetilde{\Omega}_{j,k} \right\rangle \,,
\end{equation}
is the image of the corresponding representation of the unital associative algebra generated by these operators. That is, the vector space obtained by allowing linear combinations over the complex numbers of finite products of the $\widetilde{\Omega}_{jk}$ generators, together with the identity matrix. Therefore,
\begin{equation}
\CC^{\rm PP}_{t,n} = \phi\left(\mathscr{U}(\mathfrak{u}(t))\right) = {\rm span}_{\mathbb C} \Big\langle\widetilde\Omega_{j,k} \Big\rangle\,,
\qquad 1 \leq j,k \leq t\,.
\end{equation}

Next, let us show that one may reduce to a smaller generating set. First, from Eq.~\eqref{eq-lem:Omega-rep-gl-AP} one finds
\begin{equation}
\left[\widetilde\Omega_{j,j+1},\widetilde\Omega_{j+1,j}\right] = \widetilde\Omega_{j,j} - \widetilde\Omega_{j+1,j+1}\,,
\end{equation}
so once one diagonal operator is available, say $\widetilde\Omega_{1,1}$, all diagonal operators can be generated recursively. In addition, for distinct indices one has
\begin{equation}
\left[\widetilde\Omega_{j,k},\widetilde\Omega_{k,\ell}\right] = \widetilde\Omega_{j,\ell}\,,
\end{equation}
which allows one to generate all off-diagonal operators from nearest-neighboring ones. Thus, we conclude that
\begin{equation}
\CC^{\rm PP}_{t,n} = {\rm span}_{\mathbb C} \Big\langle \widetilde\Omega_{j,j+1}\,,
\widetilde\Omega_{j+1,j}\,,
\widetilde\Omega_{1,1}
\Big\rangle \,.
\end{equation}

\end{proof}

\section{Proof of Theorem 2}
\label{ap:PP-dim}

In this section we present a proof for Theorem~\ref{th:PP-dimension}, which we here recall for convenience
\begin{theorem}[PP commutant dimensions]
\label{ap-th:PP-dimension}
    The dimension of the $t$-th order commutant of the $n$-qubit PP matchgate group is given by
    \begin{equation}
        \dim \CC^{\rm PP}_{t,n}=\prod_{j=0}^{n-1}\frac{j!\,(j+2t)!}{(j+t)!^2}\,.
    \end{equation}
\end{theorem}

\begin{proof}

    To prove the claimed formula for the dimension of $\CC^{\rm PP}_{t,n}$ we will first show how $\dim(\CC^{\rm PP}_{t,n})$ maps to a group integral of the squared norm of character of the representation $\HC^{\otimes t}$ of $G=\U(n)$, and then explicitly carry out said integral.

    We begin by recalling the basic representation-theoretic result known as Maschke's theorem~\cite{serre1977linear,fulton1991representation}, or isotypic decomposition. Given a compact group $G$, acting via some unitary representation $R$ on a vector space $V$, Maschke's theorem states that there exists a basis in which $R(g\in G)$ and $V$ are maximally block-diagonalized
    \begin{equation}
    \label{ap-eq:block-diagonalization}
        R(g) \simeq \bigoplus_{\lambda} I_{m_\lambda} \otimes r_{\lambda}(g)\,, \quad V \simeq \bigoplus_\lambda \mathbb{C}^{m_\lambda} \otimes V_{\lambda}\,.
    \end{equation}
    Above, $r_\lambda:G\rightarrow \mathbb{U}(V_\lambda)$ is an irreducible representation, usually abbreviated as irrep, of $G$ labeled by $\lambda$, and $V_\lambda$ the associated irreducible $G$-module. The integer $m_\lambda$ is called the multiplicity of the irrep labeled by $\lambda$, and it counts the number of times each irrep $r_\lambda$ appears in $R$, as well as the number of occurrences of $V_\lambda$ in $V$.

    Now consider the space of linear operators acting on $V$, i.e. $\End(V)\simeq V\otimes V^*$. The representation $R$ naturally extends to $\End(V)$ as its adjoint, and the ensuing isotypic decomposition of $\End(V)$ reads
    \begin{equation}
    \label{ap-eq:isotypic_LV}
        \End(V)\cong V\otimes V^*\cong
        \left(\bigoplus_\lambda \mathbb{C}^{m_\lambda} \otimes V_\lambda\right)\otimes
        \left(\bigoplus_\mu \mathbb{C}^{m_\mu} \otimes \left(V_\mu\right)^*\right)
        \cong \bigoplus_{\lambda,\mu} \mathbb{C}^{m_\lambda m_\mu}\otimes(V^\lambda\otimes \left(V^\mu\right)^*)\,.
    \end{equation}
    By definition, the commutant $\End_G(V)$ is the subspace of operators in $\End(V)$ that commute with $R$, equivalently that are left unchanged by its induced adjoint action. It follows from the previous decomposition that $\dim(\End_G(V))$ is equal to the multiplicity of the trivial representation of $\End(V)$ appearing in its isotypical decomposition~\eqref{ap-eq:isotypic_LV}.
    Given the compactness of $G$ and the finite dimensionality of $V$, Schur's lemma~\cite{fulton1991representation} gives that $V_\lambda\otimes \left(V_\mu\right)^*$ contains, with multiplicity one, the trivial irrep if and only if $\lambda=\mu$. Hence, the next lemma immediately follows.
    \begin{lemma}[Commutant's dimension as a sum of squared multiplicities]
    \label{ap-lem:dim_comm_sum_squares_mult} The dimension of the commutant of a group's representation acting on $V$, for the group G, is given by
        \begin{equation}
            \dim(\End_G(V)) =\sum_{\lambda} m_\lambda^2\,.
        \end{equation}
    \end{lemma}
    Now we can relate the above formula to the characters of the representation.
    \begin{lemma}[Character integral equals the sum of squared multiplicities]
    \label{ap-lem:char_integral_sum_squares}
        Let $d\mu$ be the unique normalized Haar measure of $G$, and let $\chi_V(g)=\Tr[R(g)]$ denote the character of the representation $R$ of $G$. Then
        \begin{equation}
            \int_G |\chi_V(g)|^2 d\mu(g)=\sum_{\lambda} m_\lambda^2\,.
        \end{equation}
    \end{lemma}
    \begin{proof}
        Using the decomposition $\chi_V=\sum_\lambda m_\lambda \chi_\lambda$ and the orthonormality of irreducible characters in $L^2(G)$~\cite{fulton1991representation},
        \begin{equation}
            \int_G |\chi_V|^2 d\mu(g)=\int_G \left(\sum_\lambda m_\lambda \chi_\lambda\right)\left(\sum_\mu m_\mu {\chi_\mu^*}\right)d\mu(g)
            =\sum_{\lambda,\mu} m_\lambda m_\mu \int_G \chi_\lambda {\chi_\mu^*} d\mu(g)
            =\sum_\lambda m_\lambda^2\,.
        \end{equation}
    \end{proof}

    Combining Lemmas~\ref{ap-lem:dim_comm_sum_squares_mult} and~\ref{ap-lem:char_integral_sum_squares} we get to 
    \begin{lemma}[Commutant's dimension as character integral]
    \label{ap-lem:comm_size_character_int}
        \begin{equation}
            \dim(\End_G(V)) = \int_G |\chi_V|^2 d\mu(g)\,.
        \end{equation}
    \end{lemma}

    We now apply this result to our case of interest $G=\U(n)$, $V=\HC^{\otimes t}$. 
    Since traces, and hence characters, factorize under tensor products we get $\chi_{\HC^{\otimes t}}(g)=\chi_\HC(g)^t$ thus
    \begin{equation}
    \label{ap-eq:comm_dim_int}
        \dim \mathcal{C}^{\rm PP}_{t,n}
        =\int_{\U(n)} |\chi_\HC(g)|^{2t}\,d\mu(g)\,. 
    \end{equation}
    We can explicitly carry out the integral in Eq.~\eqref{ap-eq:comm_dim_int} by resorting to the Weyl integration formula~\cite{brocker2003representations}. The latter states that, for a compact connected Lie group $G$ with maximal torus $T$ and Weyl group $W$, the group average of any class function $f$, i.e. a function invariant under conjugation of its argument by the group, can be written as
    \begin{equation}
    \label{ap-eq:weyl_integral}
        \int_G f(g) d\mu(g)=\frac{1}{|W|}\int_T f(\tau)|\delta(\tau)|^2d\mu(\tau)\,,
    \end{equation}
    where $\delta(t)$ is the Weyl denominator. Since the character is a class function, we can apply this formula to compute Eq.~\eqref{ap-eq:comm_dim_int}.
    Here, a maximal torus $T$ of a Lie group $G$ is its maximal compact, connected, abelian Lie subgroup; the Weyl group $W$ of $G$ is defined as $W(G) = N_G(T)/T$, for $N_G(T)$ the normalizer of $T$ in $G$, i.e. the subgroup of $G$ that sends $T$ to itself under group conjugation minus $T$ itself; and the Weyl denominator, sometimes also called the Weyl determinant factor, is
    \begin{equation}
    \label{ap-eq:weyl-denominator-general}
        \delta(\tau)
        =
        \prod_{\alpha > 0}
        \left(e^{\frac{\alpha(\tau)}{2}}-e^{-\frac{\alpha(\tau)}{2}}\right)\,,
    \end{equation}
    for $\alpha(\tau)$ the roots of $G$ relative to $T$, the product running only over the positive roots.
    We are now ready to explicitly evaluate all the needed ingredients. 
    However, let us first recall that the group of PP fermionic Gaussian unitaries acts naturally on 
    \begin{equation}
        V=\Lambda(\mathbb{C}^n)=\bigoplus_{r=0}^n \Lambda^r(\mathbb{C}^n)\cong\HC\,,
    \end{equation}
    that is, on the Fock space of $n$ fermionic modes, which decomposes in the sum of the alternating fixed particle number sectors.
    If $U:\mathbb{C}^n\rightarrow \mathbb{C}^n$ is the standard representation of $\U(n)$, then its induced action on the $k$-th exterior power of $\mathbb{C}^n$ is
    \begin{align}
        &\Lambda^k(U): \Lambda^k(\mathbb{C}^n) \rightarrow \Lambda^r(\mathbb{C}^n)\\
        &\Lambda^k(U)(v_1\wedge \dots \wedge v_k)= (Uv_1\wedge \dots \wedge Uv_k)\,.
    \end{align}
    Hence, on the whole Fock space we have the representation
    \begin{align}
        &R(U):\HC\rightarrow \HC\\
        &R(U)=\bigoplus_{k=0}^n \Lambda^k(U)\,,
    \end{align}
    which is sometimes referred to as the second quantization of $U$.
    
    We now identify the maximal torus $T$ of $G=\U(n)$. For a connected compact Lie group $G$ with associated Lie algebra $\mathfrak{g}$, the maximal torus corresponds to the group obtained by exponentiating the Cartan subalgebra $\mathfrak{t}\subseteq \mathfrak{g}$, the latter being the maximal abelian subalgebra contained in $\mathfrak{g}$.
    Applying this to $G=\U(n)$ one readily finds that its maximal torus $T$ is given by the group of diagonal unitary matrices $T=\{\tau={\rm diag}(e^{i\theta_1},\dots,e^{i\theta_n})\,|\, (\theta_1,\dots,\theta_n)\in [-\pi,\pi]^n\}$, which can be recovered from the fact that a convenient choice for the Cartan $\mathfrak{t}\subset\mathfrak{u}(n)$ is given by the $n\times n$ skew-hermitian matrices $\{d_j\,|\, (d_j)_{l,m} = i \delta_{l,j}\delta_{m,j}\}$.   
    Next, consider the character $\chi_\HC(U)=\Tr[R(U)]$. Choose a basis of the single-particle space $\mathbb{C}^n$ where $U$ is diagonal, then $U={\rm diag}(\eta_1,\dots,\eta_n)$ for $\eta_i=e^{i\theta_i}$ and $\theta_i\in[-\pi,\pi]$. By explicit calculation one finds
    \begin{align}
        \chi_\HC(U)&=\Tr[R(U)]=\sum_{k=0}^n\Tr[\Lambda^k(U)]\\
        &=1+\sum_{i=1}^n\eta_i+\sum_{1\leq i_1 < i_2 \leq n} \eta_{i_1}\eta_{i_2} + \dots + \sum_{1\leq i_1 <\dots < i_n \leq n} \eta_{i_1}\dots\eta_{i_n}\\
        &=\prod_{i=1}^n(1+\eta_i) = \det(\id + U)\,.
    \end{align}
    Since when the elements of the torus $T$, are diagonal matrices, the character restricted to $T$ easily reads
    \begin{equation}
    \label{ap-eq:pp_torus_char}
        \chi_\HC(\tau) = \prod_{j=1}^n(1+e^{i\theta_j}) \qquad \tau={\rm diag}(e^{i\theta_1},\dots,e^{i\theta_n})\,.
    \end{equation}
    We now turn our attention to the Weyl determinant $\delta(\tau)$. 
    The positive roots $\alpha>0$ of $\mathfrak{u}(n)$ are
    \begin{equation}
        \Phi^+ = \{e_i-e_j\,|\, 1\leq i<j \leq n\}\,
    \end{equation}
    after the usual identification of the Cartan $\mathfrak{t}$ with an $n$-dimensional euclidean space with canonical basis $\{e_j\}_{j=1}^n$.
    Then, for any $\tau={\rm diag}(e^{i\theta_1},\dots,e^{i\theta_n})$ one gets
    \begin{align}
    \label{ap-eq:pp_weyl_determ}
        |\delta(\tau)|^2&=\prod_{1\leq l<m \leq n}|e^{i\frac{\theta_l-\theta_m}{2}}-e^{-i\frac{\theta_l-\theta_m}{2}}|^2\\
        &=\prod_{1\leq l<m \leq n}|e^{-i\frac{\theta_m}{2}}\left(e^{i\theta_l}-e^{i\theta_m}\right)e^{-i\frac{\theta_l}{2}}|^2\\
        &=\prod_{1\leq l<m \leq n}|e^{i\theta_l}-e^{i\theta_m}|^2\,.
    \end{align}
    Lastly, the Weyl group of $\U(n)$ is $W\cong S_n$~\cite{goodman2009symmetry}, and hence has size $|W|=n!$.
    Furthermore, in the chosen parametrization, and since it is Abelian, the Haar measure on $T$ is $d\mu(\tau)=\prod_{j=1}^n \frac{d\theta_j}{2\pi}$.
    Substituting Eqs.~\eqref{ap-eq:pp_torus_char} and~\eqref{ap-eq:pp_weyl_determ} into Eq.~\eqref{ap-eq:weyl_integral} together with $|W|$ and $d\mu(\tau)$ we get to
    \begin{align}
    \label{ap-eq:dim-pp-integral}
        \dim(\CC_{t,n}^{\rm PP})&=\frac{1}{(2\pi)^n n!}\int_{[-\pi,\pi]^n} \prod_{j=1}^n|1+e^{i\theta_j}|^{2t} \prod_{1\leq l<m \leq n}|e^{i\theta_l}-e^{i\theta_m}|^2 \prod_{j=1}^n d\theta_j \\
        &=\frac{1}{(2\pi)^n n!}\int_{[-\pi,\pi]^n} \prod_{j=1}^n|1-e^{i\theta_j}|^{2t} \prod_{1\leq l<m \leq n}|e^{i\theta_l}-e^{i\theta_m}|^2 \prod_{j=1}^n d\theta_j \,
    \end{align}
    where in the last line we performed the change of variables $\theta_j\to\theta_j+\pi$ and noticed the periodicity of the integral to keep the integration domain fixed at $[-\pi,\pi]^n$.
    Written in this form, the integral is exactly the special Selberg integral
    \begin{align}
    \label{ap-eq:special-selberg-integral}
        M_n(a,b,\gamma)&=\frac{1}{(2\pi)^n}\int_{[-\pi,\pi]^n} \prod_{j=1}^n e^{i\frac{\theta_j(a-b)}{2}}|1-e^{i\theta_j}|^{a+b} \prod_{1\leq l<m \leq n}|e^{i\theta_l}-e^{i\theta_m}|^{2\gamma} \prod_{j=1}^n d\theta_j \\
        &=\prod_{j=0}^{n-1}\frac{\Gamma(1+a+b+j\gamma)\Gamma(1+(j+1)\gamma)}{\Gamma(1+a+j\gamma)\Gamma(1+b+j\gamma)\Gamma(1+\gamma)}\,
    \end{align}
    where the last line holds for $a,b,\gamma\in\mathbb C$ such that
    \begin{equation}
    \label{ap-eq:special-selberg-conditions}
        \Re(a+b+1)>0,\quad \Re(\gamma)>-\min\{1/n,\Re(a+b+1)/(n-1)\}\,.
    \end{equation}
    Recognizing that Eq.~\eqref{ap-eq:dim-pp-integral} is proportional to Eq.~\eqref{ap-eq:special-selberg-integral} for $a=b=t$ and $\gamma=1$, which satisfy the conditions in Eq.~\eqref{ap-eq:special-selberg-conditions}, we find the claimed expression
    \begin{align}
        \dim(\CC_{t,n}^{\rm PP})&=\frac{1}{n!}M_n(t,t,1)=\frac{1}{n!}\prod_{j=0}^{n-1}\frac{\Gamma(1+2t+j)\Gamma(1+(j+1))}{\Gamma(1+t+j)\Gamma(1+t+j)\Gamma(2)} \\
        &=\prod_{j=0}^{n-1} \frac{j!(j+2t)!}{(j+t)!^2}\,,\label{eq:dim-passive-proof}
    \end{align}
    where in the last line we have used the standard property of the gamma function over the integers
    $\Gamma(1+m)=m!$.
\end{proof}

To finish, we also provide an alternative equation for the dimension of the commutant. We begin by recalling the Weyl dimension formula. Given a semi-simple Lie algebra $\mathfrak{g}$ with positive roots $\Delta_+$ and Weyl vector $\omega=\frac{1}{2}\sum_{\alpha \in\Delta_+}\alpha$, and an irreducible highest-weight module $L(\Gamma)$ of dominant highest-weight $\Gamma$, then 
\begin{equation}
    \dim(L(\Gamma))=\prod_{\alpha\in\Delta_+}\frac{(\Gamma+\omega,\alpha)}{(\omega,\alpha)}\,.
\end{equation}
For $\mathbb{U}(k)$ the positive roots are $\Delta_+=\{e_i-e_j\,|\,1\leq i<j\leq k\}$, where $e_1,\ldots,e_k$ form an orthonormal basis of $\mathbb{R}^k$, and we can express $\omega=\frac{1}{2}\sum_{i\leq j}(e_i-e_j)$. Then, given a $\Gamma=(\lambda_1,\ldots,\lambda_k)$ and a positive root $\alpha=e_i-e_j$ one obtains
\begin{equation}
    (\Gamma+\omega,\alpha)=\lambda_i+\omega_i-\lambda_j-\omega_j\,,
\end{equation}
where we defined $\omega_i=(e_i,\omega)$. Then, since $\omega_i-\omega_j=j-1$ we obtain
\begin{equation}
    \frac{(\Gamma+\omega,\alpha)}{(\omega,\alpha)}=\frac{\lambda_i-\lambda_j+j-i}{j-1}\,,
\end{equation}
and hence 
\begin{equation}
    \dim(V_{\lambda}^{\mathbb{U}(t)})=\prod_{1\leq i<j\leq t}\frac{\lambda_i-\lambda_j+j-i}{j-1}\,.
\end{equation}
Putting it all together we obtain an alternative formula for the commutant dimension
\begin{equation}
    \dim\left(\End_{\mathbb{U}(n)}\left(\Lambda(\mathbb{C}^n\otimes \mathbb{C}^t)\right)\right)=\sum_{\lambda\subseteq n^t}\left(\prod_{1\leq i<j\leq t}\frac{\lambda_i-\lambda_j+j-i}{j-1}\right)^2= 
\prod_{i=1}^t\prod_{j=1}^t
\frac{n+i+j-1}{i+j-1}\,.
\end{equation}

\section{Scaling of the dimension of the commutant $\CC_{t,n}^{\rm PP}$}
\label{ap:scaling-dim-PP}

After obtaining the formula for $\dim\left(\CC_{t,n}^{\rm PP}\right)$ in Theorem~\ref{ap-th:PP-dimension}, namely,
\begin{equation} \label{eq-ap:dim-PP-com2}
\dim \left(C^{\mathrm{PP}}_{t,n}\right) =\prod_{j=0}^{n-1}\frac{j!(j+2t)!}{(j+t)!^2}\,,
\end{equation}
we study its scaling. We begin with fixed $t$ and scaling $n$. The first step is to  rewrite the factor
\begin{equation}
    \frac{j!(j+2t)!}{(j+t)!^2}\,.
\end{equation}
We will do so by using the elementary identity 
\begin{equation}
(m+l)! =  m!\,(m+1)\cdots (m+l) = m!\prod_{s=1}^{l}(m+s)\,,
\end{equation}
 We find
 \begin{equation}
    (j+2t)! = j!\,\prod_{s=1}^{2t}(j+s)\,, \qquad (j+t)! = j!\,\prod_{s=1}^{t}(j+s)\,.
\end{equation}
And so we arrive at
\begin{align}
    \frac{j!(j+2t)!}{(j+t)!^2}  = \frac{\displaystyle j!\,j!\,\prod_{s=1}^{2t}(j+s)}
{\displaystyle j!^2\left(\prod_{s=1}^{t}(j+s)\right)^2} = \frac{\prod_{s=1}^{2t}(j+s)} {\left(\prod_{s=1}^{t}(j+s)\right)^2}= \frac{\left(\prod_{s=1}^{t}(j+s)\right) \left(\prod_{s=t+1}^{2t}(j+s)\right)}{\left(\prod_{s=1}^{t}(j+s)\right)^2}\,.
\end{align}
Hence
\begin{equation}
    \frac{j!(j+2t)!}{(j+t)!^2} = \frac{\prod_{s=t+1}^{2t}(j+s)}{\prod_{s=1}^{t}(j+s)}\,.
\end{equation}
Reindexing the numerator by writing $s=t+k$ with $k=1,\dots,t$, we get
\begin{equation}
\frac{j!(j+2t)!}{(j+t)!^2} = \prod_{k=1}^{t}\frac{j+t+k}{j+k}.
\label{eq:single-factor-rewrite}
\end{equation}
Substituting \eqref{eq:single-factor-rewrite} into \eqref{eq-ap:dim-PP-com2} and exchanging the order of the products then gives
\begin{equation}
    \dim \left(C^{\mathrm{PP}}_{t,n}\right) = \prod_{k=1}^{t}\prod_{j=0}^{n-1}\frac{j+t+k}{j+k}\,.
\end{equation}
We now evaluate the product explicitly,
\begin{equation}
    \prod_{j=0}^{n-1}(j+t+k) = (t+k)(t+k+1)\cdots(n+t+k-1) = \frac{(n+t+k-1)!}{(t+k-1)!},
\end{equation}
and similarly,
\begin{equation}
    \prod_{j=0}^{n-1}(j+k) = k(k+1)\cdots(n+k-1) = \frac{(n+k-1)!}{(k-1)!}\,.
\end{equation}
Therefore,
\begin{equation}
     \dim\left(C^{\mathrm{PP}}_{t,n}\right) = \prod_{k=1}^{t}\prod_{j=0}^{n-1}\frac{j+t+k}{j+k} = \prod_{k=1}^{t} \frac{(n+t+k-1)!}{(t+k-1)!}\,
\frac{(k-1)!}{(n+k-1)!}\,.
\end{equation}
Finally, for each fixed $k$, the factor
\begin{equation}
    \frac{(n+t+k-1)!(k-1)!}{(n+k-1)!(t+k-1)!}
=
\frac{(k-1)!}{(t+k-1)!}\,
\prod_{l=1}^{t}(n+k+l)
\end{equation}
is a polynomial in $n$ of degree $t$, with leading coefficient $\frac{(k-1)!}{(t+k-1)!}$. Since there are $t$ such factors, their product is a polynomial in $n$ of degree $t^2$ and its leading coefficient is $\prod_{k=1}^{t}\frac{(k-1)!}{(t+k-1)!}$.

\medskip

As a sanity check, for the first few values of $t$ one finds
\begin{align}
    t=1:&\qquad \dim\left(\CC^{\mathrm{PP}}_{1,n}\right)=n+1\,,\nonumber \\
t=2:\qquad \dim&\left(\CC^{\mathrm{PP}}_{2,n}\right)
=\frac{(n+1)(n+2)^2(n+3)}{12}\,, \nonumber \\
t=3:\qquad \dim\left(\CC^{\mathrm{PP}}_{3,n}\right)&
=\frac{(n+1)(n+2)^2(n+3)^3(n+4)^2(n+5)}{8640}.
\end{align}
These formulas are all consistent with the general scaling $\dim\left(\CC^{\mathrm{PP}}_{t,n}\right)\in\Theta\left(t^2\right)$ for fixed $t$.

Next, we study the scaling for fixed $n$ and scaling $k$. Using the same strategy as above, we now rewrite each factor as
\begin{equation}
(j+2t)! = (2t)!\prod_{k=1}^{j}(2t+k)\,,\qquad (j+t)! = t!\prod_{k=1}^{j}(t+k)\,.
\end{equation}
Substituting these expressions into the $j$-th factor yields
\begin{equation}
\frac{j!(j+2t)!}{(j+t)!^2} = j!\,\frac{(2t)!}{(t!)^2}\prod_{k=1}^{j}\frac{2t+k}{(t+k)^2}\,.
\end{equation}
Thus, we find
\begin{equation}
\dim\left(C^{\mathrm{PP}}_{t,n}\right) = \left(\frac{(2t)!}{(t!)^2}\right)^n \left(\prod_{j=0}^{n-1}j!\right) \prod_{j=0}^{n-1}\prod_{k=1}^{j}\frac{2t+k}{(t+k)^2}\,.
\end{equation}
We next note that for a fixed $k \in \{1,\dots,n-1\}$, the factor $\frac{2t+k}{(t+k)^2}$ appears once for each $j$ such that $k \leq j \leq n-1$, hence exactly $n-k$ times. Therefore,
\begin{equation}
\prod_{j=0}^{n-1}\prod_{k=1}^{j}\frac{2t+k}{(t+k)^2} = \prod_{k=1}^{n-1}\left(\frac{2t+k}{(t+k)^2}\right)^{n-k}\,.
\end{equation}
Thus we arrive at the identity
\begin{equation}
\dim\left(C^{\mathrm{PP}}_{t,n}\right) = \binom{2t}{t}^{\,n}\left(\prod_{j=0}^{n-1}j!\right)\prod_{k=1}^{n-1}\left(\frac{2t+k}{(t+k)^2}\right)^{n-k}\,.
\end{equation}

\medskip

Let us now analyze the asymptotics as $t \to \infty$, for $n$ fixed. By Stirling's approximation,
\begin{equation}
\binom{2t}{t} \sim \frac{4^t}{\sqrt{\pi t}}.
\end{equation}
Moreover, for each fixed $k$,
\begin{equation}
\frac{2t+k}{(t+k)^2} = \frac{1}{t}\,\frac{2+k/t}{(1+k/t)^2} \sim \frac{2}{t}.
\end{equation}
Since $n$ is fixed, 
\begin{equation}
\prod_{k=1}^{n-1}\left(\frac{2t+k}{(t+k)^2}\right)^{n-k} \sim \prod_{k=1}^{n-1}\left(\frac{2}{t}\right)^{n-k}.
\end{equation}
Substituting these asymptotics into the exact formula gives
\begin{equation}
\dim\left(C^{\mathrm{PP}}_{t,n}\right) \sim \left(\frac{4^t}{\sqrt{\pi t}}\right)^n \left(\prod_{j=0}^{n-1}j!\right)\prod_{k=1}^{n-1}\left(\frac{2}{t}\right)^{n-k}\,.
\end{equation}
We now simplify the expressions. First,
\begin{equation}
\left(\frac{4^t}{\sqrt{\pi t}}\right)^n = 4^{nt}\,\pi^{-n/2}\,t^{-n/2}\,.
\end{equation}
Second,
\begin{equation}
\prod_{k=1}^{n-1}\left(\frac{2}{t}\right)^{n-k} = 2^{\sum_{k=1}^{n-1}(n-k)}\,t^{-\sum_{k=1}^{n-1}(n-k)}\,.
\end{equation}
Using the identity
\begin{equation}
\sum_{k=1}^{n-1}(n-k) = \sum_{r=1}^{n-1}r = \frac{n(n-1)}{2}\,,
\end{equation}
we obtain
\begin{equation}
\prod_{k=1}^{n-1}\left(\frac{2}{t}\right)^{n-k} = 2^{n(n-1)/2}\,t^{-n(n-1)/2}\,.
\end{equation}
Combining everything yields
\begin{equation}
\dim\!\bigl(C^{\mathrm{PP}}_{t,n}\bigr) \sim \pi^{-n/2}2^{n(n-1)/2}\left(\prod_{j=0}^{n-1}j!\right)4^{nt}\,t^{-n/2-n(n-1)/2}\,.
\end{equation}
Finally, we conclude that
\begin{equation}
\dim\left(C^{\mathrm{PP}}_{t,n}\right) \sim K_n\,4^{nt}\,t^{-n^2/2}\,,
\end{equation}
where
\begin{equation}
K_n = \pi^{-n/2} 2^{n(n-1)/2} \prod_{j=0}^{n-1} j!\,.
\end{equation}

\section{GT method for PP fermionic Gaussian unitaries}
\label{ap:PP_GT_method}

Here we describe in detail how to apply the GT method to construct a basis of the commutant $\CC^{\rm PP}_{t,n}$.

Let us start by recalling that the Fock space of $n$ particles is given by 
\begin{equation}
\label{eq:Fock-space-decomp}
    \FC_n=\Lambda(\mathbb{C}^n)=\bigoplus_{r=1}^n\Lambda^r(\mathbb{C}^n)\,,
\end{equation}
where $\Lambda^r$ is the anti-symmetrizer operator of $r$ copies of the space its acts on. Then, we recall that PP fermionic Gaussian unitaries represent the second quantized representation $R$ of $\mathbb{U}(n)$ which acts over creation and annihilation operators as
\begin{equation}
    R(U)a_i\ad R\ad(U)=\sum_{ij}U_{ji}a_j\ad\,,\quad R(U)a_i R\ad(U)=\sum_{ij}\overline{U}_{ji}a_j\,,\quad U\in\mathbb{U}(n)
\end{equation}
where $\overline{U}_{ji}$ denotes the conjugate of the matrix entry $U_{ji}$. 
When taking $t$-copies of this Fock space $\FC_n^{\otimes t}$, the action of PP Gaussian unitaries becomes $R(U)^{\otimes t}$.  
Now we leverage the following important identification~\cite{brosnan2019notes,cautis2014webs}
\begin{equation}
     \FC_n^{\otimes t}=\Lambda(\mathbb{C}^n\otimes \mathbb{C}^t)\,,
\end{equation}
and we invoke the Howe duality, which states that over $\FC_n^{\otimes t}$, the pair $(\mathbb{U}(n),\mathbb{U}(t))$ admits representations which are mutually commuting. Specifically, $\mathbb{U}(n)$ acts on $\FC_n^{\otimes t}$ through the representation $R^{\otimes t}$, whereas $\mathbb{U}(t)$ acts on $\FC_n^{\otimes t}$ via the representation whose infinitesimal action is generated by the operators $\widetilde\Omega_{j,k}$ defined in the main text.
As such, we know that under their joint action, the $t$ copies of the Fock space decompose as
\begin{equation}
\label{eq:Schur-Weyl-passive}
    \FC_n^{\otimes t}\cong\bigoplus_{\lambda\subseteq n^t} V_{\lambda^T}^{\mathbb{U}(n)}\otimes V_{\lambda}^{\mathbb{U}(t)}\,,
\end{equation}
where the sum runs over partitions, or Young diagrams, $\lambda$ that fit inside an $n\times t$ rectangle. Here we have also defined $\lambda^T$ as the partition obtained by transposing the Young diagram. For convenience, we recall that a partition $\lambda=(\lambda_1,\lambda_2,\cdots , \lambda_k)$ with $\lambda_1\geq\lambda_2\geq \cdots \geq \lambda_k\geq 0$ can always be drawn as a Young diagram by stacking $\lambda_1$ boxes on top of $\lambda_2$ boxes, and so on. For instance
\begin{equation}
    \lambda=(3,1,1)=\ydiagram{3,1}\,,\quad \text{leads to}\quad \lambda^T=(2,1,1)=\ydiagram{2,1,1}\,.
\end{equation}
At this point, we note the interesting fact that given a $\lambda$ (or concomitantly a $\lambda^T$), we can check to which particle number $r$ sector $\Lambda^r(\mathbb{C}^n\otimes \mathbb{C}^t)$ the irrep $V_{\lambda^T}^{\mathbb{U}(n)}\otimes V_{\lambda}^{\mathbb{U}(t)}$ belongs to by simply computing
\begin{equation}
\label{eq:particle-lambda}
    r=|\lambda|=|\lambda^T|\,.
\end{equation}

Importantly, from the decomposition of Eq.~\eqref{eq:Schur-Weyl-passive} we find the commutant as
\begin{equation}
   \CC_{t,n}^{\rm PP}=\End_{\mathbb{U}(n)}\left(\Lambda(\mathbb{C}^n\otimes \mathbb{C}^t)\right)=\bigoplus_{\lambda\subseteq n^t} \End\left(V_{\lambda}^{\mathbb{U}(t)}\right)\,,
\end{equation}
whose associated dimension is
\begin{equation}
   \dim(\CC_{t,n}^{\rm PP})= \dim\left(\End_{\mathbb{U}(n)}\left(\Lambda(\mathbb{C}^n\otimes \mathbb{C}^t)\right)\right)=\sum_{\lambda\subseteq n^t}\dim\left(V_{\lambda}^{\mathbb{U}(t)}\right)^2\,.
\end{equation}

We now turn our attention to constructing a basis of $\CC_{t,n}^{\rm PP}$. We proceed by applying the Gelfand--Tsetlin construction, which for the unitary group is classical and well studied~\cite{gelfand1950finite,zhelobenko1973compact,molev2006gelfand}. 
The sequences of restrictions happen along the group chain
\begin{equation}
    \mathbb U(1)\subset \mathbb U(2)\subset \cdots \subset \mathbb U(t)\,,
\end{equation}
or equivalently along the Lie algebra chain
\begin{equation}
    \mathfrak u(1)\subset \mathfrak u(2)\subset \cdots \subset \mathfrak u(t)\,.
\end{equation}
A GT pattern with top row $\lambda=\lambda^{(t)}$ is particularly a chain
\begin{equation}
    \lambda^{(t)}\succ \lambda^{(t-1)}\succ \cdots \succ \lambda^{(1)}\,,
\end{equation}
where
\begin{equation}
    \lambda^{(m)}=(\lambda_1^{(m)},\dots,\lambda_m^{(m)})
\end{equation}
is a dominant highest-weight of $\mathbb U(m)$, and consecutive rows satisfy the standard unitary interlacing inequalities~\cite{molev2006gelfand}
\begin{equation}
    \lambda_i^{(m)}\ge \lambda_i^{(m-1)}\ge \lambda_{i+1}^{(m)},
    \qquad i=1,\dots,m-1\,.
\end{equation}
Since every branching step $\mathbb U(m)\downarrow \mathbb U(m-1)$ is multiplicity-free~\cite{molev2006gelfand}, every such chain selects a unique basis vector in $V_\lambda^{\mathbb U(t)}$. In particular,
\begin{equation}
    |\GT(\lambda)|=\dim V_\lambda^{\mathbb U(t)}\,,
\end{equation}
where we denoted by $|\GT(\lambda)|$ the number of valid GT patterns inside the irrep $V_\lambda^{\mathbb U(t)}$.

We now explicit how a GT path is represented as an operator.
At the operator level, the dual algebra is generated by the $\widetilde\Omega_{jk}$. A convenient Cartan family is given by
\begin{equation}
    H_p=\widetilde\Omega_{pp}-\widetilde\Omega_{p+1,p+1},
    \qquad p=1,\dots,t-1\,,
\end{equation}
together with the central particle-number operator
\begin{equation}
    M=\sum_{p=1}^t \widetilde\Omega_{pp}\,.
\end{equation}
Correspondingly, the raising and lowering operators are given by
\begin{equation}
    E_p=\widetilde\Omega_{p,p+1},
    \qquad
    F_p=\widetilde\Omega_{p+1,p},
    \qquad
    p=1,\dots,t-1\,.
\end{equation}
With these, starting from a highest-weight vector $\ket{\lambda,{\rm hw}}\in V_\lambda^{\mathbb U(t)}$, every GT vector may be obtained by applying to the latter a canonically ordered product of lowering operators $F_p$. Denoting by $Z_T$ the monomial associated with the pattern $T$, one has
\begin{equation}
    \ket{\lambda,T}\propto Z_T\ket{\lambda,{\rm hw}}\,.
\end{equation}
If $P_\lambda^{\mathrm{hw}}$ denotes the projector onto the highest-weight line, then the corresponding matrix units spanning $\End(V_\lambda^{\U(t)})$ are
\begin{equation}
    X^{(\lambda)}_{T,T'}
    \propto
    Z_T\,P_\lambda^{\mathrm{hw}}\,Z_{T'}^\dagger,
    \qquad
    \lambda\subseteq n^t,
    \quad
    T,T'\in \GT(\lambda)\,.
\end{equation}
Thus, after normalization we may write
\begin{equation}
    X^{(\lambda)}_{T,T'}=\ketbra{\lambda,T}{\lambda,T'}\,.
\end{equation}
Hence, collecting the matrix units from all the irreps $\lambda$ we obtain an orthonormal basis of the PP fermionic Gaussian unitaries commutant as
\begin{equation}
    \CC^{\mathrm{PP}}_{t,n}
    =
    \left\{
        X^{(\lambda)}_{T,T'}
        \,:\,
        \lambda\subseteq n^t,\ T,T'\in\GT(\lambda)
    \right\}\,.
\end{equation}

\section{Construction of the commutant basis for PP Gaussian unitaries}\label{ap:construction-GT-PP}

Before proceeding to a case-by-case analysis of different values of $t$, we find it important to  recall that the operators
\begin{equation}
\widetilde\Omega_{j,k} = \sum_{p=1}^n \widetilde a^{\dagger (j)}_p \widetilde a^{(k)}_p \,.
\end{equation}
generate a representation of the $\mathfrak{u}(t)$ algebra. 

\subsection{$t=1$ case}

When $t=1$ we have $
\FC_n=\Lambda(\mathbb{C}^n\otimes \mathbb{C})\cong\Lambda(\mathbb{C}^n)$ and the only allowed partitions with at most one row and row length at most $n$ are $\lambda=(0),(1),\cdots ,(n)$. Their transpose is given, e.g., by $\lambda^T=(r)^T=(1^r)$, so the decomposition in Eq.~\eqref{eq:Schur-Weyl-passive} becomes
\begin{equation}
    \FC_n^{\otimes 1}\cong\bigoplus_{r=0}^n V_{(1^r)}^{\mathbb{U}(n)}\otimes V_{(r)}^{\mathbb{U}(1)}\,.
\end{equation}
Indeed, since the representations $V_{(r)}^{\mathbb{U}(1)}$ are one-dimensional (as $\mathbb{U}(1)$ is abelian compact group), and since by definition $V_{(1^r)}^{\mathbb{U}(n)}=\Lambda^r(\mathbb{C}^n)$ (recall that a column of boxes represents the antisymmetric representation), then we recover the decomposition of Eq.~\eqref{eq:Fock-space-decomp}. Moreover, since $\dim\left(V_{(r)}^{\mathbb{U}(1)}\right)=1$, the dimension of the commutant is
\begin{equation}
    \dim(\CC_{1,n}^{\rm PP})=\dim\left(\End_{\mathbb{U}(n)}\left(\FC_n^{\otimes 1}\right)\right)=\sum_{r=0}^n1=n+1\,.
\end{equation}

From Eq.~\eqref{eq:dim-t-2-passive} we know that this block is one-dimensional. Moreover, the GT restriction chain for $\mathbb{U}(1)$ is trivial, meaning that the highest-weight for each irrep is just an integer $r$, and the set of GT patterns has exactly one element. Thus, for each particle number $r$ we have  
\begin{equation}
    \GT(r)=\{*\}\,,\quad |\GT(r)|=1\,.
\end{equation}

Then, the infinitesimal generator $\widetilde\Omega_{1,1}$ of the group action can be identified with the number operator 
\begin{equation}
    N=\widetilde\Omega_{1,1}\,,
\end{equation}
so that each $\mathbb{U}(1)$-irreducible sector is an eigenstate of $N$ with eigenvalues $\{0,1,\ldots,n\}$, and therefore is the weight operator at the ``bottom'' of the GT chain. This allows us to then construct the projectors onto the irreps as
\begin{equation}
    P_r=\prod_{\substack{s= 0\\s\neq r}}^n\frac{N-s}{r-s}\,,
\end{equation}
whose action is $P_r:\Lambda(\mathbb{C}^n)\rightarrow\Lambda^r(\mathbb{C}^n)$, and which satisfy $P_rP_r'=\delta_{r,r'}P_r$ and $\sum_{r=0}^nP_r=\id$.

Since there is only one GT pattern $T=T'=\{*\}$ we obtain
\begin{equation}
    X_{*,*}^{(r)}=\ket{r,*}\bra{r,*}=P_r\,.
\end{equation}
As such, the GT construction leads to the orthonormal basis for the $t=1$ commutant
\begin{equation}
    \CC_{1,n}^{\rm PP}={\rm span}_{\mathbb{C}}\{P_0,P_1,\cdots,P_r\}\,.
\end{equation}

At this point we find it important to note that an alternative basis for the commutant can be readily found by noting that 
the number operator $N=\sum_{i=1}^n a_i^\dagger a_i$ was chosen as the the basis element for $\mathfrak{u}(1)$, meaning that $N$ is central. However, so is $N^l$ for $l=1,\ldots,n$. Indeed, we could have defined the projector in term of any of those operators. As a matter of fact, one can simple note that while the collection of operators 
$\{
\mathbbm{1}, N, N^2, \ldots, N^n\}
$
are linearly independent, they can already be used to form a basis of the commutant for $t=1$. That is
\begin{equation}
    \CC_{1,n}^{\rm PP}={\rm span}_{\mathbb{C}}\{
\id, N, N^2, \ldots, N^n\}
\end{equation}
The linear independence is easily proven as assuming that the polynomial $\sum_{l=0}^n c_l N^l$ vanishes implies that its action on a computational basis state also vanishes so that $\sum_{l=0}^n c_l m^l=0$ for all $m=0,1,\dots,n$. This implies $n+1$ distinct roots but the degree of the polynomial is $n$ so we must have $c_l=0$. One can remove the linear dependence by instead considering the basis 

\begin{equation}
    \CC_{1,n}^{\rm PP}={\rm span}_{\mathbb{C}}\{
E_1,\ldots,E_n\}\,,\quad \text{with}\quad E_a=\sum_{|A|=a}\prod_{i\in A} Z_i\,.
\end{equation}
In particular, we can express the first few elements of this bases as  
$
E_0 = \mathbb I$,
$E_1 = \sum_i Z_i=S$,
$E_2 = \sum_{i<j} Z_i Z_j$,
$E_3 = \sum_{i<j<k} Z_i Z_j Z_k\,.$
This basis is related to $N$ through Krawtchouk polynomials
\begin{equation}
E_a=K_a(n-N;n).
\end{equation}
Therefore $(\sum_i Z_i)^l=S^l$ and $N^l$ admit natural Krawtchouk expansions in the symmetric Pauli basis. Notice now that by expanding $N=\sum_r r P_r$ we can rewrite this as an explicit relation between the basis of projectors and $E_a$:
\begin{equation}
    E_a=\sum_r K_a(n-r;n)P_r\,.
\end{equation}
To prove this relation, it is enough to evaluate both sides on a computational basis state
$|z_1,\dots,z_n\rangle$, with $z_i=\pm1$, and let $k$ denote the number of sites with
$z_i=-1$. The operator $N$ counts the number of sites with $z_i=+1$, hence
\begin{equation}
N|z_1,\dots,z_n\rangle=(n-k)|z_1,\dots,z_n\rangle.
\end{equation}
On the other hand, $E_r$ is the sum of all products of $r$ distinct $Z_i$'s, so its
eigenvalue on the same state is
\begin{equation}
\sum_{j=0}^a (-1)^j \binom{k}{j}\binom{n-k}{a-j}.
\end{equation}
Indeed, choosing $j$ factors from the $k$ sites with eigenvalue $-1$ and $a-j$ factors
from the remaining $n-k$ sites with eigenvalue $+1$ produces a contribution $(-1)^j$, and
summing over all such choices gives the above expression. By definition, this is precisely the
Krawtchouk polynomial $K_a(k;n)$. Since $k=n-N$ on computational basis states, we  recover the previous relation. 
Conversely, since the Krawtchouk polynomials form an orthogonal basis for functions on $k=0,1,\dots,n$ with weight $\binom{n}{k}$, one may invert the relation and expand any polynomial in $N$ in the $E_r$ basis. In particular,
\begin{equation}
N^l=\sum_{r=0}^l c_{l,a}\,E_a,
\end{equation}
with coefficients
\begin{equation}
c_{l,a}
=
\frac{1}{2^n\binom{n}{a}}
\sum_{k=0}^n
\binom{n}{k}(n-k)^l K_a(k;n).
\end{equation}

\subsection{$t=2$ case}

Now, the partitions allowed will take the form 
\begin{equation}
    \lambda=(\lambda_1,\lambda_2)\,,\quad n\geq\lambda_1\geq\lambda_2\geq 0\,,\quad \text{such that}\quad\lambda^T=(2^{\lambda_2},1^{\lambda_1-\lambda_2})\,.
\end{equation}
Then, the decomposition of the Fock space will be
\begin{equation}
    \FC_n^{\otimes 2}\cong\bigoplus_{n\geq\lambda_1\geq\lambda_2\geq 0} V_{(2^{\lambda_2},1^{\lambda_1-\lambda_2})}^{\mathbb{U}(n)}\otimes V_{(\lambda_1,\lambda_2)}^{\mathbb{U}(2)}\,.
\end{equation}
Next, we can use the fact that
\begin{equation}\label{eqap:dim-t-2-passive}
    \dim\left(V_{(\lambda_1,\lambda_2)}^{\mathbb{U}(2)}\right)=\lambda_1-\lambda_2+1\,,
\end{equation}
so that 
\begin{equation}
    \dim(\CC_{2,n}^{\rm PP})=\dim\left(\End_{\mathbb{U}(n)}\left(\FC_n^{\otimes 2}\right)\right)=\sum_{n\geq\lambda_1\geq\lambda_2\geq 0}(\lambda_1-\lambda_2+1)^2=\frac{(n+1)(n+2)^2(n+3)}{12}\,.
\end{equation}

Next, we note that the subgroup chain is
\begin{equation}
    \mathbb{U}(1)\subset\mathbb{U}(2)\,,
\end{equation}
and the GT patterns will be a triangular array 
\begin{equation}\label{eqap:GT-pattern-passive-t2}
    \begin{array}{ccc}
       \lambda_1  & & \lambda_2  \\
         & m & 
    \end{array}
\end{equation}
with $\lambda_1\geq m\geq \lambda_2$. The set of GT patterns can thus be defined as
\begin{equation}
    \GT(\lambda_1,\lambda_2)=\{m\in\mathbb{Z}\,|\,\lambda_1\geq m\geq \lambda_2\}\,,
\end{equation}
with size
\begin{equation}
    |\GT(\lambda_1,\lambda_2)|=\lambda_1-\lambda_2+1\,,
\end{equation}
in agreement with Eq.~\eqref{eqap:dim-t-2-passive}.  Importantly, we note that we will pick $\mathbb{U}(1)$ via the standard embedding
\begin{equation}
    z\rightarrow\begin{pmatrix}
        z & 0\\0 & 1
    \end{pmatrix}\,.
\end{equation}
As such, given in the single particle space $\mathbb{C}^n\otimes \mathbb{C}^2$, with basis elements of the form $e_p\otimes f_1$ or $e_p\otimes f_2$, the action of $\mathbb{U}(1)$ is
\begin{equation}
    e_p\otimes f_1\rightarrow z (e_p\otimes f_1)\,,\quad e_p\otimes f_2\rightarrow e_p\otimes f_2\,.
\end{equation}
Thus, a particle in copy 1 contributes one factor of $z$ whereas a  particle in copy 2 contributes no factor. When moving onto the Fock space $\Lambda(\mathbb{C}^n\otimes \mathbb{C}^2)$, we take a wedge basis monomial with $n_1$ particles in copy 1 and $n_2$ in copy 2 of the form
\begin{equation}
    (e_{p_1}\otimes f_1)\wedge\cdots\wedge(e_{p_{n_1}}\otimes f_1)\wedge (e_{p_1}\otimes f_2)\wedge\cdots\wedge(e_{p_{n_1}}\otimes f_2)\,.
\end{equation}
Under $\mathbb{U}(1)$ the previous element transforms as $z^{n_1}$.

Next, we note that for $t=2$ there are four operators $\widetilde\Omega_{j,k}$. Indeed, we find it convenient to  perform the identification  
\begin{equation}
    N_1=\widetilde\Omega_{1,1}\,\quad N_2=\widetilde\Omega_{2,2}\,\quad
\end{equation}
which respectively count the number of fermions on each copy of the Fock state, as well as define the operators
\begin{equation}
    J_+=\widetilde\Omega_{1,2}\,,\quad J_-=\widetilde\Omega_{2,1}\,.
\end{equation}
From the previous choice for $\mathbb{U}(1)$, one can identify the eigenvalues of $N_1$ with the weights, and concomitantly with the GT lower entry $m$. 

\begin{remark}\label{remark-particles}
    Combining the previous fact with Eq.~\eqref{eq:particle-lambda} we thus see that given a pattern as in Eq.~\eqref{eqap:GT-pattern-passive-t2}, then the irrep will belong to a sector with $|\lambda|=\lambda_1+\lambda_2$ particles, where the first copy contains $m$ particles, and the second one $\lambda_1+\lambda_2-m$ particles.
\end{remark}

Interestingly, we can define a total number operator 
\begin{equation}
    N=N_1+N_2\,,
\end{equation}
and an effective spin-$z$ operator
\begin{equation}
    J_z=\frac{N_1-N_2}{2}\,.
\end{equation}
As such, $N$ counts the total number of particles across the two copies, while $J_z$ measures the imbalance between the copies. The ladder operators $J_\pm$ move one particle from one copy to the other:
\begin{equation}
J_+:\ (N_1,N_2)\mapsto (N_1+1,N_2-1),\qquad
J_-:\ (N_1,N_2)\mapsto (N_1-1,N_2+1).
\end{equation}
Hence $M$ is conserved by the $\mathfrak{su}(2)$ action, whereas $J_z$ changes by $\pm1$ under $J_\pm$. Moreover, one can readily see that $J_z,J_\pm$ satisfy the standard $\mathfrak{su}(2)$ commutation relations
\begin{equation}
[J_z,J_\pm]=\pm J_\pm,\qquad
[J_+,J_-]=2J_z,
\end{equation}
while $M$ commutes with all of them. Thus, their combination allows us to decompose the algebra
\begin{equation}
\mathfrak{u}(2)=\mathfrak{su}(2)\oplus \mathfrak{u}(1)    
\end{equation}
Below we will find that this interpretation allow us to provide a nice ``physical'' interpretation to the operators in the commutant in terms of spin representations of $\mathbb{SU}(2)$. 

\subsubsection{Irrep $(\lambda_1,\lambda_2)=(0,0)$}

To construct the GT basis for the commutant, we start by considering the trivial irrep $(\lambda_1,\lambda_2)=(0,0)$, whose only GT pattern
\begin{equation}
    \begin{array}{ccc}
      0 & & 0  \\
         & 0 & 
    \end{array}
\end{equation}
corresponds to $m=0$. Here, the irrep corresponding of $\mathbb{U}(2)$ is the trivial one. Since, we know from Eq.~\eqref{eq:particle-lambda} that this block belongs to the zero particle sector, we can identify the highest-weight state by the two-copy vacuum $\ket{0}$, which means that we can define the projector 
\begin{equation}
    P_{(0,0)}^{\mathrm{hw}}=\ket{0}\bra{0}\,,
\end{equation}
which leads to the single operator in the commutant
\begin{equation}
    X^{(0,0)}_{0,0}=\ket{0}\bra{0}\,.
\end{equation}

\subsubsection{Irrep $(\lambda_1,\lambda_2)=(1,0)$}

The next non-trivial case is obtained for $(\lambda_1,\lambda_2)=(1,0)=\sydiagram{1,0}$, corresponding to the block
\begin{equation}
    V_{(1)}^{\mathbb{U}(n)}\otimes V_{(1,0)}^{\mathbb{U}(2)}\cong\mathbb{C}^n\otimes \mathbb{C}^2\,.
\end{equation}
Here, the multiplicity space is of dimension $n$, whereas the copy-side irrep is of dimension $2$. Explicitly, we have two GT patterns
\begin{equation}
    T_1=\quad\begin{array}{ccc}
      1 & & 0  \\
         & 1 & 
    \end{array}\quad\quad\quad\text{and}\quad\quad\quad     T_2=\quad\begin{array}{ccc}
      1 & & 0  \\
         & 0 & 
    \end{array}\quad.
\end{equation}
By following the convention defined above  we find $T_1$ to be the distinguished highest pattern, i.e., $T_{\mathrm{hw}}=T_1$, as it is the pattern with maximal lower entry. Using Eq.~\eqref{eq:particle-lambda}, and given that $\mathbb{C}^n\otimes \mathbb{C}^2=\Lambda^1(\mathbb{C}^n\otimes \mathbb{C}^2)$, we can see that the block $(1,0)$ corresponds to the one-particle sector. Moreover, Remark~\ref{remark-particles} indicates that in $T_1$ the particle is in the first copy, whereas in $T_2$ the particle is in the second one. As such, we define the natural basis
\begin{equation}
    \ket{p,1}=\widetilde a^{\dagger (1)}_p\ket{0}\,,\quad  \ket{p,2}=\widetilde a^{\dagger (2)}_p\ket{0}\,,
\end{equation}
which satisfy 
\begin{equation}
    N_1\ket{p,1}=\ket{p,1}\,,\quad N_2\ket{p,1}=0\,,\quad N_2\ket{p,1}=0\,,\quad N_2\ket{p,2}=\ket{p,2}\,.
\end{equation}
In particular, using that $J_+\ket{p,1}=0$ and $J_+\ket{p,2}=\ket{p,1}$, we identify $\ket{p,1}$ as the highest-weight states. That is, we can define the highest-weight subspace inside the $(1,0)$ block as 
\begin{equation}
    {\rm span}_{\mathbb{C}}\{\ket{p,1}\,|\,p=1,\ldots,n\}\,.
\end{equation}
This allows us to define the projector
\begin{equation}
    P_{(1,0)}^{\mathrm{hw}}=\sum_{p=1}^n\ket{p,1}\bra{p,1}=\sum_{p=1}^n{\widetilde{a}_p^{\dagger(1)}}\ket{0}\bra{0}\widetilde{a}_p^{(1)}\,,
\end{equation}
which is, as expected, one-dimensional on the copy-side but of rank $n$ on the full Hilbert space side (due to the $n$ dimension of the multiplicity space).

From here, we can obtain the lowering operators as $Z_{T_1}=\id$ (as $T_1$ is the highest pattern), and $Z_{T_2}=J_-$. This latter identification follows from the fact that $J_-$ maps the highest-weight line to the states with zero-eigenvalue of $N_1$, i.e., it moves the single particle in the first copy (indicated by $m=1$ in $T_1$) to the second copy (indicated by $m=0$ in $T_2$). Combing these results leads to the explicit basis operators
\begin{align}
    X^{(1,0)}_{T_1,T_1}=\sum_{p=1}^n{\widetilde{a}_p^{\dagger(1)}}\ket{0}\bra{0}\widetilde{a}_p^{(1)}\,,\quad    X^{(1,0)}_{T_2,T_2}=\sum_{p=1}^n{\widetilde{a}_p^{\dagger(2)}}\ket{0}\bra{0}\widetilde{a}_p^{(2)}\,,\nonumber\\ X^{(1,0)}_{T_1,T_2}=\sum_{p=1}^n{\widetilde{a}_p^{\dagger(1)}}\ket{0}\bra{0}\widetilde{a}_p^{(2)}\,,\quad    X^{(1,0)}_{T_2,T_1}=\sum_{p=1}^n{\widetilde{a}_p^{\dagger(2)}}\ket{0}\bra{0}\widetilde{a}_p^{(1)}\,.
\end{align}

\subsubsection{Irrep $(\lambda_1,\lambda_2)=(1,1)$}

Here we focus on the case $(\lambda_1,\lambda_2)=(1,1)=\sydiagram{1,1}$, leading to the block
\begin{equation}
    V_{(2)}^{\mathbb{U}(n)}\otimes V_{(1,1)}^{\mathbb{U}(2)}\cong{\rm Sym}^2(\mathbb{C}^n)\otimes \Lambda^2(\mathbb{C}^2)\cong{\rm Sym}^2(\mathbb{C}^n)\otimes \det\,.
\end{equation} 
As expected from the last equality above (and from Eq.~\eqref{eq:dim-t-2-passive}) the copy side irrep is one dimensional but with physical multiplicity $\dim({\rm Sym}^2(\mathbb{C}^n))=\frac{n(n+1)}{2}$. Now, the only GT pattern that we have is
\begin{equation}
    T=\quad\begin{array}{ccc}
      1 & & 1  \\
         & 1 & 
    \end{array}\,.
\end{equation}
As we will see next, while the dimension is equal to one, the representation is not trivial. 

From Eq.~\eqref{eq:particle-lambda}, and using the fact that 
\begin{equation}\label{eq:decomp-teo-particles}
    \Lambda^2(\mathbb{C}^n\otimes\mathbb{C}^2)\cong{\rm Sym}^2(\mathbb{C}^n)\otimes \Lambda^2(\mathbb{C}^2)\oplus \Lambda^2(\mathbb{C}^2)\otimes{\rm Sym}^2(\mathbb{C}^n)\,,
\end{equation}
we find that the block $(1,1)$ is the first term of the summand, and thus lives within the two particle sector. Moreover, Remark~\ref{remark-particles} indicates that there is one particle in the first copy, and one particle in the second copy. Hence, a natural basis to consider are two particle states of the general form
\begin{equation}
    \ket{\psi}=\sum_{p,q=1}^n\psi_{pq}{\widetilde{a}_p^{\dagger(1)}}{\widetilde{a}_q^{\dagger(2)}}\ket{0}\,.
\end{equation}
We can find the highest-weight state by applying to $\ket{\psi}$ the raising operator $J_+$, and determining the coefficients $\psi_{pq}$ for which the state is annihilated. In particular, using the relation
\begin{equation}
    J_+{\widetilde{a}_p^{\dagger(1)}}{\widetilde{a}_q^{\dagger(2)}}\ket{0}=-{\widetilde{a}_p^{\dagger(1)}}{\widetilde{a}_q^{\dagger(1)}}\ket{0}\,,
\end{equation}
we find
\begin{equation}
    J_+\ket{\psi}=-\sum_{p,q=1}^n\psi_{pq}{\widetilde{a}_p^{\dagger(1)}}{\widetilde{a}_q^{\dagger(1)}}\ket{0}\,.
\end{equation}
We can therefore see that $J_+$ will annihilate the state if $\psi_{pq}=\psi_{qp}$, indicating that the  highest-weight subspace corresponds to the subspace of symmetric coefficient matrices. This prompts us to define the orthogonal basis for $p<q$
\begin{equation}
    \ket{p,q;+}=\frac{1}{\sqrt{2}}({\widetilde{a}_p^{\dagger(1)}}{\widetilde{a}_q^{\dagger(2)}}+{\widetilde{a}_q^{\dagger(1)}}{\widetilde{a}_p^{\dagger(2)}})\ket{0}\,,
\end{equation}
while for $p=q$
\begin{equation}
    \ket{p,p;+}={\widetilde{a}_p^{\dagger(1)}}{\widetilde{a}_p^{\dagger(2)}}\ket{0}\,.
\end{equation}
There are $\frac{n(n+1)}{2}$ such states, thus exactly matching the dimension of ${\rm Sym}^2(\mathbb{C}^n)$ as expected. Interestingly, one can also check that $J_-\ket{p,q;+}=0$, in compliance with the fact that the space is one-dimensional and that these states are simultaneously highest and lowest-weight states.

The previous allows us to find that the highest subspace is ${\rm span}_{\mathbb{C}}\{\ket{p,q;+}\,|\, 1\leq p< q\leq n\}\cup\{\ket{p,p;+}\,|\, 1\leq p\leq n\}$, and thus define the highest-weight projector
\begin{align}
    P_{(1,1)}^{\mathrm{hw}}&=\sum_{1\leq p\leq q\leq n}\ket{p,q;+}\bra{p,q;+}\nonumber\\
    &=\sum_{1\leq p< q\leq n}\frac{1}{2}({\widetilde{a}_p^{\dagger(1)}}{\widetilde{a}_q^{\dagger(2)}}+{\widetilde{a}_q^{\dagger(1)}}{\widetilde{a}_p^{\dagger(2)}})\ket{0}\bra{0}({\widetilde{a}_q^{(2)}}{\widetilde{a}_p^{(1)}}+{\widetilde{a}_p^{(2)}}{\widetilde{a}_q^{(1)}})+\sum_{p=1}^n{\widetilde{a}_p^{\dagger(1)}}{\widetilde{a}_p^{\dagger(2)}}\ket{0}\bra{0}{\widetilde{a}_p^{(2)}}{\widetilde{a}_p^{(1)}}\,.
\end{align}
Since there is only one GT pattern $Z_T=\id$ and 
\begin{equation}
  X^{(1,1)}_{T,T} =P_{(1,1)}^{\mathrm{hw}}\,.
\end{equation}

\subsubsection{Irrep $(\lambda_1,\lambda_2)=(2,0)$}

As a final example, let us consider 
Here we focus on the case $(\lambda_1,\lambda_2)=(2,0)=\sydiagram{2}$, leading to the block
\begin{equation}
    V_{(1,1)}^{\mathbb{U}(n)}\otimes V_{(2)}^{\mathbb{U}(2)}\cong \Lambda^2(\mathbb{C}^2)\otimes{\rm Sym}^2(\mathbb{C}^n)\,.
\end{equation} 
That is, to the second term in the irrep decomposition of the two particle sector (see  Eq.~\eqref{eq:decomp-teo-particles}). As such, we know that the block has degree $\lambda=2$ and physical multiplicity $\dim(\Lambda^2(\mathbb{C}^2))=\binom{n}{2}$. The three GT patterns are 
\begin{equation}
    T_1=\quad\begin{array}{ccc}
      2 & & 0  \\
         & 2 & 
    \end{array}\quad\quad\quad\text{,}\quad\quad\quad     T_2=\quad\begin{array}{ccc}
      2 & & 0  \\
         & 1 & 
    \end{array}\quad\quad\quad\text{,}\quad\quad\quad     T_3=\quad\begin{array}{ccc}
      2 & & 0  \\
         & 0 & 
    \end{array}\quad,
\end{equation}
meaning that
\begin{equation}
    \GT(2,0)=\{T_1,T_2,T_3\}\,,
\end{equation}
whose dimension $|\GT(2,0)|=3$ matches that of $V_{(2)}^{\mathbb{U}(1)}$. 

From Remark~\ref{remark-particles}, we know that $T_1$ (the distinguished highest pattern) corresponds to both particles in copy 1, $T_2$ to one particle in each copy, and $T_3$ to both particles in the second copy. Hence, for the highest-weight $m=2$ we will consider the basis
\begin{equation}
    \ket{p,q;2}={\widetilde{a}_p^{\dagger(1)}}{\widetilde{a}_q^{\dagger(1)}}\ket{0}\,,\quad p< q\,,
\end{equation}
while for the middle line $m=1$
\begin{equation}
    \ket{p,q;-}=\frac{1}{\sqrt{2}}({\widetilde{a}_p^{\dagger(1)}}{\widetilde{a}_q^{\dagger(2)}}-{\widetilde{a}_q^{\dagger(1)}}{\widetilde{a}_p^{\dagger(2)}})\ket{0}\,,
\end{equation}
and for the lowest GT line $m=0$
\begin{equation}
    \ket{p,q;0}={\widetilde{a}_p^{\dagger(2)}}{\widetilde{a}_q^{\dagger(2)}}\ket{0}\,,\quad p< q\,.
\end{equation}
In all cases, we can verify that since the physical multiplicity is $\Lambda^2(\mathbb{C}^n)$, the indexes must be antisymmetric in $p,q$. One  can also verify that
\begin{equation}
    J_+\ket{p,q;2}=0\,,
\end{equation}
as expected, meaning that the highest-weight subspace will be ${\rm span}_{\mathbb{C}}\{\ket{p,q;2}\,|\, 1\leq p<q\leq n\}$, and the projector onto said subspace is
\begin{align}
    P_{(2,0)}^{\mathrm{hw}}=\sum_{1\leq p<q\leq n}\ket{p,q;2}\bra{p,q;2}=\sum_{1\leq p< q\leq n}{\widetilde{a}_p^{\dagger(1)}}{\widetilde{a}_q^{\dagger(1)}}\ket{0}\bra{0}{\widetilde{a}_q^{(1)}}{\widetilde{a}_p^{(1)}}\,.
\end{align}

Next, we study the action of the lowering operator 
\begin{equation}
    J_-\ket{p,q;2}=\sqrt{2}\ket{p,q;-}\,,\quad J_-\ket{p,q;1}=\sqrt{2}\ket{p,q;0}\,,\quad J_-\ket{p,q;0}=0\,.
\end{equation}
Thus, the normalized lowering operators that move between GT lines are
\begin{equation}
    Z_{T_1}=\id\,,\quad Z_{T_2}=\frac{1}{\sqrt{2}}J_-\,,\quad \text{and}\quad Z_{T_2}=\frac{1}{2}(J_-)^2\,.
\end{equation}
Combining all the previous results leads to the diagonal GT projectors
\begin{equation}
   X^{(2,0)}_{T_1,T_1}= P_{(2,0)}^{\mathrm{hw}}\,, \quad X^{(2,0)}_{T_2,T_2}= \sum_{1\leq p<q\leq n}\ket{p,q;-}\bra{p,q;-}\,, \quad \text{and}\quad X^{(2,0)}_{T_3,T_3}= \sum_{1\leq p<q\leq n}\ket{p,q;0}\bra{p,q;0}\,.
\end{equation}
Next, the off-diagonal elements are of the form
\begin{equation}
    X^{(2,0)}_{T_i,T_j}=\sum_{1\leq p<q\leq n}\ket{p,q;\ell_i}\bra{p,q;\ell_i}\,,\quad \text{where }\quad  \ell_i=\begin{cases}
        2\,, \text{for } i=1\,,\\
        -\,, \text{for } i=2\,,\\
        0\,, \text{for } i=2\,.
    \end{cases}
\end{equation}

\subsection{Spin interpretation of the $t=2$ irreps}\label{app:spininterptr}

Above we have hinted at the fact that since the operators $J_z,J_\pm$ satisfy the $\mathfrak{su}(2)$ commutation relations, we could use this fact to provide a nice spin interpretation to the commutant elements. Here we flesh out this analogy and illustrate it for the irreps previously considered. First, in the highest-weight of the irrep $(\lambda_1,\lambda_2)$, the total particle number operator $M$ and the spin-$z$ will take values
\begin{equation}\label{eq:central-charge}
    N=\lambda_1+\lambda_2\,,\quad J_z=\frac{\lambda_1-\lambda_2}{2}\,.
\end{equation}
This allows us to define the $\mathbb{SU}(2)$ spin index $j=\frac{\lambda_1-\lambda_2}{2}$, which takes values $j=0,\frac{1}{2},1,\ldots,\frac{n}{2}$ and verify that
\begin{equation}
    \dim(V_{(\lambda_1,\lambda_2)}^{\mathbb{U}(2)})=2j+1\,,
\end{equation}
i.e., $V_{(\lambda_1,\lambda_2)}^{\mathbb{U}(2)}$ is a spin $j$ representation of $\mathbb{SU}(2)$. Importantly, now we can relate the GT index $m$ to a spin index $m_{\rm spin}=-j,-j+1\ldots,j$ by subtracting the central charge $M$. That is, 
\begin{equation}
    m_{\rm spin}=m-\frac{N}{2}\,.
\end{equation}

From the previous, we can view the irrep $(0,0)$ as a spin-$0$ representation with a single element in the commutant. Then, $(\lambda_1,\lambda_2)=(1,0)$ irrep corresponds to a spin-$\frac{1}{2}$, where $m=1$ is $m_{\rm spin}=\frac{1}{2}$, while $m=0$ leads to $m_{\rm spin}=-\frac{1}{2}$. Here, we obtain four commutant elements, as expected from the the fact that the vector space is 2 dimensional and thus the algebra of all linear operators on that space has dimension $2\times 2=4$. Next, we can consider the central charge $M=2$ case, which will arise from $(\lambda_1,\lambda_2)=(1,1),(2,0)$, and respectively lead to a singlet $j=0$ and a spin $j=1$ representation. 

In this language, we can identify irreps from the pair $(N,j)$, whereby the copy irrep has dimension $2j+1$ and the commutant $(2j+1)^2$. To count the total number of irrep elements, one thus needs to find how many different $j$'s lead to the same value $N$. Solving from Eq.~\eqref{eq:central-charge} shows that we can have the values $N=2j,2j+1,\ldots,n$. Thus, there are $n-2j+1$ possible values for $N$. Combining it all recovers the formula 
\begin{equation}
\dim(\CC_{2,n}^{\rm PP})=\sum_{j=0,\frac12,1,\frac32,\dots,\frac n2} (n-2j+1)(2j+1)^2
=
\frac{(n+1)(n+2)^2(n+3)}{12}\,.
\end{equation}
An intuitive picture emerges: We can think that the total number of fermions shared among the two copies $N$ corresponds to  $N$ spin $1/2$ particles. Then, creating a particle on the first copy fixes a spin up, while creating on the second a spin down. In this way, the total spin index $j$ corresponds to the different irreducible representations arising from the composition of these effective spins with $m_{\rm spin}$ the values of the $z$-component  of the total spin.

\subsection{Indications for the general $(\lambda_1,\lambda_2)$ case}

The previous examples already showcase the basic ingredients needed to construct all commutant operators for a general $(\lambda_1,\lambda_2)$ irrep. Here, the blocks will take the form
\begin{equation}
    V_{\lambda^T}^{\mathbb{U}(n)}\otimes V_{\lambda}^{\mathbb{U}(2)}=V_{(2^{\lambda_2},1^{\lambda_1-\lambda_2})}^{\mathbb{U}(n)}\otimes V_{(\lambda_1,\lambda_2)}^{\mathbb{U}(2)}\,.
\end{equation} 
On the copy side we will have that 
\begin{equation}
    V_{(\lambda_1,\lambda_2)}^{\mathbb{U}(2)}\cong(\det)
^{\lambda_2}\otimes {\rm Sym}^{\lambda_1-\lambda_2}(\mathbb{C}^2)\,,
\end{equation}
meaning that we can expect $\lambda_2$ singlet-type irreps while a spin $j=(\lambda_1-\lambda_2)/2$ multiplet. 

Then, we note that the block $(\lambda_1,\lambda_2)$ will sit within a given central charge value $M$ with GT patterns indexed by $m=\lambda_2,\ldots,\lambda_1$, The highest-weight line will arise from the distinguished pattern with $m=\lambda_1$ lying in the sector with $\lambda_1$ particles in copy 1 and $\lambda_2$ particles in copy 2. As such, a natural basis will follow from the copy-singlet pair operators
\begin{equation}
    S_{pq}=\begin{cases}
        {\widetilde{a}_p^{\dagger(1)}}{\widetilde{a}_q^{\dagger(2)}}+{\widetilde{a}_q^{\dagger(1)}}{\widetilde{a}_p^{\dagger(2)}}\,,\quad p\leq q\,,\\
        {\widetilde{a}_p^{\dagger(1)}}{\widetilde{a}_p^{\dagger(2)}}\,,\quad\quad\quad\quad\quad\,\,\,\,\, p= q\,,
    \end{cases}
\end{equation}
which contribute one particle to each copy and will be annihilated by the raising operator $J_+$.  Thus, the highest-weight line will be spanned by states of the form
\begin{equation}
    {\widetilde{a}_{i_1}^{\dagger(1)}}\cdots {\widetilde{a}_{i_q}^{\dagger(1)}}S_{j_1k_1}\cdots S_{j_{\lambda_2}k_{\lambda_2}}\ket{0}\,,\quad \text{where}\quad q=\lambda_1-\lambda_2\,,
\end{equation}
followed by a projection onto the irreps of $\mathbb{U}(n)$-type $(2^{\lambda_2},1^q)$. In summary, the highest-weight subspace for $(\lambda_1,\lambda_2)$ is $q$-copy-1 fermions plus $\lambda_2$ singlets with physical symmetry $(2^{\lambda_2},1^q)$. For instance we can recover all previously considered cases
\begin{itemize}
    \item $(1,0)$ leads to $q=1$, $\lambda_2=0$, meaning a copy-1 creation operator,
    \item $(1,1)$ leads to $q=0$, $\lambda_2=1$, meaning a singlet $S_{pq}$,
    \item $(2,0)$ leads to $q=2$, $\lambda_2=0$, meaning two copy-1 creation operator with physical antisymmetry,
    \item $(2,1)$ leads to $q=1$, $\lambda_2=1$, meaning one copy-1 creation operator times one singlet pair, projected to shape $(2,1)$.
\end{itemize}
As such, the block projector will take the form 
\begin{align}
    P_{(\lambda_1,\lambda_2)}=P_{M=\lambda_1+\lambda_2}P_{J^2=j(j+1)}\,.
\end{align}
and the highest-weight one
\begin{align}
    P_{(\lambda_1,\lambda_2)}^{\mathrm{hw}}=P_{(\lambda_1,\lambda_2)}P_{m=\lambda_1}\,.
\end{align}
After that, the whole block is generated by the single lowering operator $J_1$. Since there is only one simple root for $\mathbb{U}(2)$, every GT basis vector is obtained by repeated lowering the highest-weight states. Indeed, if we define a highest-weight vector as $\ket{(\lambda_1,\lambda_2),{\mathrm{hw}}}$, then the GT basis are
\begin{equation}
    \ket{(\lambda_1,\lambda_2),m}=\sqrt{\frac{(m-\lambda_2)!}{(\lambda_1-\lambda_2)!(\lambda_1-m)!}}(J_-)^{\lambda_1-m}\ket{(\lambda_1,\lambda_2),{\mathrm{hw}}}\,.
\end{equation}
The matrix units will then take the form
\begin{equation}
    X^{(\lambda_1,\lambda_2)}_{m,m'}=Z_mP_{(\lambda_1,\lambda_2)}^{\mathrm{hw}}Z_{m'}\,,\quad \text{with}\quad Z_m=\sqrt{\frac{(m-\lambda_2)!}{(\lambda_1-\lambda_2)!(\lambda_1-m)!}}(J_-)^{\lambda_1-m}\,.
\end{equation}

\subsection{$t=3$ case}

When going to $t=3$ the irreps of $\mathbb{U}(3)$ will be labeled by partitions 
\begin{equation}
    \lambda=(\lambda_1,\lambda_2,\lambda_3)\,,\quad m\geq \lambda_1\geq \lambda_2\geq \lambda_3\geq 0\,,
\end{equation}
fitting inside an $n\times 3$ rectangle. 
Then, when it comes to obtaining the GT patterns, we will work with the restriction chain
\begin{equation}
    \mathbb{U}(1)\subset\mathbb{U}(2)\subset\mathbb{U}(3)\,.
\end{equation}
In the first restriction $\mathbb{U}(3)\downarrow\mathbb{U}(2)$ (which is multiplicity free), the irreps of $\mathbb{U}(2)$ that appear are labeled by
\begin{equation}
    \lambda'=(\gamma_1,\gamma_2)\quad m\geq \gamma_1\geq \gamma_2\geq 0\,,
\end{equation}
and must satisfy the interlacing relations
\begin{equation}    \lambda_1\geq\gamma_1\geq\lambda_2\geq\gamma_2\geq\lambda_3\,.
\end{equation}
The next restriction $\mathbb{U}(2)\downarrow\mathbb{U}(1)$, with the same convention acting on the first copy, leads to irreps labeled by 
\begin{equation}
    \lambda''=(\eta)\,,
\end{equation}
with the interleaving condition
\begin{equation}
    \gamma_1\geq\eta\geq\lambda_2\,.
\end{equation}
Combining the previous results leads to the GT patterns
\begin{equation}\label{eq:GT-pattern-passive-t3}
    \begin{array}{ccccc}
       \lambda_1  & & \lambda_2 & &\lambda_3  \\
         & \gamma_1 & & \gamma_2 &\\
                  & & \eta & &  \\
    \end{array}\quad\,.
\end{equation}

Here, we can see that the top row will fix the total particle number by 
\begin{equation}
    N=|\lambda|\,,
\end{equation}
the middle row will fix the number of particles in the first and second copies
\begin{equation}
    N_1+N_2=|\lambda'|\,,
\end{equation}
and the third row the number of particles in the first copy
\begin{equation}
    N_1=|\lambda''|\,.
\end{equation}
This realization can be then used to find appropriate natural bases in which to find the highest-weight states as in the $t=2$ case.

\section{Proofs of Lemmas 3, 4 and Theorem 3}
\label{ap:generators-active}

Here we provide the basic proofs concerning general Gaussian unitaries. We begin by proving Lemma~\ref{lemma:QinC}, which states that $\widetilde{Q}_{j,k}\in \CC_{t,n}\,.$

\begin{lemma}\label{lemma-ap:QinC}
    Let $\widetilde{Q}_{j,k}$ be an operator defined as in Eq.~\eqref{eq:tildeQdef}. Then,   $\widetilde{Q}_{j,k}\in \CC_{t,n}$.
\end{lemma}

\begin{proof}
    We begin by showing that, for all fermionic Gaussian unitaries $R(U)\in\mathbb{SPIN}(2n)$ and all $j,k$,
\begin{equation}
    \left[R(U)^{\otimes t},Q_{j,k}\right]=0\,,
\end{equation}
with $Q_{j,k} =\frac{1}{2}\sum_{\mu=1}^{2n} c_\mu^{(j)}\otimes c_\mu^{(k)}$. Indeed recalling that under conjugation by $R(U)^{\otimes t}$ each Majorana operator $c_\mu^{(j)}$ transforms as
\begin{equation}
    R(U)^{\otimes t} c_\mu^{(j)} R^\dagger(U)^{\otimes t}=\sum_{\nu} U_{\mu\nu}\,c_\nu^{(j)}\,,
\end{equation}
where $U\in\mathbb{SO}(2n)$ is a special orthogonal matrix, we find
\begin{equation} 
    R(U)^{\otimes t} Q_{j,k} R^\dagger(U)^{\otimes t}
=\frac{1}{2}\sum_{\nu,\nu'} U_{\mu\nu}U_{\mu \nu'}\, c_\nu^{(j)}\otimes c_{\nu'}^{(k)} =  \sum_{\nu,\nu'} \delta_{\nu\nu'} \, c_\nu^{(j)}\otimes c_{\nu'}^{(k)} =Q_{jk}\,,
\end{equation}
where we used the orthogonality of $U$, i.e., $\sum_\mu U_{\mu\nu}U_{\mu\nu'}=\delta_{\nu\nu'}$. Therefore, the operators $Q_{j,k}$ belong to the commutant $\CC_{t,n}$. Next, we introduce the operators
\begin{equation}
    \widetilde{c}_\mu^{\,(j)} = \Gamma_1 \otimes \Gamma_2 \otimes \cdots \otimes c_\mu^{(j)} \otimes \id \otimes \id \otimes \cdots \,.
\end{equation}
These modified Majorana operators are constructed in such a way that they satisfy the anticommutation relations
\begin{equation}
    \left\{\widetilde{c}_\mu^{\,(j)},\widetilde{c}_\nu^{\,(k)}\right\} = 2 \delta_{\mu\nu} \delta_{jk}\,.
\end{equation}
That is, besides the usual anticommutation relations of distinct Majorana operators acting on the same copy, the $\widetilde{c}_\mu^{\,(j)}$ operators anti-commute across different copies of $\HC^{\otimes t}$. To see this, let us first assume that $j\neq k$, so that without loss of generality we can choose $k>j$, and simply compute
\begin{align}
    \left\{\widetilde{c}_\mu^{\,(j)},\widetilde{c}_\nu^{\,(k)}\right\} & =  \id \otimes \id \otimes\cdots\otimes c_\mu^{\,(j)} \Gamma_j \otimes \Gamma_{j+1} \otimes \cdots \otimes \widetilde{c}_\nu^{\,(k)} \otimes \id \otimes \id \otimes \cdots \nonumber \\ &+  \id \otimes \id \otimes\cdots \otimes \Gamma_j c_\mu^{\,(j)} \otimes \Gamma_{j+1} \otimes \cdots \otimes \widetilde{c}_\nu^{\,(k)} \otimes \id \otimes \id \otimes \cdots \nonumber \\ &=  \id \otimes \id \otimes\cdots \otimes \left\{c_\mu^{\,(j)}, \Gamma_j\right\} \otimes \Gamma_{j+1} \otimes \cdots \otimes \widetilde{c}_\nu^{\,(k)} \otimes \id \otimes \id \otimes \cdots \nonumber \\ &=  0\,,
\end{align}
where we used that 
\begin{align} \label{eq-ap:Gamma-anticomm}
    \left\{c_\mu^{(j)}, \Gamma_j\right\} & =   \left\{c_\mu^{(j)}, (-i)^n c_1^{(j)}\cdots c_{2n}^{(j)}\right\} \nonumber \\ &=  (-i)^n c_\mu^{(j)}  c_1^{(j)}\cdots c_{2n}^{(j)} +  (-i)^n  c_1^{(j)}\cdots c_{2n}^{(j)}  c_\mu^{(j)} \nonumber \\& = (-i)^n \left((-1)^{\mu-1} c_1^{(j)}\cdots c_{\mu-1}^{(j)} c_{\mu+1}^{(j)} \cdots c_{2n}^{(j)} + (-1)^{2n-\mu} c_1^{(j)}\cdots c_{\mu-1}^{(j)} c_{\mu+1}^{(j)} \cdots c_{2n}^{(j)} \right) \nonumber \\ &= (-i)^n \left((-1)^{\mu-1}+(-1)^{-\mu}\right) \,c_1^{(j)}\cdots c_{\mu-1}^{(j)} c_{\mu+1}^{(j)} \cdots c_{2n}^{(j)} \nonumber \\ &=0 \,.
\end{align}
Moreover, when $j=k$ we find that
\begin{align}
      \left\{\widetilde{c}_\mu^{\,(j)},\widetilde{c}_\nu^{\,(j)}\right\} = \id\otimes\id\cdots \otimes \left\{c_\mu^{(j)},c_\nu^{(j)}\right\} \otimes \id\otimes \id\otimes \cdots = 2\,\delta_{\mu\nu}\,,
\end{align}
so we arrive at the anticommutation relations
\begin{equation}
    \left\{\widetilde{c}_\mu^{\,(j)},\widetilde{c}_\nu^{\,(k)}\right\} = 2\,\delta_{\mu\nu} \delta_{jk}\,.
\end{equation}
Furthermore, notice the elementary property $\left(\widetilde{c}_\mu^{\,(j)}\right)^2 = \id\,$.

Using the $\widetilde{c}_\mu^{\,(j)}$ Majoranas, we can define the operators
\begin{equation}
    \widetilde{Q}_{j,k} = \frac{1}{2}\sum_{\mu=1}^{2n} \widetilde{c}_\mu^{\,(j)} \,\widetilde{c}_\mu^{\,(k)}\,.
\end{equation}
Since $\widetilde{Q}_{j,k}$ can be obtained from $Q_{j,k}$  by simply multiplying by the parity operator  $\Gamma_t = (-i)^n c_1^{(t)}\cdots c_{2n}^{(t)}$ on the copies such that $j\leq t<k$, and  $\Gamma_t$ commutes with fermionic Gaussian unitaries $R(U)^{\otimes t}$ for all $R(U)\in\mathbb{SPIN}(2n)$ (and thus also belongs to $\CC_{t,n}$), its product with a $Q_{j,k}$ commutes with those unitaries as well. We can then conclude that
\begin{equation}
    \left[R(U)^{\otimes t},\widetilde{Q}_{j,k}\right]=0 \qquad \forall U\in\mathbb{SPIN}(2n) \,,
\end{equation}
so the generators $\widetilde{Q}_{j,k}$ are indeed in the commutant $\CC_{t,n}$.
\end{proof}
\vspace{0.3cm}

Next, we prove Lemma \ref{lem:Q-rep-gl} showing that the operator $\widetilde{Q}_{j,k}$ close the $\mathfrak{so}(t)$ algebra.

\begin{lemma}\label{lem-ap:Q-rep-gl}
The operators $\widetilde{Q}_{j,k}$ satisfy the Lie algebra commutation relations
\small
\begin{equation} 
\left[\widetilde{Q}_{j,k},\widetilde{Q}_{j',k'}\right]= \delta_{kj'} \widetilde{Q}_{j,k'} + \delta_{jk'} \widetilde{Q}_{k,j'} -\delta_{kk'} \widetilde{Q}_{j,j'} - \delta_{jj'} \widetilde{Q}_{k,k'} \nonumber\,,
\end{equation}
\normalsize
and therefore generate a representation of the $\mathfrak{so}(t)$ algebra.
\end{lemma}

\begin{proof}
    We simply need to compute the commutation relations of the $\{\widetilde{Q}_{j,k}\}$, and compare them to the commutation relations of the matrices $\{L_{j,k}\}$  in the canonical basis of the standard representation of $\so(t)$, i.e.,
    \begin{equation}\label{eq-ap:L_jk-comm}
        \left[L_{j,k}, L_{j',k'}\right] = \delta_{kj'} L_{j,k'} + \delta_{jk'} L_{k,j'} -\delta_{kk'} L_{j,j'} - \delta_{jj'} L_{k,k'}\,.
    \end{equation}
    Here, the matrix $L_{j,k}$ has entries $\left(L_{j,k}\right)_{ml} = \delta_{jm}\delta_{kl}-\delta_{jl}\delta_{km}$, and  the $L_{j,k}$ can be readily seen to be a basis for the real vector space of real anti-symmetric matrices. 

    The computation gives
    \begin{align}\label{apeq:commutators}
        \left[ \widetilde{Q}_{j,k},\widetilde{Q}_{j',k'} \right] &= \frac{1}{4} \left( \sum_{\mu=1}^{2n} \widetilde{c}_\mu^{\,(j)} \widetilde{c}_\mu^{\,(k)} \sum_{\nu=1}^{2n} \widetilde{c}_\nu^{\,(j')} \widetilde{c}_\nu^{\,(k')} - \sum_{\nu=1}^{2n} \widetilde{c}_\nu^{\,(j')} \widetilde{c}_\nu^{\,(k')} \sum_{\mu=1}^{2n} \widetilde{c}_\mu^{\,(j)} \widetilde{c}_\mu^{\,(k)}\right) \nonumber \\ &= \frac{1}{4} \left(\delta_{kj'} \sum_{\mu=1}^{2n}\sum_{\nu=1}^{2n} \widetilde{c}_\mu^{\,(j)}   \widetilde{c}_\mu^{\,(k)} \widetilde{c}_\nu^{\,(k)} \widetilde{c}_\nu^{\,(k')} -  \delta_{kj'} \sum_{\mu=1}^{2n}\sum_{\nu=1}^{2n} \widetilde{c}_\nu^{\,(k)} \widetilde{c}_\mu^{\,(k)}  \widetilde{c}_\nu^{\,(k')} \widetilde{c}_\mu^{\,(j)}   \right. \nonumber \\ &+   \delta_{jk'} \sum_{\mu=1}^{2n} \sum_{\nu=1}^{2n} \widetilde{c}_\mu^{\,(j)} \widetilde{c}_b^{\,(j)}  \widetilde{c}_\mu^{\,(k)}  \widetilde{c}_\nu^{\,(j')}  - \delta_{jk'}  \sum_{\mu=1}^{2n} \sum_{\nu =1}^{2n}\widetilde{c}_\nu^{\,(j')}    \widetilde{c}_\nu^{\,(j)} \widetilde{c}_\mu^{\,(j)}   \widetilde{c}_\mu^{\,(k)}  \nonumber \\ &-  \delta_{kk'} \sum_{\mu=1}^{2n} \sum_{\nu=1}^{2n} \widetilde{c}_\mu^{\,(j)}   \widetilde{c}_\mu^{\,(k)} \widetilde{c}_\nu^{\,(k)}  \widetilde{c}_\nu^{\,(j')}  + \delta_{kk'} \sum_{\mu=1}^{2n} \sum_{\nu=1}^{2n} \widetilde{c}_\nu^{\,(j')}   \widetilde{c}_\nu^{\,(k)} \widetilde{c}_\mu^{\,(k)}  \widetilde{c}_\mu^{\,(j)}   \nonumber \\ & - \delta_{jj'} \left. \sum_{\mu=1}^{2n} \sum_{\nu=1}^{2n} \widetilde{c}_\mu^{\,(j)} \widetilde{c}_\nu^{\,(j)}   \widetilde{c}_\mu^{\,(k)}  \widetilde{c}_\nu^{\,(k')}   + \delta_{jj'} \sum_{\mu=1}^{2n} \sum_{\nu=1}^{2n} \widetilde{c}_\nu^{\,(j)} \widetilde{c}_\mu^{\,(j)}   \widetilde{c}_\nu^{\,(k')}  \widetilde{c}_\mu^{\,(k)}   \right) \nonumber \\ &=  \delta_{kj'}\widetilde{Q}_{j,k'} + \delta_{jk'} \widetilde{Q}_{k,j'} - \delta_{kk'} \widetilde{Q}_{j,j'} - \delta_{jj'} \widetilde{Q}_{k,k'} \,.
    \end{align}
\end{proof}

We are now in a position to prove 
Theorem~\ref{th:generators}, restated here for convenience. 

\begin{theorem}[Commutant generators]\label{th-ap:generators}
    The $t$-th order commutant of $n$-mode fermionic Gaussian unitaries, $\CC_{t,n}$, is generated as
    \begin{equation}
        \CC_{t,n} = \Span_{\CBB} \left< \widetilde{Q}_{j,j+1}, \Gamma_1 \right> \,,
    \end{equation}
    where $j\in[t-1]$.
\end{theorem}

\begin{proof}

Since the commutation relations in Lemma~\ref{lem:Q-rep-gl} are isomorphic to those of the $L_{j,k}$ basis of $\so(t)$ in Eq.~\eqref{eq-ap:L_jk-comm}, the $\widetilde{Q}_{j,k}$ realize a representation $\varphi$ of $\so(t)$. Once this connection is established, we can invoke a classical result from representation theory, namely, the fermionic Howe duality~\cite{hasegawa1989spin}, which plays for the pair $(\mathbb{SPIN}(2n),\so(t))$ the analogous role that the Schur–Weyl duality plays for the unitary and symmetric groups $(\UBB(d),S_t)$.

\begin{supproposition}[Howe duality, commutant version] \label{lem-ap:Howe-generators}
   Let  $\varphi$ be the Lie algebra homomorphism
   \begin{align} \label{eq-ap:Q_jk-rep}
    \varphi: \quad&\so(t)\; \longrightarrow \;\operatorname{End}\left(\HC_+^{\otimes t}\right) \nonumber\\
    & L_{j,k} \;\;\,\longrightarrow \; \widetilde{Q}_{j,k}\,.
\end{align}
Then, the commutant of the diagonal action of $\mathbb{SPIN}(2n)$ in $\HC_+^{\otimes t}$ is given by
\begin{equation}
    \CC_{t,n}\left(\HC_+^{\otimes t}\right) = \operatorname{End}_{\mathbb{SPIN}(2n)}\left(\HC_+^{\otimes t}\right)=\varphi\left(\mathscr{U}\left(\so(t)\right)\right)\,,
\end{equation}
where $\varphi$ denotes  the canonical extension of the representation $\varphi$ to the universal enveloping algebra $\mathscr{U}(\so(t))$, as in Eq.~\eqref{eq-ap:universal-algebra}.
\end{supproposition}

The fermionic Gaussian representation of $\mathbb{SPIN}(2n)$ acts on $\HC^{\otimes t} =(\HC_+ +\HC_-)^{\otimes t} = \bigoplus_{\vec \sigma} \HC_{\vec \sigma}$, with $\vec{\sigma}=\sigma_1\dots\sigma_t\in\{+,-\}^t$.  The Howe duality in Supplemental Proposition~\ref{lem-ap:Howe-generators} provides the commutant in the $\HC_{++\cdots}$ sector as $\CC_{t,n}\left(\HC_+^{\otimes t}\right) = \Span_\CBB\left\langle \widetilde{Q}_{j,k}\right\rangle$. To obtain the commutant in the full space $\HC^{\otimes t}$, we need to resort to Ref.~\cite{wenzl2020dualities}, that obtains the commutant $\CC_{t,n}$ in the full $\HC^{\otimes t}$ as a $\mathbb{Z}_2$-extension of $\CC_{t,n}\left(\HC_+^{\otimes t}\right)$. 
\begin{supproposition}[Howe duality, full spinor extension]Let  $\phi$ be the Lie algebra homomorphism
   \begin{align} \label{eq-ap:Q_jk-rep2}
    \phi: \quad&\so(t)\; \longrightarrow \;\operatorname{End}\left(\HC^{\otimes t}\right) \nonumber\\
    & L_{j,k} \;\;\,\longrightarrow \; \widetilde{Q}_{j,k}\,.
\end{align}
    The commutant of the diagonal action of $\mathbb{SPIN}(2n)$ in $\HC^{\otimes t}$, $\CC_{t,n}=  \operatorname{End}_{\mathbb{SPIN}(2n)}\left(\HC^{\otimes t}\right)$ is generated by $\phi\left(\mathscr{U}\left(\so(t)\right)\right)$ together with an additional generator $F$, such that $F^2=1$ and
    \begin{equation}
        F \widetilde{Q}_{j,k}  F =\begin{cases}
            - \widetilde{Q}_{j,k} \quad {\rm if\;} j=1\; {\rm o}r\; k=1 \\
            \phantom{-} \widetilde{Q}_{j,k} \quad {\rm otherwise}
        \end{cases}\,.
    \end{equation}
\end{supproposition}
It is easy to show that $\Gamma_1$ fulfills the requirements of the additional involution $F$. Simply note that $\Gamma_1^2=\id$ and
\begin{equation}
    \Gamma_r \widetilde{Q}_{j,k} = (-1)^{\delta_{jr} + \delta_{kr}} \widetilde{Q}_{j,k} \Gamma_r \,,
\end{equation}
which can be explicitly checked, since for $r=j$ we have
\begin{equation}
     \Gamma_r \widetilde{c}_\mu^{\,(j)} = \Gamma_1 \otimes \Gamma_2 \otimes \cdots \otimes \Gamma_j c_\mu^{(j)} \otimes \id\otimes\cdots= - \Gamma_1 \otimes \Gamma_2 \otimes \cdots \otimes c_\mu^{(j)} \Gamma_j \otimes \id\otimes\cdots = - c_\mu^{(j)} \Gamma_r\,,
\end{equation}
where we used Eq.~\eqref{eq-ap:Gamma-anticomm}, and it is clear that $\left[\Gamma_r,c_\mu^{(j)}\right]=0$ when $r\neq j$. Therefore, $\CC_{t,n} = \Span_\CBB\left\langle \widetilde{Q}_{j,k},\Gamma_1\right\rangle$. Finally, any $\widetilde{Q}_{j,k}$ can be generated by commutation using only neighboring $\widetilde{Q}_{j,j+1}$ generators (see Eq.\ \eqref{apeq:commutators}) which concludes the proof.

\end{proof}

\section{Proof of Theorem 4}
\label{ap:active-dim}

In this section, we provide a detailed proof of Theorem~\ref{th-ap:active-dimension}, which we here recall for convenience.

\begin{theorem}[Matchgate commutant dimension]\label{th-ap:active-dimension}
    The dimension of the $t$-th order commutant of the $n$-qubit matchgate group is given by
    \begin{equation}
    \label{ap-eq:PP-dim}
        \dim(\CC_{t,n}) =\frac{1}{2^{n-1}} \prod_{j=0}^{n-1} \frac{(2j)!\,(2t+2j)!}{(t+j)!\,(t+n+j-1)!}\ \,.
    \end{equation}
\end{theorem}

\begin{proof}
    Let $\HC=\mathbb{C}^{2^n}$ be the Hilbert space of $n$ qubits, and let $G=\Spin(2n)$ denote the $n$-qubit matchgate group acting on it by left multiplication by the representation $R(U)$, for $U\in G$.
    To prove the claimed formula for the dimension of $\CC_{t,n}$ we will follow the very same approach as we did for the PP case in Appendix~\ref{ap:PP-dim}.

    Recall from there that given a compact group $G$ and a finite dimensional representation of its $V\cong\mathbb C^{\dim V}$ one has
    \begin{equation}
        \dim (\End_G(V)) = \sum_\lambda m_\lambda^2=\int_G |\chi_V(g)|^2d\mu(g)\,,
    \end{equation}
    for $\lambda$ indexing the irreps appearing in the isotypical decomposition of $V$ under $G$, $m_\lambda$ their multiplicities, and $\chi_V=\Tr[g]$ (identifying the group element $g$ with its representation on $V)$ the character of the representation of $G$ over $V$.
    Lastly, we also recall the Weyl integration formula from Eq.~\eqref{ap-eq:weyl_integral}
    \begin{equation}
    \label{ap-eq:weyl_integral_2}
        \int_G f(g) d\mu(g)=\frac{1}{|W|}\int_T f(\tau)|\delta(\tau)|^2d\mu(\tau)\,.
    \end{equation}
    Above, as already explained in Section~\ref{ap:PP-dim}, $T\subset G$ is the maximal torus of $G$, $W=N(G)/T$ is $G$'s Weyl group, $\delta(\tau)$ is the Weyl denominator (see Eq.~\eqref{ap-eq:weyl-denominator-general}), and $d\mu(\tau)$ the Haar measure on $T$.
    
    We now again proceed to explicitly evaluate all the needed ingredients. Let us start by identifying the maximal torus $T$ of $G=\Spin(2n)$.
    The Lie algebra of $G=\Spin(2n)$ is $\mathfrak{g}=\mathfrak{so}(2n)$, and its Cartan subalgebra $\mathfrak{t}$ is well known~\cite{hall2013lie,goodman2009symmetry}. In the spinor representation $V=\HC$ that we are considering, a convenient basis for $\mathfrak{t}$ is $\mathfrak{t}=\Span(\{iZ_j/2\}_{j=1}^n)$, i.e. the local Pauli $Z$ operators. Hence, the maximal torus of interest is the matchgate subgroup consisting of $n$ independent rotations along the $Z$ axis on each qubit, i.e.,
    \begin{equation}
    \label{eq-ap:spin_torus}
        T=\left\{\prod_{j=1}^n e^{i\frac{\theta_j}{2} Z_j}\;\Big|\;\theta_j \in [0,2\pi]\,\quad \forall j\in 1\dots n\right\}\,.
    \end{equation}
    The character $\chi_\HC$ restricted to $T$ is hence easily computed as 
    \begin{equation}
    \label{ap-eq:character_torus}
        \chi_\HC(\tau) =\Tr[\tau] = \prod_{j=1}^n\left(e^{i\frac{\theta_j}{2}}+e^{-i\frac{\theta_j}{2}}\right)=
        \prod_{j=1}^n2\cos\left(\frac{\theta_j}{2}\right)\,,\quad {\rm for}\;\;\tau=\prod_{j=1}^n e^{i\frac{\theta_j}{2} Z_j}\,.
    \end{equation}
    The Weyl group of $G$ is known as well, see for example~\cite{kirillov2008introduction}. Particularly, it corresponds to the group of even signed permutations of the angles $(\theta_1,\dots,\theta_n)$, hence its size is $|W|=2^{n-1} n!$. 
    Lastly, we need the roots of $G$ relative to $T$. Again, these are well known in the representation theory literature~\cite{kirillov2008introduction}. The weights of the spinor representation of $\Spin(2n)$ correspond to the $2^n$ $n$-tuples $(\pm\frac{1}{2},\dots,\pm\frac{1}{2})$, each appearing with multiplicity one. Notice that those correspond to the signs one can assign to the Cartan subalgebra generators, and that the Weyl group acts by permuting them and flipping the sign of an even number of them. Defining $e_1,\dots,e_n$ the standard basis of the weight space, i.e. of the Cartan subalgebra, the positive roots are
    \begin{equation}
        \Phi^+=\{e_i-e_j,e_i+e_j\;|\;1\leq i<j\leq n\}\,.
    \end{equation}
    Hence, the Weyl determinant can be computed as
    \begin{align}
    \label{ap-eq:weyl-denominator-spin}
        \delta(\tau)
        &=
        \prod_{1\leq l < m\leq n}
        \left(e^{i\frac{\theta_l-\theta_m}{2}}-e^{-i\frac{\theta_l-\theta_m}{2}}\right)
        \prod_{1\leq p < q\leq n}\left(e^{i\frac{\theta_p+\theta_q}{2}}-e^{-i\frac{\theta_p+\theta_q}{2}}\right)\\
        &=
        \prod_{1\leq l < m\leq n}
        (2i)\sin\left(\frac{\theta_l-\theta_m}{2}\right)
        \prod_{1\leq p < q\leq n}(2i)\sin\left(\frac{\theta_p+\theta_q}{2}\right)\\
        &=
        \prod_{1\leq l < m\leq n}
        (-4)\sin\left(\frac{\theta_l-\theta_m}{2}\right)
        \sin\left(\frac{\theta_l+\theta_m}{2}\right)\,.
    \end{align}

    We now have all the ingredients to evaluate the integral in Eq.~\eqref{ap-eq:weyl_integral_2} for the class function $f(g)=\chi_\HC(g)$. Indeed, using Eqs.~\eqref{ap-eq:character_torus} and~\eqref{ap-eq:weyl-denominator-spin}, together with $|W|=2^{n-1}/n!$ and the fact that with the parametrization in Eq.~\eqref{eq-ap:spin_torus} the haar measure on the torus $T$ reads $d\mu(\tau)=\prod_j d\theta_j/2\pi$\footnote{This follows immediately from the fact that $T$ is abelian.}, we have
    \begin{align}
    \label{ap-eq:weyl_integral_character}
        \dim(\CC_{t,n})&=\int_{\Spin(2n)} |\chi_\HC(g)|^{2t} d\mu(g) = \frac{1}{|W|}\int_T |\chi_\HC(\tau)|^{2t}|\delta(\tau)|^2d\mu(\tau)\\
        &=\frac{1}{2^{n-1}n!}
        \int_{[0,2\pi]^n}
        \left[\prod_{j=1}^n \left(2\cos\left(\frac{\theta_j}{2}\right)\right)^{2t}\right]
        \left[\prod_{1\leq l < m\leq n}
        16\sin\left(\frac{\theta_l-\theta_m}{2}\right)^2\sin\left(\frac{\theta_l+\theta_m}{2}\right)^2\right]\prod_{j=1}^n \frac{d\theta_j}{2\pi}
        \,.
    \end{align}
    The problem is hence reduced to carrying out the trigonometric integral above.
    Let us perform a simple change of variables to bring all the trigonometric functions to $\sin$ and get rid of the annoying denominators inside them. By setting
    \begin{equation}
        \phi_j = \frac{\theta_j}{2} - \frac{\pi}{2}\,,\qquad d\phi_j = \frac{d\theta_j}{2}\,,
    \end{equation}
    we get to 
    \begin{equation}
    \label{ap-eq:weyl_integral_character_halfangle}
        \dim(\CC_{t,n})
        =\frac{1}{2^{n-1}n!}
        \int_{[0,\pi]^n}
        \left[\prod_{j=1}^n \left(2\sin\left(\phi_j\right)\right)^{2t}\right]
        \left[\prod_{1\leq l < m\leq n}
        16\sin\left(\phi_l-\phi_m\right)^2\sin\left(\phi_l+\phi_m\right)^2\right]\prod_{j=1}^n \frac{d\phi_j}{\pi}
        \,.
    \end{equation}
    Now let us recall the standard trigonometric identity
    \begin{equation}
        \sin(\phi_l-\phi_m)\sin(\phi_l+\phi_m)
        =\sin^2\phi_l-\sin^2\phi_m\,,
    \end{equation}
    and let us set
    \begin{equation}
        x_j=\cos\phi_j\,,\qquad \phi_j\in[0,\pi]\,,\qquad x_j\in[-1,1]\,.
    \end{equation}
    Then
    \begin{equation}
        \sin^2\phi_j=1-x_j^2\,,
        \qquad
        d\phi_j=\frac{dx_j}{\sqrt{1-x_j^2}}\,,
    \end{equation}
    and moreover
    \begin{equation}
        \sin^2(\phi_l-\phi_m)\sin^2(\phi_l+\phi_m)
        =\left(x_l^2-x_m^2\right)^2\,.
    \end{equation}
    Therefore, Eq.~\eqref{ap-eq:weyl_integral_character_halfangle} becomes
    \begin{align}
        \dim(\CC_{t,n})
        &=
        \frac{2^{2tn}\,16^{\binom n2}}{2^{n-1}n!\,\pi^n}
        \int_{[-1,1]^n}
        \left[\prod_{j=1}^n (1-x_j^2)^{t-\frac12}\right]
        \left[\prod_{1\le l<m\le n}(x_l^2-x_m^2)^2\right]
        \prod_{j=1}^n dx_j
        \nonumber\\
        &=
        \frac{2^{2tn+2n(n-1)-n+1}}{n!\,\pi^n}
        \int_{[-1,1]^n}
        \left[\prod_{j=1}^n (1-x_j^2)^{t-\frac12}\right]
        \left[\prod_{1\le l<m\le n}(x_l^2-x_m^2)^2\right]
        \prod_{j=1}^n dx_j\,.
    \label{ap-eq:weyl_integral_xvars}
    \end{align}
    Since the integrand is even in each variable $x_j$, we can set
    \begin{equation}
        y_j=x_j^2\in[0,1]\,.
    \end{equation}
    This way, using
    \begin{equation}
        2\int_0^1 f(x^2)\,dx=\int_0^1 y^{-1/2}f(y)\,dy
    \end{equation}
    independently for each variable, we get to
    \begin{align}
        \dim(\CC_{t,n})
        &=
        \frac{2^{2tn+2n(n-1)-n+1}}{n!\,\pi^n}
        \int_{[0,1]^n}
        \left[\prod_{j=1}^n y_j^{-1/2}(1-y_j)^{t-\frac12}\right]
        \left[\prod_{1\le l<m\le n}(y_l-y_m)^2\right]
        \prod_{j=1}^n dy_j\,.
    \label{ap-eq:weyl_integral_selberg_form}
    \end{align}
    The last step is then recognizing that this integral is precisely a Selberg integral
    \begin{align}
    \label{ap-eq:selberg_integral}
        S_n(\alpha,\beta,\gamma)
        &=
        \int_0^1 \cdots \int_0^1
        \prod_{i=1}^n t_i^{\alpha-1}(1-t_i)^{\beta-1}
        \prod_{1\le i<j\le n} |t_i-t_j|^{2\gamma}
        \,dt_1\cdots dt_n \\
        &=
        \prod_{j=0}^{n-1}
        \frac{
        \Gamma(\alpha+j\gamma)\,
        \Gamma(\beta+j\gamma)\,
        \Gamma\!\bigl(1+(j+1)\gamma\bigr)
        }{
        \Gamma\!\bigl(\alpha+\beta+(n+j-1)\gamma\bigr)\,
        \Gamma(1+\gamma)
        }\,,
    \end{align}
    where the last line holds whenever
    \begin{equation}
    \label{ap-eq:selberg_convergence_conditions}
        \Re(\alpha)>0,\qquad
        \Re(\beta)>0,\qquad
        \Re(\gamma)>
        -\min\!\left(
        \frac1n,\,
        \frac{\Re(\alpha)}{n-1},\,
        \frac{\Re(\beta)}{n-1}
        \right)\,.
    \end{equation}
    In particular, we can then see that Eq.~\eqref{ap-eq:weyl_integral_selberg_form} is a Selberg integral with parameters
    \begin{equation}
    \label{ap-eq:selberg_params}
        \alpha=\frac12,\qquad \beta=t+\frac12,\qquad \gamma=1\,,
    \end{equation}
    which clearly satisfy the conditions in Eq.~\eqref{ap-eq:selberg_convergence_conditions}.
    Hence, substituting the parameters in Eq.~\eqref{ap-eq:selberg_params} into Eq.~\eqref{ap-eq:selberg_integral}, Eq.~\eqref{ap-eq:weyl_integral_selberg_form} evaluates to
    \begin{align}
        \dim(\CC_{t,n})
        &=
        \frac{2^{2tn+2n(n-1)-n+1}}{n!\,\pi^n}
        \prod_{j=0}^{n-1}
        \frac{
        \Gamma\!\left(j+\frac12\right)
        \Gamma\!\left(t+j+\frac12\right)
        \Gamma(j+2)
        }{
        \Gamma(t+n+j)
        }\,.
    \label{ap-eq:weyl_integral_gamma_answer}
    \end{align}
    Since by construction $t\in\mathbb Z_{\ge 0}$, this can be further simplified by using the standard gamma function identities
    \begin{equation}
        \Gamma\!\left(j+\frac12\right)=\frac{(2j)!}{4^j j!}\sqrt{\pi}\,,
        \qquad
        \Gamma\!\left(t+j+\frac12\right)=\frac{(2t+2j)!}{4^{t+j}(t+j)!}\sqrt{\pi}º,,
        \qquad
        \Gamma(j+2)=(j+1)!\,,
    \end{equation}
    which lead to
    \begin{align}
        \dim(\CC_{t,n})
        &=
        \frac{1}{2^{n-1}}
        \prod_{j=0}^{n-1}
        \frac{(2j)!\,(2t+2j)!}{(t+j)!\,(t+n+j-1)!}\,,
    \label{ap-eq:weyl_integral_factorial_answer}
    \end{align}
    which matches Eq.~\eqref{ap-eq:PP-dim}, hence completing the proof.
\end{proof}

Let us conclude this section by showing how this expression simplifies for the first three values of $t$, which until now were the only cases where $\CC_{t,n}$ had been characterized. From Eq.~\eqref{ap-eq:weyl_integral_factorial_answer} one readily finds
\begin{equation}
    k=1:\ \dim=2\,,\qquad
    k=2:\ \dim=4n+2\,,\qquad
    k=3:\ \dim=\frac{2}{3}(n+1)(2n+1)(2n+3)\,,
\end{equation}
which match the known dimensions of $\CC_{\{1,2,3\},n}$, as reported in~\cite{diaz2023showcasing,heyraud2024unified}.

\section{GT method for general Gaussian unitaries}
\label{ap:gauss_GT_method}

Here we give a detailed description of how to apply the GT method to construct a basis of the commutant $\CC_{t,n}$.

Recall that in this case, the group of interest is $G=\Spin(2n)$, which acts on the Hilbert space of $n$ qubits $\HC\cong\CBB^{2^n}$ via the spinor representation $U$. Under the latter, $\HC$ decomposes in the irreducible even and odd parity sectors $\HC\cong\HC_+\oplus\HC_-$.
Let us further recall that, as stated in the main text and proven in Appendix~\ref{ap:generators-active}, the operators $\widetilde{Q}_{jk}$ satisfy the $\mathfrak{so}(t)$ commutation relations, hence the associative algebra
\begin{equation}
    A_{t,n}=\Alg\langle \widetilde{Q}_{jk}\,:\,1\le j<k\le t\rangle
\end{equation}
is the image of the universal enveloping algebra $\mathscr{U}(\mathfrak{so}(t))$ in the representation carried by the $t$-fold spinor tensor power. 
By the skew Howe duality~\cite{aboumrad2022skew,wenzl2020dualities}, we know that on $\HC^{\otimes t}$ there are two commuting actions given by the group $\Spin(2n)$, acting identically on each of the $t$ copies, and by the associative algebra generated by the $\widetilde{Q}_{jk}$ and $\Gamma_1$. Since they commute, the corresponding Howe duality implies that $\HC^{\otimes t}$ decomposes as
\begin{equation}
    \HC^{\otimes t} \cong \bigoplus_\lambda V_\lambda^\mathbb{SPIN}(2n) \otimes W_\lambda^{\mathfrak{so}(t)}\,,
\end{equation}
for $V_\lambda^\mathbb{SPIN}(2n)$ and $W_\lambda^{\mathfrak{so}(t)}$ irreducible modules of the two aforementioned actions, and $\lambda$ a common label thereof. One can think of this decomposition as the usual isotypic decomposition in irreps of $\Spin(2n)$, where the multiplicity space is now endowed with a richer structure coming from the action of the $\mathfrak{so}(t)$ associative algebra. From this decomposition, we also immediately get a decomposition of the commutant $\CC_{t,n}$, since the latter corresponds to the $\mathfrak{so}(t)$ associative algebra action on the $W_\lambda^{\mathfrak{so}(t)}$ bits. Particularly one reads
\begin{equation}
    \CC_{t,n} \cong \bigoplus_\lambda \End (W_\lambda^{\mathfrak{so}(t)})\,.
\end{equation}
Hence, we can use this decomposition to organize a basis of $\CC_{t,n}$ by dividing it in blocks $\lambda$ and assigning inside each block a canonical matrix basis of $\End(W_\lambda^{\mathfrak{so}(t)})$, i.e. the set of matrices with a single non-zero entry at any given position.
To do this, we can first build a canonical basis of $W_\lambda^{\mathfrak{so}(t)}$, which we will denote as $\{\ket{\lambda, T}\}$, for $T$ an index labeling the $\dim(W_\lambda^{\mathfrak{so}(t)})$ canonical basis vectors, and then promote it to one of $\End(W_\lambda^{\mathfrak{so}(t)})$ by the set of projectors $\{\ketbra{\lambda,T}{\lambda,T'}\}$.

The fact that $W_\lambda^{\mathfrak{so}(t)}$ is an irreducible module of the action of the associative algebra of $\mathfrak{so}(t)$ allows us to construct the basis $\{\ket{\lambda, T}\}$ by the standard Gelfand-Tsetlin procedure of lowering the $\mathfrak{so}(t)$ highest-weight inside $W_\lambda^{\mathfrak{so}(t)}$, $\ket{\lambda,{\rm hw}}$, by all possible sequences of lowering operators. This way, the problem becomes completely characterized by the weights and root system of the $\mathfrak{so}$ algebra.

We start by dividing the general case $t\in \mathbb Z_{\ge 0}$ into the even $t=2r$ and odd $t=2r+1$ cases. Indeed, depending on the parity of $t$ the algebra $\mathfrak{so}(t)$ is either of type $D$ (even) or $B$ (odd)~\cite{kirillov2008introduction}. Inside the generators $\widetilde{Q}_{ij}$ algebra, we can identify the Cartan generators
\begin{equation}
    H_p = \widetilde{Q}_{2p-1,\,2p},
    \qquad p=1,\dots,r\,,
\end{equation}
and associate the Cartan algebra with the vector space $\mathbb C^r$ with canonical basis $\{e_1,\dots,e_r\}$. With this choice, the positive root vectors can be expressed as follows
\begin{equation}
    \Phi^+=\begin{cases}
        \{e_p + e_q, e_p - e_q\,|\,1\leq p<q \leq r\}, &{t=2r},\\
        \{e_p + e_q, e_p - e_q, e_p\,|\,1\leq p<q \leq r\}, &{t=2r+1}.
    \end{cases}
\end{equation}
The operators associated to the positive roots give the raising operators that annihilate the highest-weight, i.e. the common eigenvector of the Cartan generators $H_p$ with highest cumulative eigenvalue. Hence, for the case $t=2r$ one has raising operators
\begin{align}
    E_{e_p-e_q}
    &=
    \frac12\Big(
        \widetilde{Q}_{2p-1,\,2q-1}+\widetilde{Q}_{2p,\,2q}
        +i\big(\widetilde{Q}_{2p,\,2q-1}-\widetilde{Q}_{2p-1,\,2q}\big)
    \Big),\\
    E_{e_p+e_q}
    &=
    \frac12\Big(
        \widetilde{Q}_{2p-1,\,2q-1}-\widetilde{Q}_{2p,\,2q}
        +i\big(\widetilde{Q}_{2p,\,2q-1}+\widetilde{Q}_{2p-1,\,2q}\big)
    \Big),
\end{align}
while in the odd case $t=2r+1$, in addition to the operators above one also has the added ones
\begin{equation}
    E_{e_p}
    =
    \frac{1}{\sqrt2}\Big(\widetilde{Q}_{2p-1,\,2r+1}+i\widetilde{Q}_{2p,\,2r+1}\Big)\,.
\end{equation}
Labeling $\alpha\in \Phi^+$ the positive roots, from the raising operators $E_\alpha$ one can define the lowering ones by taking their Hermitian conjugate. Let us dub them as $F_\alpha=E^\dagger_\alpha$.
Furthermore, under these conventions, the standard simple roots are
\begin{equation}
    \alpha_i=e_i-e_{i+1},
    \qquad i=1,\dots,r-1\,,
\end{equation}
together with
\begin{equation}
    \alpha_r=
    \begin{cases}
        e_r, & t=2r+1\ \text{(type $B_r$)},\\
        e_{r-1}+e_r, & t=2r\ \text{(type $D_r$)}.
    \end{cases}
\end{equation}
The corresponding simple lowering operators are therefore
\begin{equation}
    f_i=F_{\alpha_i},
    \qquad i=1,\dots,r\,.
\end{equation}
Notice that all the operators introduced so far are written explicitly in terms of linear combinations of the generators $\widetilde{Q}_{jk}$, hence as quadratic Majorana operators.

Now, the irreducible $\mathfrak{so}(t)$-modules appearing in $A_{t,n}$, i.e., the spaces $W_\lambda^{\mathfrak{so}(t)}$, are indexed by dominant highest-weights $P_+$
\begin{equation}
    \lambda=(\lambda_1,\dots,\lambda_r)\in P_+\,,
\end{equation}
written in the canonical basis $(e_1,\dots,e_r)$, together with the standard dominance conditions
\begin{equation}\label{eq-so-tBr-app}
    \lambda_1\ge \lambda_2\ge \cdots \ge \lambda_r\ge 0
    \qquad\text{for $t=2r+1$ (type $B_r$)},
\end{equation}
and
\begin{equation}\label{eq-so-tDr-app}
    \lambda_1\ge \lambda_2\ge \cdots \ge \lambda_{r-1}\ge |\lambda_r|
    \qquad\text{for $t=2r$ (type $D_r$)}.
\end{equation}
Up to now we simply described the basic properties of the algebra $\mathfrak{so}(t)$. However, in the setting at hand not all of the highest-weights occur. Indeed, the decomposition of $V=\HC^{\otimes t}$ under the commuting actions of $\Spin(2n)$ and of $A_{t,n}$ is finite, and the available $\mathfrak{so}(t)$ highest-weights are truncated by the number $n$ of physical fermionic modes. Concretely, one has the constraint
\begin{equation}
    \lambda_1\le n\,.
\end{equation}
We denote the resulting finite set by
\begin{equation}
    \Lambda_{t,n}=
    \left\{
        \lambda\in P_+(\mathfrak{so}(t))
        \,|\,
        \lambda_1\le n
    \right\}\,.
\end{equation}
The decomposition of the commutant then reads
\begin{equation}
    \CC_{t,n}\cong \bigoplus_{\lambda\in\Lambda_{t,n}} \End(W_\lambda^{\mathfrak{so}(t)})\,,
\end{equation}
and therefore
\begin{equation}
    \dim (A_{t,n})
    =
    \sum_{\lambda\in\Lambda_{t,n}} (\dim W_\lambda^{\mathfrak{so}(t)})^2
    =
    \frac12\,\dim(\CC_{t,n})\,.
\end{equation}

What is left now is to construct, for each fixed $\lambda\in P_+$, the canonical GT basis of $W_\lambda^{\mathfrak{so}(t)}$. To do so, one considers the natural chain of embedded subalgebras
\begin{equation}
    \mathfrak{so}(2)\subset \mathfrak{so}(3)\subset \cdots \subset \mathfrak{so}(t)\,,
\end{equation}
where $\mathfrak{so}(m)$ is generated by the $\widetilde{Q}_{jk}$ with $1\le j<k\le m$, namely by restricting to the first $m$ copies among the available $t$. 

Particularly, fixing a top weight $\lambda=\lambda^{(t)}\in \Lambda_{t,n}$, an orthogonal GT pattern $T$ of top row $\lambda$ is a chain
\begin{equation}
    \lambda^{(t)} \succ \lambda^{(t-1)} \succ \cdots \succ \lambda^{(2)}\,,
\end{equation}
where $\lambda^{(m)}$ is a dominant highest-weight of $\mathfrak{so}(m)$, and the symbol $\succ$ means that consecutive rows satisfy the standard multiplicity-free branching inequalities. In the odd-to-even step
\begin{equation}
    \mathfrak{so}(2r+1)\downarrow \mathfrak{so}(2r)\,,
\end{equation}
if
\begin{equation}
    \lambda^{(2r+1)}=(\lambda_1,\dots,\lambda_r),
    \qquad
    \lambda^{(2r)}=(\mu_1,\dots,\mu_r)\,,
\end{equation}
then one must have
\begin{equation}
    \lambda_1\ge \mu_1\ge \lambda_2\ge \mu_2\ge \cdots \ge \lambda_{r-1}\ge \mu_{r-1}\ge \lambda_r\ge |\mu_r|\,.
\end{equation}
Similarly, in the even-to-odd step
\begin{equation}
    \mathfrak{so}(2r)\downarrow \mathfrak{so}(2r-1)\,,
\end{equation}
if
\begin{equation}
    \lambda^{(2r)}=(\lambda_1,\dots,\lambda_r)\,,
    \qquad
    \lambda^{(2r-1)}=(\nu_1,\dots,\nu_{r-1})\,,
\end{equation}
then one must have
\begin{equation}
    \lambda_1\ge \nu_1\ge \lambda_2\ge \nu_2\ge \cdots \ge \lambda_{r-1}\ge \nu_{r-1}\ge |\lambda_r|\,.
\end{equation}
These are the orthogonal branching rules that guide the construction of the GT basis. Since the branching is multiplicity-free at each level, every chain $T$ picks up a unique basis vector in $W_\lambda^{\mathfrak{so}(t)}$, which we denote by $\ket{\lambda,T}$. 
In particular one has
\begin{equation}
    |\GT(\lambda)|=\dim W_\lambda^{\mathfrak{so}(t)}\,,
\end{equation}
where $\GT(\lambda)$ denotes the set of all GT patterns with top row $\lambda$.

In practice, to construct the GT basis of an $\mathfrak{so(t)}$ irrep $\lambda$, we can adopt the following recursive procedure. First, list all admissible top rows $\lambda=\lambda^{(t)}\in\Lambda_{t,n}$. Then, for each such $\lambda^{(t)}$, list all $\lambda^{(t-1)}$ satisfying the appropriate interlacing inequalities. For each such $\lambda^{(t-1)}$, list all admissible $\lambda^{(t-2)}$, and so on, until $m=2$ is reached. Because all rows are bounded by $\lambda_1\le n$, this procedure is finite.

We now explain how to realize the vectors $\ket{\lambda,T}\in W_\lambda^{\mathfrak{so}(t)}$ and the corresponding matrix units spanning $\End(W_\lambda^{\mathfrak{so}(t)})$, explicitly in terms of the generators $\widetilde{Q}_{jk}$. For this purpose, one starts from the highest-weight vector $\ket{\lambda,{\rm hw}}\in W_\lambda^{\mathfrak{so}(t)}$, and then lowers it along the path specified by the pattern $T$. Let us denote the operator implementing the lowering path specified by $T$ by $Z_T$. The latter thus satisfies
\begin{equation}
    \ket{\lambda,T}\propto Z_T\ket{\lambda,{\rm hw}}\,.
\end{equation}
The operator $Z_T$ can be built recursively as a product of step operators, one for each branching step
\begin{equation}
    \mathfrak{so}(m)\downarrow \mathfrak{so}(m-1)\,.
\end{equation}
At the step $m\to m-1$, only the negative roots of $\mathfrak{so}(m)$ that do not belong to the root system of $\mathfrak{so}(m-1)$ are allowed to appear. In this sense we can say that the step operator is constructed from the roots that disappear when passing from $\mathfrak{so}(m)$ to its smaller subalgebra.
Concretely, for the branching
\begin{equation}
    \mathfrak{so}(2r+1)\downarrow \mathfrak{so}(2r)\,,
\end{equation}
the roots present in $\mathfrak{so}(2r+1)$ but absent in $\mathfrak{so}(2r)$ are
\begin{equation}
    e_1,\dots,e_r\,,
\end{equation}
so the corresponding step operator is a polynomial in
\begin{equation}
    F_{e_1},\dots,F_{e_r}\,.
\end{equation}
In turn, for the branching
\begin{equation}
    \mathfrak{so}(2r)\downarrow \mathfrak{so}(2r-1)\,,
\end{equation}
the roots present in $\mathfrak{so}(2r)$ but absent in $\mathfrak{so}(2r-1)$ are
\begin{equation}
    e_p-e_r,\qquad e_p+e_r,
    \qquad p=1,\dots,r-1,
\end{equation}
so the corresponding step operator needs to be a polynomial in
\begin{equation}
    F_{e_p-e_r}\,,\qquad F_{e_p+e_r}\,,
    \qquad p=1,\dots,r-1\,.
\end{equation}
Notice that there is still an ambiguity in the order that the lowering operators should be applied. Thus, in order to make this choice canonical, one fixes a Poincar\'e--Birkhoff--Witt (PBW) order on the negative roots~\cite{molev2006gelfand}, for instance the lexicographic order
\begin{equation}
    F_{e_1-e_2},\,F_{e_1-e_3},\,\dots,\,F_{e_{r-1}-e_r},
    \,
    F_{e_1+e_2},\,F_{e_1+e_3},\,\dots,\,F_{e_{r-1}+e_r}\,,
\end{equation}
and, in the odd case, one inserts
\begin{equation}
    F_{e_1},\,F_{e_2},\,\dots,\,F_{e_r}
\end{equation}
at a fixed position. Now suppose $\eta=\lambda^{(m)}$ and $\mu=\lambda^{(m-1)}$ are two consecutive rows of a GT pattern. Among all PBW monomials in the step roots for $m\to m-1$, choose the lexicographically first monomial
\begin{equation}
    Y(\eta,\mu)
\end{equation}
with the property that, when applied to the highest-weight vector of the $\mathfrak{so}(m)$-module of highest-weight $\eta$, it has a nonzero component on the $\mathfrak{so}(m-1)$-highest vector of highest-weight $\mu$. Since the branching is multiplicity-free, this space is one-dimensional, so after fixing a normalization one obtains a unique step operator
\begin{equation}
    Z(\eta,\mu)
\end{equation}
such that
\begin{equation}
    Z(\eta,\mu)\ket{\eta, {\mathrm{hw}}}
\end{equation}
is a nonzero $\mathfrak{so}(m-1)$-highest vector of highest-weight $\mu$. For a full GT pattern
\begin{equation}
    T=\bigl(\lambda^{(t)},\lambda^{(t-1)},\dots,\lambda^{(2)}\bigr),
\end{equation}
one then defines
\begin{equation}
    Z_T
    =
    Z\bigl(\lambda^{(t)},\lambda^{(t-1)}\bigr)\,
    Z\bigl(\lambda^{(t-1)},\lambda^{(t-2)}\bigr)\cdots
    Z\bigl(\lambda^{(3)},\lambda^{(2)}\bigr)\,.
\end{equation}
This way, $Z_T$ is an explicit ordered product of the root vectors
\begin{equation}
    F_{e_p-e_q},\qquad F_{e_p+e_q},\qquad F_{e_p},
\end{equation}
and hence, via the definitions of the lowering operators, an explicit polynomial in the quadratic Majorana generators $\widetilde{Q}_{jk}$.

Now that we are able to construct a basis of $W_\lambda^{\mathfrak{so}(t)}$, to promote it to a basis of $\End(W_\lambda^{\mathfrak{so}(t)})$ we need to first construct the projector onto the highest-weight of the block $W_\lambda^{\mathfrak{so}(t)}$, and then apply the operators $Z_T, Z_{T'}$ at its sides. For each $m=2,\dots,t$, define the quadratic Casimir of the embedded subalgebra $\mathfrak{so}(m)\subset \mathfrak{so}(t)$ by
\begin{equation}
    C^{(m)}=\sum_{1\le j<k\le m} \widetilde{Q}_{jk}^2\,.
\end{equation}
On an irreducible $\mathfrak{so}(m)$-module of highest-weight $\mu$, its eigenvalue is
\begin{equation}
    \chi_m(\mu)=\langle \mu,\mu+2\rho_m\rangle,
\end{equation}
where $\rho_m$ is the so called Weyl vector of $\mathfrak{so}(m)$. Explicitly,
\begin{equation}
    \rho_{2r+1}=\Bigl(r-\tfrac12,\,r-\tfrac32,\,\dots,\,\tfrac12\Bigr),
    \qquad
    \rho_{2r}=(r-1,\,r-2,\,\dots,\,1,\,0)\,.
\end{equation}
Since the chain of orthogonal restrictions is multiplicity-free, the commuting family
\begin{equation}
    C^{(2)},\,C^{(3)},\,\dots,\,C^{(t)}
\end{equation}
has joint eigenspaces of dimension one, and these joint eigenvalues separate the GT patterns. Therefore, for every pattern $T\in \GT(\lambda)$, there is a one-dimensional projector $P_T$, obtained by Lagrange interpolation in the $C^{(m)}$'s. In practice, one may write
\begin{equation}
    P_T
    =
    \prod_{m=2}^t
    \prod_{\nu\in\mathcal S_m(T)\setminus\{\lambda^{(m)}\}}
    \frac{C^{(m)}-\chi_m(\nu)}
         {\chi_m(\lambda^{(m)})-\chi_m(\nu)},
\end{equation}
where $\mathcal S_m(T)$ is the finite set of permitted $m$-th rows that are compatible with the higher rows of $T$. Among all patterns of top row $\lambda$, there is a distinguished highest one, denoted $T_{\mathrm{hw}}(\lambda)$, obtained by choosing at every step the maximal allowed next row. We denote by
\begin{equation}
    P_{\lambda}^{\mathrm{hw}}=P_{T_{\mathrm{hw}}(\lambda)}
\end{equation}
the projector onto the corresponding highest-weight line. The desired matrix units are then
\begin{equation}
    X^{(\lambda)}_{T,T'}
    =
    Z_T\,P_{\lambda}^{\mathrm{hw}}\, Z_{T'}^{\dagger},
    \qquad
    \lambda\in \Lambda_{t,n},
    \quad
    T,T'\in \GT(\lambda)\,.
\end{equation}
After rescaling the operators $Z_T$, these operators satisfy the canonical matrix-unit relations
\begin{equation}
    X^{(\lambda)}_{T,T'}\,X^{(\mu)}_{S,S'}
    =
    \delta_{\lambda\mu}\,\delta_{T',S}\,
    X^{(\lambda)}_{T,S'}\,.
\end{equation}
That is, after normalization we can identify them with the canonical rank-one basis operators
\begin{equation}
    X^{(\lambda)}_{T,T'}=\ketbra{\lambda,T}{\lambda,T'}.
\end{equation}
Hence we can endow the algebra $A_{t,n}$ with the orthonormal basis
\begin{equation}
    A_{t,n}
    =
    \left\{
        X^{(\lambda)}_{T,T'}
        \,:\,
        \lambda\in\Lambda_{t,n},\ T,T'\in \GT(\lambda)
    \right\}\,.
\end{equation}
Lastly, we recall that, as a vector space
\begin{equation}
    \CC_{t,n} = A_{t,n}\oplus \Gamma_1A_{t,n}\,,
\end{equation}
where the sector $\Gamma_1A_{t,n}$ is obtained from the algebra $A_{t,n}$ by multiplying its elements with the parity operator $\Gamma_1$.

\section{Construction of the commutant basis for general Gaussian unitaries}\label{ap:generalgaussianconstr}

\subsection{$t=1$ case}
This case needs to be considered separately from $t\geq 2$ as the copy side algebra $\mathfrak{so}(1)=0$ is trivial, and there are no nontrivial orthogonal highest-weights to choose from. In particular, there are also no $\widetilde{Q}_{j,k}$ operators to realize as one needs $j$ strictly smaller than $k$. Therefore, the associative algebra is just the scalar algebra $A_{1,n}=\mathbb{C}\id$. Hence, we can complete the full commutant as
\begin{equation}
    \CC_{1,n}={\rm span}_{\mathbb{C}}\{\id,\Gamma\}\,.
\end{equation}

\subsection{$t=2$ case}\label{app:active2}

Now, the algebra acting on the copies is $\mathfrak{so}(2)$, and hence abelian. As such, the GT chain is again trivial since it terminates at the first term. However, unlike the $t=1$ scenario we here have a single bilinear operator $\widetilde{Q}_{1,2}$, from which we can define the Hermitian generator $M=-i\widetilde{Q}_{1,2}$ of the infinitesimal action of $\mathfrak{so}(2)$. The operator $M$ has eigenvalues $\{-n,-n+1,\cdots,n-1,n\}$, meaning that we can construct the projectors onto the irreps as
\begin{equation}
    P_m=\prod_{\substack{s= -n\\s\neq m}}^n\frac{M-s}{m-s}\,.
\end{equation}
Then, a basis of the full commutant is
\begin{equation}
    \CC_{2,n}={\rm span}_{\mathbb{C}}\{P_m,\Gamma P_m\,|\, m=-n,-n+1,\ldots,n\}\,.
\end{equation}
The dimension of the commutant is thus
\begin{equation}
    \dim(\CC_{2,n})=2(2n+1)\,.
\end{equation}

To finish this case, we point out an interesting fact. Namely, that the $t=2$ case for general fermionic Gaussian unitaries closely resembles the $t=1$ case for PP. First of all the algebras are abelian, and the GT chain is trivial. The analogy extends to a very useful application: The operator $M=-i\widetilde{Q}_{12}$ has the role of the operator $N=\widetilde{\Omega}_{11}$. This also means that $P_m$ of the general fermionic Gaussian unitaries framework has the same role as $P_r$ in the particle-preserving subgroup. Then, we can define the symmetric polynomials  
\begin{equation}
    E_k
    =
    \sum_{1\le \mu_1<\cdots<\mu_k\le 2n}
    y_{\mu_1}\cdots y_{\mu_k},
    \qquad
    k=0,1,\ldots,2n,
    \qquad
    E_0=\id\,,
\end{equation}
where $y_\mu=-i \widetilde{c}^{(1)}_\mu\widetilde{c}^{(2)}_\mu$ are hermitian operators satisfying $[y_\mu,y_\nu]=0$. 
Proceeding exactly as in the passive case, using $M=\frac{1}{2}\sum_\mu y_\mu$, we can then prove the relation through Krawtchouk polynomials
\begin{equation}
E_k=K_k(n-M;2n)\,,
\end{equation}
which also implies
\begin{equation}
    E_k=\sum_m K_k(n-m;2n)P_m\,.
\end{equation}
Now a key observation: The operators $E_k$ lead one directly to
\begin{equation}
E_{2k}=Q_{2k}^0,
\qquad
E_{2k+1}=-i\,Q_{2k+1}^1.
\end{equation}
\begin{equation}
Q_k^0=\sum_{\mu_1<\cdots<\mu_k} c_{\mu_1}\cdots c_{\mu_k}\otimes c_{\mu_1}\cdots c_{\mu_k},
\qquad
Q_k^1=\sum_{\mu_1<\cdots<\mu_k} \Gamma c_{\mu_1}\cdots c_{\mu_k}\otimes \, c_{\mu_1}\cdots c_{\mu_k}\,,
\end{equation}
which, up to convention, is the basis of the commutant provided in \cite{diaz2023showcasing} associated with Schur's lemma. The odd $Q^0_k$ and even $Q^1_k$ are recovered by multiplying by $\Gamma_1$. Notice that the final Majorana are not dressed, instead the operator $\Gamma$ which dresses the Majoranas, gets canceled for even $k$ while for odd $k$ we recover the linear map connecting isomorphic irreps in operator space. We have thus recovered the explicit relation between the basis $Q_k^{i}$ and the projector basis $P_m$. All of this in analogy with the first commutant of PP.

\subsection{$t=3$ case}\label{ap:activet3}

This case corresponds to the first one where the copy algebra, $\mathfrak{so}(3)$, is non-abelian. However, we will see that many calculations simplify from the fact that 
\begin{equation}
    \mathfrak{so}(3)\cong\mathfrak{su}(2)\,.
\end{equation}
Indeed, $\mathfrak{so}(3)$ has rank $1$ and the available irreps are labeled by the single spin weight $\lambda=0,1,\ldots,n$. Hence, we know that for each $\dim(W_\lambda^{\mathfrak{so}(3)})=2\lambda+1$. 
Then, the GT chain follows from 
\begin{equation}
    \mathbb{SO}(3)\downarrow\mathbb{SO}(2)\,,
\end{equation}
and the patterns for a given $\lambda$ are obtained from the $\mathbb{SO}(2)$ weight, i.e.,
\begin{equation}
    \GT(\lambda)=-\lambda,\lambda+1,\ldots,\lambda\,,\quad \text{such that} \quad |\GT(\lambda)|=2\lambda+1\,.
\end{equation}
As such, GT patterns are simply 
\begin{equation}
    (\lambda,m)=-\lambda,-\lambda+1,\ldots,\lambda\,.
\end{equation}

To obtain the commutant basis we define the Hermitian spin operators
\begin{equation}
    J_z=-i\widetilde{Q}_{1,2}\,,\quad J_+=\widetilde{Q}_{1,3}+i\widetilde{Q}_{2,3}\,,\quad J_-=-(\widetilde{Q}_{1,3}-i\widetilde{Q}_{2,3})\,,\label{eqap:ops-t3-so2n}
\end{equation}
which can be readily found to satisfy the $\mathfrak{su}(2)$ commutation relations. These allow us to define the ``total angular momentum'' operator
\begin{equation}
    J^2=J_z^2+\frac{1}{2}(J_+J_-+J_-J_+)\,,
\end{equation}
with eigenvalues $\lambda(\lambda+1)$. This then leads to the irrep projector
\begin{equation}
    P_\lambda=\prod_{\substack{\lambda'= 0\\\lambda'\neq \lambda}}^n\frac{J^2-\lambda'(\lambda'+1)}{\lambda(\lambda+1)-\lambda'(\lambda'+1)}\,,
\end{equation}
and the one dimensional ones
\begin{equation}
    P_{\lambda,m}=P_\lambda\prod_{\substack{m'= -n\\m'\neq m}}^n\frac{J_z-m'}{m-m'}\,.
\end{equation}

From the previous, we can then obtain the commutant basis operators follow by sandwiching ladder operators. However, we also find it convenient to provide a derivation of the commutant basis that closely resembles that of PP unitaries. Namely, explicitly finding  the actual highest-weight subspace inside the physical Hilbert space. The trick here is to introduce creation and annihilation operators for emerging fermions.

\medskip
\medskip
\begin{center}
    \textbf{Emerging fermions in the $t=3$ case}
\end{center}
\medskip
\medskip

Contrary to what happens in the PP unitary case, there is a priori no natural Fock basis of states satisfying $J_z|\psi\rangle=\lambda |\psi\rangle$ and $J_+|\psi\rangle=0$ such that we can explicitly and easily construct the GT projector. 
As a matter of fact, one crucial step in the passive case was the use of the vacuum and of the division of Fock space into fixed particle number sector. Notably, there is a crucial idea that allows us to recover a  particle-like picture by defining ``emergent'' fermions occupying the first two copies. For this, let us define the following operators
\begin{equation}
    f_\mu=\frac12\Big(\widetilde c_\mu^{\,(1)}+i\widetilde c_\mu^{\,(2)}\Big),\qquad
    f_\mu^\dagger=\frac12\Big(\widetilde c_\mu^{\,(1)}-i\widetilde c_\mu^{\,(2)}\Big),
\end{equation}
which satisfy genuine fermionic anticommutation relations 
\begin{equation}
    \{f_\mu,f_\nu^\dagger\}=\delta_{\mu\nu},\qquad \{f_\mu,f_\nu\}=0,\qquad \{f_\mu^\dagger,f_\nu^\dagger\}=0.
\end{equation} 

Notice that the mode index run from $a=1,\ldots , 2n$. 
In terms of these operators,
$
    -i\,\widetilde c_\mu^{\,(1)}\widetilde c_\mu^{\,(2)}=1-2f_\mu^\dagger f_\mu,
$
so that
\begin{equation}
    J_z=-i\widetilde{Q}_{1,2}=\frac12\sum_{\mu=1}^{2n}(1-2f_\mu^\dagger f_\mu)
    =n-N_f\,,
\end{equation}
where we defined the total $f$-number operator
\begin{equation}
    N_f=\sum_{\mu=1}^{2n} f_\mu^\dagger f_\mu\,.
\end{equation}
Thus the $J_z$-eigenspaces are precisely the fixed-$N_f$ sectors, and the possible weights are
\begin{equation}
    m=n-N_f,\qquad N_f=0,1,\dots,2n.
\end{equation}
Namely
\begin{equation}
    m=-n,-n+1,\dots,n-1,n.
\end{equation}
Using the $f_\mu$'s, the ladder operators of Eq,~\eqref{eqap:ops-t3-so2n} take the form
\begin{equation}
    J_+=\sum_{\mu=1}^{2n} f_\mu\,\widetilde c_\mu^{\,(3)},\qquad
    J_-=-\sum_{\mu=1}^{2n} f_\mu^\dagger\,\widetilde c_\mu^{\,(3)}.
\end{equation}
Therefore the active $t=3$ problem admits an emergent structure  where $N_f$ plays the role of the number of fermions sectors in the PP $t=2$ case. Specifically, we can see that the $f_\mu$ create a standard fermion on the first two copies of the physical space, while $\widetilde c_\mu^{\,(3)}$ inserts a  Majorana on the third copy. Then, $J_\pm$ create and destroy a standard fermion on copies one and two, and flips the parity sector on the third copy.

The advantage of these operators is that now we can look for a highest-weight state of spin $j$ which satisfies
\begin{equation}
    J_z|\psi\rangle=j\,|\psi\rangle,\qquad J_+|\psi\rangle=0.
\end{equation}
Since $J_z=n-R$, the natural procedure is as in the PP case. We fix $N_f$, solve $J_+|\psi\rangle=0$ in that sector to identify the corresponding highest-weight space and the corresponding projector. Then act with $J_-$.

Let us describe this procedure explicitly with the simple case of $R=0$.
Let $|0_f\rangle$ denote the $f$-vacuum, defined by
\begin{equation}
    f_\mu|0_f\rangle=0,\qquad a=1,\dots,2n.
\end{equation}
Then
\begin{equation}
    J_z|0_f\rangle=n\,|0_f\rangle.
\end{equation}
Moreover, for any state $|\chi\rangle_3$ in the third copy,
\begin{equation}
    J_+\big(|0_f\rangle\otimes |\chi\rangle_3\big)=0,
\end{equation}
since each $f_\mu$ annihilates $|0_f\rangle$. Therefore the full space
\begin{equation}
    \mathcal H^{\mathrm{hw}}_{j=n}
    =
    \Big\{
    |0_f\rangle\otimes |\chi\rangle_3
    \,:\,
    |\chi\rangle_3\in \mathcal H^{(3)}
    \Big\}
\end{equation}
is a highest-weight space of spin $j=n$. Its degeneracy is
\begin{equation}
    \dim \mathcal H^{\mathrm{hw}}_{j=n}=2^n.
\end{equation}
The projector onto this highest-weight space is thus
\begin{equation}
    P^{\mathrm{hw}}_{j=n}
    =
    |0_f\rangle\langle 0_f|\otimes \id_{(3)}.
\end{equation}
From this projector one may generate the full spin-$n$ block by repeated action of $J_-$ on the ket side and $J_+$ on the bra side, in direct analogy with the passive $t=2$ construction. 

 \subsection{$t=4$ case}

To finish, let us consider the $t=4$ case. Here, the GT construction can be simplified by realizing the algebra isomorphism
\begin{equation}
\mathfrak{so}(4)\cong \mathfrak{su}(2)_+\oplus \mathfrak{su}(2)_-.
\end{equation}
In particular, we can define the operators
\begin{equation}
    J_1^\pm=\frac12(\widetilde{Q}_{23}\pm \widetilde{Q}_{14}),
    \qquad
    J_2^\pm=\frac12(\widetilde{Q}_{31}\pm \widetilde{Q}_{24}),
    \qquad
    J_3^\pm=\frac12(\widetilde{Q}_{12}\pm \widetilde{Q}_{34})\,,
\end{equation}
which can be shown to satisfy the properties
\begin{equation}
    [J_i^+,J_j^+]=i\epsilon_{ijk}J_k^+,
    \qquad
    [J_i^-,J_j^-]=i\epsilon_{ijk}J_k^-,
    \qquad
    [J_i^+J_j^-]=0,
\end{equation}
so the $+$ and $-$ generators give the two commuting $\mathfrak{su}(2)$ algebras. These allow us to then define  the corresponding ladder operators
\begin{equation}
\label{ap-eq:t4-ladder-def}
    J_\pm^\pm = J_1^\pm \pm iJ_2^\pm.
\end{equation}

The total angular moment operators of the two commuting $\mathfrak{su}(2)$ factors are then
\begin{equation}
\label{ap-eq:t4-casimirs-def}
    J^2_\pm
    =
    (J_1^\pm)^2+(J_2^\pm)^2+(J_3^\pm)^2.
\end{equation}
Equivalently,
\begin{equation}
\label{ap-eq:t4-casimirs-ladder}
    J^2_\pm
    =
    (J_3^\pm)^2+\frac12\bigl(J_+^\pm J_-^\pm + J_-^\pm J_+^\pm\bigr)=
    \frac14\sum_{1\le r<s\le 4}\widetilde Q_{rs}^2
    \;\pm\;
    \frac12\Big(
        \widetilde Q_{23}\widetilde Q_{14}
        +
        \widetilde Q_{31}\widetilde Q_{24}
        +
        \widetilde Q_{12}\widetilde Q_{34}
    \Big).
\end{equation}

The irreducible $\mathfrak{so}(4)$-modules are therefore labeled by a pair of spins $(j_+,j_-)$, with
\begin{equation}
\label{ap-eq:t4-casimirs-evalues}
    J^2_+=j_+(j_++1)\,I,
    \qquad
    J^2_-=j_-(j_-+1)\,I.
\end{equation}
Hence, in this $t=4$ realization, the GT patterns reduce to the choice of $J_z$ magnetization for the two commuting $\mathfrak{su}(2)$ algebras
\begin{equation}
    m_+=-j_+,-j_++1,\dots,j_+,
    \qquad
    m_-=-j_-,-j_-+1,\dots,j_-.
\end{equation}
Thus the GT basis of the copy irrep is naturally indexed by pairs $(m_+,m_-)$. 

The previous thus allows us to define the projector onto the highest-weight line
\begin{equation}
    P_{j_+,j_-}^{\mathrm{hw}}\,.
\end{equation}
Then, since the two $\mathfrak{su}(2)$ commute, the lowering operator  is
\begin{equation}
\label{ap-eq:t4-lowering}
    Z_{m_+,m_-}
    =
    (J_-^+)^{\,j_+-m_+}\,
    (J_-^-)^{\,j_--m_-},
\end{equation}
and the corresponding elements of the commutant will be \begin{equation}
\label{ap-eq:t4-matrix-units}
    X^{(j_+,j_-)}_{m_+,m_-;\,m_+',m_-'}
    =
    (J_-^+)^{\,j_+-m_+}\,
    (J_-^-)^{\,j_--m_-}\,
    P_{j_+,j_-}^{\mathrm{hw}}\,
    (J_+^+)^{\,j_+-m_+'}\,
    (J_+^-)^{\,j_--m_-'}.
\end{equation}

\section{Average PP Gaussian stabilizer purity in a fixed particle sector}
\label{ap:avg_pp_stab_sector}

In this section we show how to derive the exact closed formula for the average stabilizer purity over the set of fermionic Gaussian states with fixed particle number presented in the main text.
Recall that the PP Gaussian unitaries correspond to the group $\U(n)$, acting on the full Fock space as
\begin{equation}
    \FC_n=\Lambda(\CBB^n)=\bigoplus_{r=0}^n \Lambda^r(\CBB^n)\,.
\end{equation}
For each particle number $r\in\{0,\dots,n\}$, we define
\begin{equation}
    \ket{r}=a_1^\dagger a_2^\dagger \cdots a_r^\dagger \ket{0}
    \in \Lambda^r(\CBB^n)\,,
\end{equation}
as the reference $r$-particle Slater determinant, whose orbit under the PP Gaussian unitaries generates the whole set of fermionic Gaussian states with fixed particle number $r$.
Furthermore, we denote 
\begin{align}
    &R_r(U)\,,\qquad U\in\U(n)\,,\\
    &R_r(U)(v_1\wedge\dots\wedge v_r) = Uv_1\wedge\dots\wedge Uv_r\,,
\end{align}
the representation of $\U(n)$ acting on $\Lambda^r(\CBB^n)$, where we identify the group element $U\in\U(n)$ with its standard representation acting on $\CBB^n$.
To compute the average linearized stabilizer entropy over that set, we need to study
\begin{equation}
\label{ap-eq:S4_pp_sector_def}
    \EBB_{U\sim \U(n)} S_4(R_r(U)\ket{r})
    =
    \EBB_{U\sim \U(n)}\sum_{P\in\PC_n}\frac{1}{4^n}\Tr[R_r(U)\ketbra{r}{r}R_r(U)^\dagger P]^4\,.
\end{equation}

Before recalling the result and proving it, let us introduce some auxiliary notations. First, we denote by
\begin{equation}
    d_{n,r}=\dim\bigl(V^{\U(n)}_{(4^r)}\bigr)\,,
\end{equation}
the dimension of the $\U(n)$ irrep associated with $(4^r)$, the rectangular tableau with $r$ rows and $4$ columns.
Using the hook-content formula~\cite{macdonald1998symmetric}, one finds
\begin{equation}
\label{ap-eq:d_nr_formula}
    d_{n,r}
    =
    \prod_{i=1}^r\prod_{j=1}^4
    \frac{n+j-i}{r+5-i-j}\,.
\end{equation}
Furthermore, given a partition $\lambda=(\lambda_1,\lambda_2,\dots)$ of $r$, we define its double as
\begin{equation}
    2\lambda=(2\lambda_1,2\lambda_2,\dots)\,,
\end{equation}
and we denote by $H(2\lambda)$ the hook-product of the Young diagram $2\lambda$.
For later ease of notation, we also introduce the generalized Pochhammer symbol
\begin{equation}
    (a)_\lambda
    =
    \prod_{j\geq 1}
    \left(a-\frac{j-1}{2}\right)_{\lambda_j}\,,
    \qquad
    (x)_k=x(x+1)\cdots(x+k-1)\,.
\end{equation}
Lastly, we define
\begin{equation}
\label{apap-eq:ck_def}
    c_k
    =
    4^k
    \sum_{\lambda\vdash k,\ \ell(\lambda)\leq 3}
    \frac{(2k)!}{H(2\lambda)}
    \frac{\left(\frac{3}{2}\right)_\lambda}{(3)_\lambda}\,,
\end{equation}
for $\ell(\lambda)$ the number of rows of the Young diagram associated to the partition $\lambda$, namely the length of the partition itself as a tuple.

We are now ready to recall the result and then proceed to prove it.
\begin{theorem}[Average stabilizer purity in $\Lambda^r(\CBB^n)$]
\label{ap-th:avg_pp_stab_sector}
For any given total number of fermionic modes $n\geq 0$ and any fixed particle sector thereof $0\leq r\leq n$,
\begin{equation}
\label{ap-eq:avg_pp_stab_sector_formula_1}
    \mathbb{E}_{U\sim\U(n)}
    S_4\!\left(R_r(U)\ket{r}\right)
    =
    \frac{1}{2^n\,d_{n,r}}
    \sum_{k=0}^{\min(r,n-r)}
    \frac{n!}{(n-r-k)!\,(2k)!\,(r-k)!}\,c_k.
\end{equation}
Equivalently, substituting Eq.~\eqref{apap-eq:ck_def},
\begin{equation}
\label{ap-eq:avg_pp_stab_sector_formula_2}
    \mathbb{E}_{U\sim\U(n)}
    S_4\!\left(R_r(U)\ket{r}\right)
    =
    \frac{1}{2^n\,d_{n,r}}
    \sum_{k=0}^{\min(r,n-r)}
    \frac{n!}{(n-r-k)!\,(r-k)!}
    4^k
    \sum_{\lambda\vdash k,\ \ell(\lambda)\leq 3}
    \frac{1}{H(2\lambda)}
    \frac{\left(\frac{3}{2}\right)_\lambda}{(3)_\lambda}\,.
\end{equation}
\end{theorem}

\begin{proof}
We start by defining the uniform superposition of fourth tensor powers of Paulis as
\begin{equation}
\label{ap-eq:Qn_pp_def}
    Q_n=\sum_{P\in\PC_n} P^{\otimes 4}\,.
\end{equation}
This way, given any pure state $\ket{\psi}$ we have
\begin{equation}
\label{ap-eq:replica_identity_pp}
    \sum_{P\in\PC_n}\Tr[\ketbra{\psi}{\psi}P]^4
    =
    \Tr[\ketbra{\psi}{\psi}^{\otimes 4}Q_n]\,.
\end{equation}
Hence, for the case of interest $\ket{\psi}=R_r(U)\ket{r}$, and averaging over $U\sim\U(n)$ we find
\begin{equation}
\label{ap-eq:avg_pp_stab_as_trace}
    \mathbb{E}_{U\sim\U(n)}
    S_4\!\left(R_r(U)\ket{r}\right)
    =
    \frac{1}{4^n}
    \Tr\!\left[
    \left(
    \mathbb{E}_{U\sim\U(n)}
    R_r(U)^{\otimes 4}
    \ketbra{r}{r}^{\otimes 4}
    R_r(U)^{\dagger\otimes 4}
    \right)
    Q_n
    \right]\,.
\end{equation}
Now, $V_r=\Lambda^r(\CBB^n)$ is the irreducible $\U(n)$-module of highest-weight $(1^r)$, i.e. the one associated with the Young diagram consisting of a single height-$r$ column. If we denote by $\{e_j\}_{j=1}^n$ the canonical basis of $\CBB^n$, the highest-weight vector $v_r\in V_r$ is $v_r=e_1\wedge\dots\wedge e_r$, thus corresponding to $\ket{r}$. Consequently,
\begin{equation}
    v_r^{\otimes 4}\in V_r^{\otimes 4}
\end{equation}
is a highest-weight vector of weight $(4^r)$, since under tensor products weights add~\cite{kumar2010tensor}. 
Now consider the irrep $V_{(4^r)}$ which appears, with multiplicity one~\cite{kumar2010tensor} in the decomposition of $V_r^{\otimes 4}$. This irrep is exactly the one generated by the orbit of the highest-weight vector $v_r^{\otimes 4}$, i.e. $V_{(4^r)}=\Span\{R_r(U)^{\otimes 4}v_r^{\otimes 4}\,|\, U\in\U(n)\}$. As can be checked via the GT machinery developed in Appendix~\ref{ap:PP_GT_method}, the copy-side irrep associated to $(4^r)$ is $(r,r,r,r)$ which is indeed one-dimensional. Since $V_{(4^r)}$ occurs with multiplicity one, and given the identification of $\ket{r}$ with $v_r$, the twirl $\mathbb{E}_{U\sim\U(n)} R_r(U)^{\otimes 4} \ketbra{r}{r}^{\otimes 4} R_r(U)^{\dagger\otimes 4}$ must be proportional to the identity operator over $V_{(4^r)}$ itself. Namely, as an operator in $\End(\HC^{\otimes 4})$ the twirl projects onto the multiplicity free isotypic component $(4^r)$ of $\U(n)$.
Hence, denoting by $P_{(4^r)}$ said projector, we have
\begin{equation}
\label{ap-eq:physical_projector_average}
    \mathbb{E}_{U\sim\U(n)}
    R_r(U)^{\otimes 4}
    \ketbra{r}{r}^{\otimes 4}
    R_r(U)^{\dagger\otimes 4}
    =
    \frac{P_{(4^r)}}{d_{n,r}}\,,
\end{equation}
where the normalization factor arises since the twirl must preserve the trace of the input state $\ketbra{r}{r}^{\otimes 4}$, and we recall that the definition of $d_{n,r}$ is given in Eq.~\eqref{ap-eq:d_nr_formula}.
Substituting into Eq.~\eqref{ap-eq:avg_pp_stab_as_trace}, we get
\begin{equation}
\label{ap-eq:avg_pp_stab_projector_overlap}
    \mathbb{E}_{U\sim\U(n)}
    S_4\!\left(R_r(U)\ket{r}\right)
    =
    \frac{1}{4^n\,d_{n,r}}
    \Tr[Q_n P_{(4^r)}]\,.
\end{equation}

Thus, we have to compute the overlap $\Tr[Q_n P_{(4^r)}]$. We notice that the operator $Q_n$ can be factorized over the qubits' indices as follows
\begin{equation}
\label{ap-eq:Qn_factorization_pp}
    Q_n=q^{\otimes n},
    \qquad
    q=I^{\otimes 4}+X^{\otimes 4}+Y^{\otimes 4}+Z^{\otimes 4}\,.
\end{equation}

Now, consider the standard exterior-algebra identification~\cite{brosnan2019notes,cautis2014webs}
\begin{equation}
\label{ap-eq:fermionic_identification_pp}
    \left(\Lambda \CBB^n\right)^{\otimes 4}
    \cong
    \Lambda(\CBB^n\otimes \CBB^4)
    \cong
    \left(\Lambda \CBB^4\right)^{\otimes n}\,.
\end{equation}
Notice that through this identification one can freely exchange the roles of physical modes and tensor copies.
This identification is the core of the skew Howe duality~\cite{cautis2014webs}, which as we saw implies that the diagonal physical action of $\U(n)$ commutes with the natural copy-side action of $\U(4)$.

More precisely, we here recall that the skew Howe duality gives the multiplicity-free decomposition
\begin{equation}
\label{ap-eq:skew_howe_pp}
    \Lambda(\CBB^n\otimes \CBB^4)
    \cong
    \bigoplus_{\lambda\subseteq n^4}
    V_{\lambda^T}^{\U(n)}\otimes V_{\lambda}^{\U(4)}\,,
\end{equation}
where $\lambda^T$ denotes the transpose partition.
Thus, as stated previously, for the partition $\lambda^T=(4^r)$, one has
\begin{equation}
    \lambda=(r,r,r,r)\,,
\end{equation}
and the corresponding $\U(4)$-representation is the one-dimensional determinant representation
\begin{equation}
    V_{(r,r,r,r)}^{\U(4)}\cong {\det}^r\,.
\end{equation}
Therefore, we can resort to character orthogonality~\cite{kamnitzer2011representation,hall2013lie} to express the projector onto $V_{(4^r)}$ as the projector onto the copy side as
\begin{equation}
\label{ap-eq:copy_projector_def}
    \Pi_r^{\rm copy}
    =
    \int_{\U(4)} \overline{\det(g)^r}\,\rho_n(g)\,d\mu(g),
\end{equation}
where we used $\rho_n$ to denote the copy-side representation on $\left(\Lambda \CBB^4\right)^{\otimes n}$, and $d\mu(g)$ is the normalized Haar measure on $\U(4)$. Notice that we used $g$ for the elements of $\U(4)$ to distinguish them from $U$ which we used for the physical $\U(n)$.
Since the decomposition in Eq.~\eqref{ap-eq:skew_howe_pp} is multiplicity-free, the projector $P_{(4^r)}$ on the physical side and the projector $\Pi_r^{\rm copy}$ on the copy side are the same operator once both are viewed inside the common space in Eq.~\eqref{ap-eq:fermionic_identification_pp}. This is easily seen by noticing that the representation $\rho_n$ only acts on the copy side, hence it's proportional to the identity on the irreps of $\U(n)$. Then, by character orthogonality carrying out the integral in Eq.~\eqref{ap-eq:copy_projector_def} kills the support on all irreps but $(r,r,r,r)$, recovering $P_{(4^r)}$.
Hence, using also the factorization in Eq.~\eqref{ap-eq:Qn_factorization_pp}
\begin{equation}
\label{ap-eq:trace_qn_copy}
    \Tr[Q_n P_{(4^r)}]
    =
    \Tr[q^{\otimes n}\Pi_r^{\rm copy}]\,.
\end{equation}

If we denote by $\rho$ the local $\U(4)$ action on $\Lambda \CBB^4$, we can then write $\rho_n(g)=\rho(g)^{\otimes n}$.
With this, and substituting Eq.~\eqref{ap-eq:copy_projector_def} into Eq.~\eqref{ap-eq:trace_qn_copy}, we can express the wanted overlap as
\begin{equation}
\label{ap-eq:trace_as_local_integral}
    \Tr[Q_n P_{(4^r)}]
    =
    \int_{\U(4)} \overline{\det(g)^r}\,F(g)^n\,d\mu(g)\,,
\end{equation}
where we defined
\begin{equation}
\label{ap-eq:Fg_def}
    F(g)=\Tr_{\Lambda \CBB^4}[q\,\rho(g)]\,,
\end{equation}
the local copy overlap. Next, we focus on the latter. Since we want to project each local term $q$ onto $\Lambda \CBB^4$, we ask ourselves what the action of $q$ is on the latter. 
Let $\{e_S\}_{S\subseteq [4]}$ be the standard occupation (computational) basis of $\Lambda \CBB^4$, indexed by subsets $S\subseteq [4]$.
Under the local Jordan--Wigner identification, one checks directly that
\begin{equation}
\label{ap-eq:X_action_local}
    X^{\otimes 4} e_S = e_{S^c},
\end{equation}
\begin{equation}
\label{ap-eq:Y_action_local}
    Y^{\otimes 4} e_S = (-1)^{|S|} e_{S^c},
\end{equation}
and
\begin{equation}
\label{ap-eq:Z_action_local}
    Z^{\otimes 4} e_S = (-1)^{|S|} e_S,
\end{equation}
where $S^c=[4]\setminus S$.
Indeed, $X$ flips computational basis states, $Z$ contributes a sign $(-1)$ on an occupied qubit, and for $Y$ one uses $Y\ket{0}=i\ket{1}$ and $Y\ket{1}=-i\ket{0}$, which yields the net factor $(-1)^{|S|}$ (which is equal to $(-1)^{|S^c|}$) on four copies, by $(-i)^4=1$.

Combining Eq.~\eqref{ap-eq:X_action_local} with Eq.~\eqref{ap-eq:Z_action_local}, we find
\begin{equation}
\label{ap-eq:q_action_local}
    q\,e_S
    =
    \begin{cases}
        0, & |S| \text{ odd},\\
        2e_S+2e_{S^c}, & |S| \text{ even}\,.
    \end{cases}
\end{equation}
Now let us parametrize the copy group $\U(4)$ as
\begin{equation}
    g=zh,
    \qquad
    z\in \U(1),
    \quad
    h\in \SU(4)\,.
\end{equation}
Indeed, the unitary group $\U(d)$ can always be decomposed as $\U(d)\cong (\SU(d)\times\U(1))/\mathbb{Z_d}$. 
Moreover, the action of $zh$ on $\Lambda^r(\CBB^4)$ reads
\begin{equation}
\label{ap-eq:zh_on_degree_r}
    \rho(zh)\big|_{\Lambda^r(\CBB^4)}
    =
    z^r\,\Lambda^r(h)\,
\end{equation}
since the phases $z$ factorize and accumulate as a global phase.
Now, since $q$ annihilates the odd-occupied sectors by Eq.~\eqref{ap-eq:q_action_local}, only the occupation sectors $|S|\in\{0,2,4\}$ can contribute to the trace. Hence we can write
\begin{equation}
    F(zh)=\Tr_{\Lambda \CBB^4}[q\,\rho(zh)]
    =
    \Tr_{\Lambda^0 \CBB^4}[q\,\rho(zh)]+\Tr_{\Lambda^2 \CBB^4}[q\,\rho(zh)]+\Tr_{\Lambda^4 \CBB^4}[q\,\rho(zh)]\,.
\end{equation}
Notice now that $\Lambda^0 \CBB^4$ is the trivial irrep, where $\rho(zh)=\Lambda^0(zh)=1$. Also, its basis is given by $\ket{0000}$ alone, hence one readily gets $\Tr_{\Lambda^0 \CBB^4}[q\,\rho(zh)]=2$, where the only contributions come from $I^{\otimes 4}$ and $Z^{\otimes 4}$. Then, $\Lambda^4 \CBB^4$ is the one dimensional determinant representation $\rho(zh)=\Lambda^4(zh)=\det(zh)=z^4$, and again one easily gets $\Tr_{\Lambda^4 \CBB^4}[q\,\rho(zh)]=2z^4$.
On the two-particles sector $\Lambda^2(\CBB^4)$, consider the involution $J=X^{\otimes4}:\Lambda^2(\CBB^4)\rightarrow\Lambda^2(\CBB^4)$
\begin{equation}
\label{ap-eq:J_def_local}
    J e_S = e_{S^c},
    \qquad |S|=2\,.
\end{equation}
Then Eq.~\eqref{ap-eq:q_action_local} implies that, on the degree-two sector, $q$ acts as $2(I+J)$.
Combining these observations, we arrive at
\begin{equation}
\label{ap-eq:Fzh_formula}
    F(zh)
    =
    2\left(1+z^4+z^2 M(h)\right)\,,
\end{equation}
where we introduced
\begin{equation}
\label{ap-eq:Mh_def}
    M(h)=\Tr_{\Lambda^2(\CBB^4)}\!\left[(I+J)\Lambda^2(h)\right]\,.
\end{equation}
Now let us go back to the factorization $\U(4)\cong(\SU(4)\times U(1))/\Bbb Z_4$. With the expression in Eq.~\eqref{ap-eq:Fzh_formula} we can show that we can work with the larger $\SU(d)\times\U(1)$. Indeed, since the map
\begin{equation}
    \phi:\U(1)\times \SU(4)\to \U(4),\qquad (z,h)\mapsto zh\,,
\end{equation}
is surjective with kernel
\begin{equation}
    \Delta\Bbb Z_4=\{(\omega,\omega^{-1}I_4):\omega^4=1\}\,,
\end{equation}
the correct group isomorphism indeed should read $\U(4)\cong (\U(1)\times \SU(4))/\Delta\Bbb Z_4$. Nonetheless, the integrand in Eq.~\eqref{ap-eq:trace_as_local_integral} is invariant under the kernel $\Delta\Bbb Z_4$, since
\begin{align}
    \overline{\det(zh)^r} = \overline{z}^{4r} &\to \overline{\omega}^{4r} \overline{z}^{4r} = \overline{z}^{4r}\\
    z^4 &\to \omega^4 z^4 = z^4\\
    z^2 &\to \omega^2 z^2 \\
    M(h) &\to M(\omega^{-1}h) = \omega^{-2}M(h)\,,
\end{align}
and the last two lines are multiplied together. In the last line we used that $\Lambda^2(\omega^{-1}h)=\omega^{-2}\Lambda^2(h)$.
Hence, we may evaluate the $\U(4)$ integral using the normalized product Haar measure on $\U(1)\times \SU(4)$.

Substituting Eq.~\eqref{ap-eq:Fzh_formula} into Eq.~\eqref{ap-eq:trace_as_local_integral}, we obtain
\begin{equation}
\label{ap-eq:trace_after_local_eval}
    \Tr[Q_n P_{(4^r)}]
    =
    2^n
    \int_{\U(1)\times \SU(4)}
    z^{-4r}\left(1+z^4+z^2M(h)\right)^n\,d\mu(z)\,d\mu(h)\,.
\end{equation}

We now expand the $n$-th power.
Choose $n_0$ factors equal to $1$, $n_2$ factors equal to $z^2M(h)$, and $n_4$ factors equal to $z^4$, with
\begin{equation}
\label{ap-eq:abc_sum}
    n_0 + n_2 + n_4=n\,.
\end{equation}
The corresponding monomial thus contributes
\begin{equation}
    z^{2n_2+4n_4-4r}M(h)^{n_2}\,.
\end{equation}
When carrying out the $\U(1)$ integral, only the terms with a total zero power of $z$ will survive, namely
\begin{equation}
\label{ap-eq:z_balance}
    2n_2+4n_4=4r\,.
\end{equation}
Since $4(r-n_4)$ is a multiple of four, we need to have $n_2=2k$ for some integer $k\geq 0$.
In turn, Eq.~\eqref{ap-eq:z_balance} becomes $n_4=r-k$, and Eq.~\eqref{ap-eq:abc_sum} gives $n_0=n-r-k$.
Using these, we get to
\begin{equation}
\label{ap-eq:trace_sum_with_Jk}
    \Tr[Q_n P_{(4^m)}]
    =
    2^n
    \sum_{k=0}^{\min(r,n-r)}
    \frac{n!}{(n-r-k)!\,(2k)!\,(r-k)!}\,J_k \,,
\end{equation}
where we defined the leftover integral as
\begin{equation}
\label{ap-eq:Jk_def}
    J_k=\int_{\SU(4)} M(h)^{2k}\,d\mu(h)\,.
\end{equation}

We are left with carrying out the integral $J_k$.
The two-particles sector $\Lambda^2(\CBB^4)$ is six-dimensional and it carries the representation $\Lambda^2$ of the local copy group $\SU(4)$. This representation preserves the symmetric bilinear form $B$ defined by
\begin{equation}
\label{ap-eq:bilinear_form_wedge2}
    \alpha\wedge\beta
    =
    B(\alpha,\beta)\, e_1\wedge e_2\wedge e_3\wedge e_4,
    \qquad
    \alpha,\beta\in\Lambda^2(\CBB^4)\,.
\end{equation}
Namely, $B(\alpha,\beta)$ extracts the coefficient of $\alpha\wedge\beta$ over the top form $e_1\wedge e_2\wedge e_3\wedge e_4$. Notice that $B$ is symmetric because the wedge of two even forms is.
The invariance of $B$ under the irrep $\Lambda^2$ of $\SU(4)$ is trivial, since= if $h\in \SU(4)$ then
\begin{align}
    B(\Lambda^2(h)\alpha, \Lambda^2(h)\beta) e_1\wedge e_2\wedge e_3\wedge e_4
    &=
    \Lambda^2(h)\alpha \wedge \Lambda^2(h)\beta
    \\
    &=\Lambda^4(h)(\alpha\wedge\beta)\\
    &=
    B(\alpha,\beta) \Lambda^4(h)(e_1\wedge e_2\wedge e_3\wedge e_4)\\
    &=
    \det(h)B(\alpha,\beta) e_1\wedge e_2\wedge e_3\wedge e_4\\
    &=
    B(\alpha,\beta) e_1\wedge e_2\wedge e_3\wedge e_4\,,
\end{align}
because $\det(h)=1$.
Therefore, in a basis orthonormal with respect to $B$, the matrix of $\Lambda^2(h)$ is an element of $\SO(6)$.\footnote{Strictly speaking, $\Lambda^2$ gives a representation into $\SO(6,\mathbb{C})$, since it preserves a nondegenerate symmetric bilinear form on the complex space $\Lambda^2(\mathbb{C}^4)$. Upon choosing the standard real form, this becomes the usual $6$-dimensional vector representation of $\SO(6)$, consistent with the exceptional isomorphism $\SU(4)\cong \Spin(6)$~\cite{lawson2016spin,fulton1991representation}.}

Now choose such a basis so that the involution $J$ in Eq.~\eqref{ap-eq:J_def_local} is diagonal.
Since $J^2=I$ and $\Tr(J)=0$ on the six-dimensional two-particles space, it has eigenvalues $+1$ and $-1$, each with multiplicity $3$.
Thus, after a suitable change of basis, we can write
\begin{equation}
\label{ap-eq:J_diagonal}
    J=
    \begin{pmatrix}
        I_3 & 0\\
        0 & -I_3
    \end{pmatrix},
    \qquad
    P=\frac{I+J}{2}
    =
    \begin{pmatrix}
        I_3 & 0\\
        0 & 0
    \end{pmatrix}\,.
\end{equation}
Now consider
\begin{equation}
    O_h\in \SO(6)
\end{equation}
as the matrix representation of $\Lambda^2(h)$ in this basis.
Then, from Eq.~\eqref{ap-eq:Mh_def},
\begin{equation}
\label{ap-eq:Mh_as_trace_PO}
    M(h)
    =
    \Tr[(I+J)O_h]
    =
    2\Tr(P O_h).
\end{equation}

At this point we use the exceptional isomorphism~\cite{lawson2016spin,fulton1991representation}
\begin{equation}
    \SU(4)\cong \Spin(6)\,,
\end{equation}
under which $\Lambda^2$ is precisely the standard double-cover map onto $\SO(6)$.
As a consequence, the normalized Haar measure on $\SU(4)$ is mapped to the normalized Haar measure on $\SO(6)$.
Therefore, Eq.~\eqref{ap-eq:Jk_def} becomes
\begin{equation}
\label{ap-eq:Jk_as_mur}
    J_k
    =
    4^k \mu_k\,,
\end{equation}
for
\begin{equation}
\label{ap-eq:muk_def}
    \mu_k
    =
    \int_{\SO(6)} \Tr(P O)^{2k}\,d\mu(O)\,,
\end{equation}
Combining Eqs.~\eqref{ap-eq:avg_pp_stab_projector_overlap}, \eqref{ap-eq:trace_sum_with_Jk}, and \eqref{ap-eq:Jk_as_mur}, we finally arrive at
\begin{equation}
\label{ap-eq:avg_pp_stab_before_moments}
    \mathbb{E}_{U\sim\U(n)}
    S_4\!\left(R_r(U)\ket{r}\right)
    =
    \frac{1}{2^n\,d_{n,r}}
    \sum_{k=0}^{\min(r,n-r)}
    \frac{n!}{(n-r-k)!\,(2k)!\,(r-k)!}\,
    4^k\mu_k\,.
\end{equation}

All that is left is to compute the moments $\mu_k$ in Eq.~\eqref{ap-eq:muk_def}.
Consider the generating function
\begin{equation}
\label{ap-eq:Phi_def}
    \Phi(t)
    =
    \int_{\SO(6)} e^{t\Tr(P O)}\,d\mu(O)\,.
\end{equation}
Since $-I_6\in \SO(6)$ and $\Tr(P(-O))=-\Tr(P O)$, all odd moments vanish.
Hence
\begin{equation}
\label{ap-eq:Phi_even_expansion}
    \Phi(t)
    =
    \sum_{k=0}^\infty \frac{\mu_k}{(2k)!}\,t^{2k}.
\end{equation}

On the other hand, the classical Herz formula~\cite{herz1955bessel} allows us to express $\Phi(t)$ as the hypergeometric function of matrix argument
\begin{equation}
\label{ap-eq:Herz_formula}
    \Phi(t)
    =
    {}_0F_1\!\left(3;\frac{t^2}{4}P\right)\,.
\end{equation}
We now expand the right-hand side in real zonal polynomials~\cite{gross1987special}.
If $C_\lambda$ denotes the zonal polynomial indexed by a partition $\lambda$, then
\begin{equation}
\label{ap-eq:zonal_expansion}
    {}_0F_1(b;X)
    =
    \sum_{\lambda}
    \frac{C_\lambda(X)}{|\lambda|!\,(b)_\lambda}.
\end{equation}
Using the homogeneity $C_\lambda(\alpha X)=\alpha^{|\lambda|}C_\lambda(X)$, Eq.~\eqref{ap-eq:Herz_formula} becomes
\begin{equation}
\label{ap-eq:Phi_zonal_expansion_explicit}
    \Phi(t)
    =
    \sum_{k=0}^\infty
    \frac{t^{2k}}{4^k\,k!}
    \sum_{\lambda\vdash k}
    \frac{C_\lambda(P)}{(3)_\lambda}\,.
\end{equation}
Since $P$ is a rank-$3$ projector, only partitions with at most $3$ parts can contribute to the sum above.
Moreover, on the support of $P$, the matrix $P$ acts as the identity, hence
\begin{equation}
\label{ap-eq:C_lambda_P_equals_I3}
    C_\lambda(P)=C_\lambda(I_3)
    \qquad
    \text{for all }\lambda \text{ with } \ell(\lambda)\leq 3.
\end{equation}
The standard specialization formula~\cite{gross1987special} for real zonal polynomials gives
\begin{equation}
\label{ap-eq:zonal_specialization}
    C_\lambda(I_3)
    =
    4^{|\lambda|}\,|\lambda|!\,
    \frac{\left(\frac{3}{2}\right)_\lambda}{H(2\lambda)}\,.
\end{equation}
Substituting Eqs.~\eqref{ap-eq:C_lambda_P_equals_I3} and \eqref{ap-eq:zonal_specialization} into Eq.~\eqref{ap-eq:Phi_zonal_expansion_explicit}, and then comparing the coefficient of $t^{2k}$ with Eq.~\eqref{ap-eq:Phi_even_expansion}, yields
\begin{equation}
\label{ap-eq:muk_closed_form}
    \frac{\mu}{(2k)!}
    =
    \sum_{\lambda\vdash k,\ \ell(\lambda)\leq 3}
    \frac{1}{H(2\lambda)}
    \frac{\left(\frac{3}{2}\right)_\lambda}{(3)_\lambda}\,,
\end{equation}
or equivalently
\begin{equation}
\label{ap-eq:muk_closed_form_2}
    \mu_k
    =
    \sum_{\lambda\vdash k,\ \ell(\lambda)\leq 3}
    \frac{(2k)!}{H(2\lambda)}
    \frac{\left(\frac{3}{2}\right)_\lambda}{(3)_\lambda}\,.
\end{equation}

Lastly, define $c_k=4^k\mu_k$.
By Eq.~\eqref{ap-eq:muk_closed_form_2}, this is exactly Eq.~\eqref{apap-eq:ck_def}.
Substituting $4^k\mu_k=c_k$ into Eq.~\eqref{ap-eq:avg_pp_stab_before_moments} proves Eq.~\eqref{ap-eq:avg_pp_stab_sector_formula_1}, while inserting Eq.~\eqref{ap-eq:muk_closed_form_2} directly proves Eq.~\eqref{ap-eq:avg_pp_stab_sector_formula_2}.
This completes the proof.
\end{proof}

\section{Average Gaussian linearized stabilizer entropy}
\label{ap:avg_gaussian_stab}

In this section we prove the result stated in the main text for the average linear stabilizer entropy of fermionic Gaussian states, which we here recall for convenience
\begin{equation}
\label{ap-eq:gaussian_stab}
    \mathbb{E}_{U\sim\Spin(2n)}S_4(U\ket{0})=\mathbb{E}_{U\sim\Spin(2n)} \sum_{P\in\mathcal{P}_n}\frac{1}{4^n}\Tr[U\ketbra{0}{0}U^\dagger P]^4 = \frac{1}{\Cat_{n+1}}\,,
\end{equation}
where $\PC_n=\{I,X,Y,Z\}^{\otimes n}$ is the $n$-qubit Pauli group, the Catalan numbers are given by $\Cat_m=\frac{1}{m+1}\binom{2m}{m}$, and we are identifying the element of the $\Spin(2n)$ group with its spinor representation for ease of notation.

Let us start by using the standard replica trick of the trace to write, for any state $\rho$
\begin{equation}
    \sum_{P\in\mathcal{P}_n} \Tr[\rho P]^4=\Tr[\rho^{\otimes 4}Q_n]\,,\qquad Q_n=\sum_{P\in\mathcal{P}_n}P^{\otimes 4}\,.
\end{equation}
Then, by linearity of the trace, we can write
\begin{equation}
    \mathbb{E}_{U\sim\Spin(2n)} \sum_{P\in\mathcal{P}_n}\frac{1}{4^n}\Tr[U\ketbra{0}{0}U^\dagger P]^4 = 
    \frac{1}{4^n}\Tr\left[\left(\mathbb{E}_{U\sim\Spin(2n)} U^{\otimes 4} \ketbra{0}{0}^{\otimes 4} (U^\dagger)^{\otimes 4}\right)Q_n\right]\,
\end{equation}
and the first factor in the right-hand side is by definition the fourth-order twirl of four copies of the all-zero state
\begin{equation}
    \mathbb{E}_{U\sim\Spin(2n)} U^{\otimes 4} \ketbra{0}{0}^{\otimes 4} (U^\dagger)^{\otimes 4} = \mathcal{T}^{(4)}(\ketbra{0}{0}^{\otimes 4})\equiv\tau_n^{(4)}\,.
\end{equation}
Therefore, the evaluation of the average Gaussian stabilizer purity is equivalent to computing the overlap
\begin{equation}
    \mathbb{E}_{U\sim\Spin(2n)}S_4(U\ket{0}) = \frac{1}{4^n}\Tr[\tau_n^{(4)} Q_n]\,.
\end{equation}
Since $\ket{0}\in\HC_+$, which follows immediately from $Z^{\otimes n}\ket{0}=\ket{0}$, we necessarily have that $\tau_n^{(4)}\in A_{4,n}$, which corresponds to
\begin{equation}
    A_{4,n}=\Alg\langle \widetilde Q_{rs}: 1\le r<s\le 4\rangle\,.
\end{equation}
As explained in the main text, the algebra $A_{4,n}$ is the full commutant of the $t=4$ matchgate action on $\HC_+^{\otimes 4}$, and the decomposition of $\mathcal H_+^{\otimes 4}$ under the commuting actions of $G$ and $A_{4,n}$ has the form
\begin{equation}
\label{ap-eq:joint_decomp_final}
    \mathcal H_+^{\otimes 4}
    \cong
    \bigoplus_{\lambda=(j_+,j_-)}
    V_\lambda^{\Spin(2n)} \otimes W_\lambda^{\mathfrak{so}(4)}\,.
\end{equation}
Above, we denoted $V_\lambda^{\Spin(2n)}$ the irreps of $\Spin(2n)$, and by $W_\lambda^{\mathfrak{so}(4)}$ those of the action of $A_{4,n}$. The irrep label $\lambda$ can be identified with a pair of total spins $\lambda=(j_+,j_-)$, with $0\leq j_++j_-\leq n$, because of the isomorphism $\so(4)\cong \su(2)\oplus\su(2)$. We also  recall that the matrix units forming a basis of $W_\lambda^{\mathfrak{so}(4)}$ read
\begin{equation}
\label{eq:X_basis_factorization_final}
    X^\lambda_{m_+,m_-;\,m_+',m_-'}
    =
    I_{V_\lambda}\otimes
    \ket{m_+,m_-}\!\bra{m_+',m_-'}\,,
\end{equation}
for $(m_+,m_-)$ the magnetic numbers of the spins $\lambda=(j_+,j_-)$, and $-j_\pm\leq m_\pm\leq j_\pm$.

We now show that the four copies of the all-zero state lie in the irrep $\lambda=(0,0)$, which is intuitive.
Indeed, it suffices to notice that for any generator $\widetilde Q_{ij}$, one has
\begin{align}
    \widetilde Q_{ij}\ket{0}^{\otimes 4} &= \frac{1}{2} \sum_{\mu=1}^{2n} \widetilde{c}_\mu^{(i)} \otimes \widetilde{c}_\mu^{(j)} \ket{0}^{\otimes 4} \\
    &=\frac{1}{2} \sum_{\mu=1}^{2n} \Gamma_i c_\mu^{(i)} \otimes c_\mu^{(j)}\ket{0}^{\otimes 4} \\
    &= -\frac{1}{2} \sum_{\mu=1}^{2n} c_\mu^{(i)} \otimes c_\mu^{(j)}\ket{0}^{\otimes 4} \\
    &= -\frac{1}{2} \sum_{k=1}^{n} \left(c_{2k-1}^{(i)} \otimes c_{2k-1}^{(j)} + c_{2k}^{(i)} \otimes c_{2k}^{(j)}\right)\ket{0}^{\otimes 4}\,,
\end{align}
where we used the explicit definition of the dressed Majorana operators, the fact that the parity operators anticommute with any Majorana, and the fact that the parity operators appearing in the copies between $i$ and $j$ act as the identity over the all-zero state.
Now, recall that, in the usual Jordan-Wigner mapping, the odd Majoranas $c_{2k-1}$ are Pauli strings with an $X$ at position $k$ and $Z$s on the left, while the even ones are the same under the replacement $X\leftrightarrow Y$. Since the $Z$ part of those strings acts trivially on the all-zero state, we just need the actions $X\ket{0}=\ket{1}$ and $Y\ket{0}=i\ket{1}$. With this one immediately sees
\begin{equation}
    \left(c_{2k-1}^{(i)} \otimes c_{2k-1}^{(j)} + c_{2k}^{(i)} \otimes c_{2k}^{(j)}\right)\ket{0}^{\otimes 4}=\ket{0\dots1_k^{(i)}\dots 0\dots1_k^{(j)}\dots 0} + (i)^2\ket{0\dots1_k^{(i)}\dots 0\dots1_k^{(j)}\dots 0} = 0\,, 
\end{equation}
where $\ket{0\dots1_k^{(i)}\dots 0\dots1_k^{(j)}\dots 0}$ denotes the state with two excitations at position $k$ in the copies $i$ and $j$. Hence
\begin{equation}
    \widetilde Q_{ij}\ket{0}^{\otimes 4}=0\,.
\end{equation}
In turn this implies that the all-zero state belongs to the irrep $\lambda=(0,0)$, since it's annihilated by the Casimir operators in Eq.~\eqref{ap-eq:t4-casimirs-def} that are built from the $\widetilde Q_{ij}$.

Now, the trivial irrep $W_{(0,0)}$ is one-dimensional, as zero-spin particles are scalars, and this implies
\begin{equation}
    \tau_n^{(4)}=\frac{\id_{V_{(0,0)}}}{|V_{(0,0)}|}\otimes \ketbra{00}{00}=\frac{X^{(0,0)}_{0,0;0,0}}{|V_{(0,0)}|}\,,
\end{equation}
for $\id_{V_{(0,0)}}$ the operator equal to the identity on the irrep $V_{(0,0)}$ of $\Spin(2n)$, $|V_{(0,0)}|$ its dimension, and $\ketbra{00}{00}$ the projector onto the trivial irrep $W_{(0,0)}$. The normalization factor appears since twirling a state preserves its trace. It's easy to see that this is a general fact, and not something special about $t=4$. Indeed, the only property we made use of is the fact that the all-zero state is annihilated by the generators $\widetilde Q_{ij}$. Namely, the twirl of the all-zero state, and hence of any Gaussian state, is proportional to the normalized projector onto the $\Spin(2n)$ irrep dual to the trivial $\so(t)$ one.
Hence, we are left with computing the ratio
\begin{equation}
\label{ap-eq:average_active_stab_as_ratio}
    \mathbb{E}_{U\sim\Spin(2n)}S_4(U\ket{0})=\frac{\Tr[X^{(0,0)}_{0,0;0,0}Q_n]}{4^n|V_{(0,0)}|}\,.
\end{equation}
Now, the normalization factor $|V_{(0,0)}|$ is the multiplicity of the trivial copy irrep $W_{(0,0)}$ inside the even-parity sector $\HC_+^{\otimes 4}$. Let us denote this multiplicity by $m_{(0,0)}^+$, so that
\begin{equation}
    |V_{(0,0)}|=m_{(0,0)}^+.
\end{equation}
To compute it, we temporarily enlarge from $\HC_+^{\otimes 4}$ to the full spinor space $\HC^{\otimes 4}$, because the copy action integrates naturally to $\Spin(4)$ only on the latter. Let $m_{(0,0)}^-$ be the multiplicity of the trivial copy irrep inside $\HC_-^{\otimes 4}$. Then the total multiplicity of the trivial copy irrep in $\HC^{\otimes 4}$ is
\begin{equation}
    m_{(0,0)}^+ + m_{(0,0)}^-\,.
\end{equation}
These two multiplicities are equal. Indeed, the operator
\begin{equation}
    M=\widetilde c_1^{(1)}\widetilde c_1^{(2)}\widetilde c_1^{(3)}\widetilde c_1^{(4)}
\end{equation}
flips the fermion parity in each copy, hence maps $\HC_+^{\otimes 4}$ onto $\HC_-^{\otimes 4}$. Moreover, $M$ commutes with the copy algebra, since each generator $\widetilde Q_{ij}$ contains one Majorana in copy $i$ and one in copy $j$, and therefore acquires two minus signs when moved through $M$. Thus $M$ intertwines the copy action on the two parity sectors, and in particular maps trivial copy subrepresentations to trivial copy subrepresentations. Therefore
\begin{equation}
    m_{(0,0)}^+=m_{(0,0)}^-.
\end{equation}
It follows that the full multiplicity is exactly twice the one we want
\begin{equation}
    m_{(0,0)}^+ + m_{(0,0)}^- = 2m_{(0,0)}^+ = 2|V_{(0,0)}|\,.
\end{equation}
Hence, by character orthogonality,
\begin{equation}
    2|V_{(0,0)}|
    =
    \int_{\Spin(4)}\chi_{\HC^{\otimes4}}(g)\overline{\chi_{(0,0)}}(g)\,d\mu(g)
    =
    \int_{\Spin(4)}\chi_{\HC^{\otimes4}}(g)\,d\mu(g)\,,
\end{equation}
where $d\mu$ is the normalized Haar measure and $\chi_{(0,0)}(g)=1$, and we used that the character of the trivial irrep is $\chi_{(0,0)}(g)=1$. 
Hence, the promotion to the full spinor representation is harmless, but allows us to use the potent tool of Weyl integration~\cite{sepanski2007compact}.
Now, from the point of view of the copy group action $\Spin(4)$, each original Majorana index $\mu\in[2n]$ is acted upon in the same way, since the generators $\widetilde Q_{ij}$ are uniform sums over those indices. This means that, just as the original $\Spin(2n)$ group acts diagonally on the copies, and $\Spin(4)$ action acts diagonally on the Majorana modes, hence by the representation $U^{\otimes 2n}$, for $U$ the standard full-spinor representation $(1/2,0)+(0,1/2)$ of $\Spin(4)$~\cite{lawson2016spin}. Since the latter is isomorphic to the sum of two standard representations of $\SU(2)$, its character is just the sum of two standard $\SU(2)$ characters. With this in mind, we can proceed along the very same lines of Section~\ref{ap:PP-dim}, since all the ingredients are the same. Particularly, denoting the two $\SU(2)$ as $\SU(2)_+$ and $\SU(2)_-$, with associated torus coordinates $(\theta_+,\theta_-)$ we arrive at
\begin{equation}
\label{ap-eq:full_denominator_integral}
    2|V_{(0,0)}|
    =
    \frac{4}{\pi^2}
    \int_0^\pi d\theta_+
    \int_0^\pi d\theta_-\,
    \left(2\cos\theta_+ + 2\cos\theta_-\right)^{2n}
    \sin^2\theta_+\sin^2\theta_-\,.
\end{equation}
Expanding the power in the integrand, we can notice how odd powers integrate to zero by symmetry, and only even powers survive. Thus we can write
\begin{equation}
    \int_{\Spin(4)}\chi_{\HC^{\otimes4}}(g)d\mu(g)
    =
    \sum_{k=0}^{n}\binom{2n}{2k} I_k I_{n-k},
\end{equation}
where
\begin{equation}
    I_m=\frac{2}{\pi}\int_0^\pi (2\cos\theta)^{2m}\sin^2\theta\,d\theta\,.
\end{equation}
This is a known integral, and evaluates to a Beta function, thus
\begin{align}
    I_m
    &=
    \frac{2^{2m+1}}{\pi}\int_0^\pi \cos^{2m}\theta\,\sin^2\theta\,d\theta \\
    &=
    \frac{2^{2m+2}}{\pi}\int_0^{\pi/2} \cos^{2m}\theta\,\sin^2\theta\,d\theta \\
    &=
    \frac{2^{2m+1}}{\pi}
    B\!\left(m+\frac{1}{2},\frac{3}{2}\right) \\
    &=
    \frac{2^{2m+1}}{\pi}
    \frac{\Gamma(m+\frac{1}{2})\Gamma(\frac{3}{2})}{\Gamma(m+2)} \\
    &=
    \frac{2^{2m}}{\sqrt{\pi}}
    \frac{\Gamma(m+\frac{1}{2})}{\Gamma(m+2)} \\
    &=
    \frac{1}{m+1}\binom{2m}{m}
    =
    \Cat_m\,.
\end{align}
This readily imples
\begin{equation}
    2|V_{(0,0)}|
    =
    \sum_{k=0}^{n}\binom{2n}{2k}\Cat_k\Cat_{n-k}\,.
\end{equation}
Conveniently, we can recognize in the last sum the Shapiro convolution for Catalan numbers~\cite{andrews2011shapiro}, which gives
\begin{equation}
\label{ap-eq:full_denominator_final}
    |V_{(0,0)}|
    =
    \frac{1}{2}\Cat_n\Cat_{n+1}\,.
\end{equation}
Lastly, we compute the overlap $\Tr[X^{(0,0)}_{0,0;0,0}Q_n]$ in the same fashion. As before, to use Weyl integration we temporarily work in the full spinor space $\HC^{\otimes 4}$, where the copies $\Spin(4)$ action is naturally defined. This again doubles the contribution of interest, one copy comes from the trivial copies sector inside $\HC_+^{\otimes 4}$, which is the one selected by $\ket{0}^{\otimes 4}$, and an identical copy comes from $\HC_-^{\otimes 4}$. The same intertwiner $M$ defined previously exchanges the two sectors and commutes with both the copy algebra and $Q_n$, as can be easily checked, so the two contributions are equal. Hence the full-spinor character integral computes twice the desired overlap.
Now, observe that $Q_n$ factorizes over the Majorana indices. Indeed, for each qubit index $k\in[n]$ one has
\begin{equation}
    X_k^{\otimes 4}=\widetilde c_{2k-1}^{\otimes 4},
    \qquad
    Y_k^{\otimes 4}=\widetilde c_{2k}^{\otimes 4},
    \qquad
    Z_k^{\otimes 4}=\widetilde c_{2k-1}^{\otimes 4}\widetilde c_{2k}^{\otimes 4}\,,
\end{equation}
because the Jordan-Wigner Z strings overlap four times and therefore cancel out. As a consequence,
\begin{align}
    Q_n
    &=
    \prod_{k=1}^{n}\left(I+X_k^{\otimes 4}+Y_k^{\otimes 4}+Z_k^{\otimes 4}\right) \\
    &=
    \prod_{k=1}^{n}\left(I+c_{2k-1}^{\otimes 4}\right)\left(I+c_{2k}^{\otimes 4}\right) \\
    &=
    \prod_{\mu=1}^{2n}\left(I+\Gamma_\mu\right)\,,
\end{align}
where we introduced, in complete analogy with the physical indices, the copy parity operator
\begin{equation}
    \Gamma_\mu=c_\mu^{(1)}c_\mu^{(2)}c_\mu^{(3)}c_\mu^{(4)}\,.
\end{equation}
Since $\Gamma_\mu^2=I$, the operator
\begin{equation}
    p_\mu=\frac{I+\Gamma_\mu}{2}
\end{equation}
is a projector. Finally then, we can write the normalized sum of fourth tensor powers of Paulis as the projector
\begin{equation}
\label{ap-eq:Qn_factorized_projector}
    \frac{Q_n}{4^n}
    =
    \prod_{\mu=1}^{2n} p_\mu\,.
\end{equation}
Now consider a fixed Majorana index $\mu$. The projector $p_\mu$ commutes with the local copy of $\so(4)$, just as the physical parity operators commute with the $\Spin(2n)$ action, since $\Gamma_\mu$ is the product of all four local Majoranas and every local bilinear contains exactly two of them, so moving the bilinear through $\Gamma_\mu$ yields two minus signs that cancel each other. Moreover, the image of $p_\mu$ has dimension two, since $\Gamma_\mu$ has eigenvalues $\pm 1$ and zero trace on the local four-dimensional copies space.
Now, a two-dimensional irreducible representation of
$
\Spin(4)\cong \SU(2)\times\SU(2)
$
must be one of the two fundamental irreps. Conveniently, we do not need to resolve in which one we land, because the resulting integral is symmetric under the exchange $\theta_+\leftrightarrow\theta_-$. So, without loss of generality, we may say that the local image of $p_\mu$ has character
\begin{equation}
    2\cos\theta_+\,.
\end{equation}
Since the $2n$ labels are independent, in the full spinor space the image of $Q_n/4^n$ has then character
\begin{equation}
    \left(2\cos\theta_+\right)^{2n}\,.
\end{equation}
Therefore
\begin{align}
    \frac{2}{4^n}\Tr[\id_{V_{(0,0)}}\otimes \ketbra{00}{00}Q_n]
    &=
    \frac{4}{\pi^2}
    \int_0^\pi d\theta_+
    \int_0^\pi d\theta_-\,
    \left(2\cos\theta_+\right)^{2n}
    \sin^2\theta_+\sin^2\theta_-\\
    &=
    \left(\frac{2}{\pi}\int_0^\pi (2\cos\theta_+)^{2n}\sin^2\theta_+\,d\theta_+\right)
    \left(\frac{2}{\pi}\int_0^\pi \sin^2\theta_-\,d\theta_-\right)\\
    &= I_n = \Cat_n\,.
\end{align}
That is
\begin{equation}
\label{ap-eq:numerator_final}
    \frac{1}{4^n}\Tr[\id_{V_{(0,0)}}\otimes \ketbra{00}{00}Q_n] = \frac{1}{2}\Cat_n\,.
\end{equation}
Finally, substituting Eqs.~\eqref{ap-eq:full_denominator_final} and \eqref{ap-eq:numerator_final} into Eq.~\eqref{ap-eq:average_active_stab_as_ratio} we obtain
\begin{equation}
    \mathbb{E}_{U\sim\Spin(2n)}S_4(U\ket{0})
    =
    \frac{\frac{1}{2}\Cat_n}{\frac{1}{2}\Cat_n\Cat_{n+1}}
    =
    \frac{1}{\Cat_{n+1}},
\end{equation}
which is exactly the claimed result in Eq.~\eqref{ap-eq:gaussian_stab}.

\end{document}